\documentclass[12pt]{article}

\usepackage{fullpage}
\usepackage[margin=1.1in]{geometry}
\usepackage{setspace}
\usepackage{amsfonts,amsmath,amssymb,amsthm,graphicx,bm}
\usepackage{mathtools}
\usepackage{mathrsfs}
\usepackage{ifthen}
\usepackage{changepage}
\usepackage{enumerate}
\usepackage{scalerel}
\usepackage{accents}
\usepackage{enumitem}
\usepackage{float}      % for [H]
\usepackage[utf8]{inputenc} % allow utf-8 input
\usepackage[T1]{fontenc} 
\usepackage{csquotes}
\usepackage{comment}
\usepackage{bbm}
\usepackage{array}
\usepackage{makecell}
\usepackage{booktabs}
\usepackage{xfrac}
\usepackage{thm-restate}
\usepackage{tikz}
\usepackage{pgfplots}
\usepackage{subcaption}
\usepackage{graphicx}
\usepackage{multirow}
\usetikzlibrary{patterns}
\usepgfplotslibrary{fillbetween}
\usetikzlibrary{intersections}
\usepackage{wrapfig}
\usetikzlibrary{positioning,fit,calc,arrows.meta}
\usepackage{fix-cm}
\usepackage[ruled,vlined,linesnumbered]{algorithm2e}
\allowdisplaybreaks
\usepackage{xcolor}
\usepackage{natbib}
\bibpunct{(}{)}{;}{a}{,}{,}

\allowdisplaybreaks
\theoremstyle{plain}
\newtheorem{theorem}{Theorem}[section]
\newtheorem{lemma}{Lemma}[section]
\newtheorem{claim}{Claim}[section]
\newtheorem{proposition}{Proposition}[section]
\newtheorem{corollary}{Corollary}[section]

\theoremstyle{plain}
\newtheorem{definition}{Definition}[section] % definition numbers are dependent on theorem numbers
\newtheorem{example}{Example}[section]

\newtheorem{conjecture}[definition]{Conjecture}

\theoremstyle{plain}


\usepackage[colorlinks=true, linkcolor=blue!70!black, citecolor=blue!70!black, urlcolor=blue, breaklinks=true]{hyperref}
\usepackage{cleveref}
\Crefname{algocf}{Algorithm}{Algorithms}
\crefname{claim}{claim}{claims}

\usepackage[suppress]{color-edits}
\addauthor{wt}{purple}
\addauthor{tao}{orange}

\newcommand{\yc}[1]{} %{{\color{teal}[Yiling: #1]}}
\newcommand{\jamie}[1]{}% {{\color{green}%[Jamie: #1]}}

\newcommand{\xhdr}{\paragraph}

\newcommand{\reals}{\mathbb{R}}

\newcommand{\eps}{\varepsilon}

\DeclareMathOperator*{\argmax}{argmax}

\newcommand{\state}{\theta}

\newcommand{\stateCDF}{F}
\newcommand{\statePDF}{f}
\newcommand{\signalProb}{\pi}
\newcommand{\signalspace}{\mathcal S}
\newcommand{\signal}{s}
\newcommand{\cvxsignaling}{\signalProb^{\cc{convex}}}

\newcommand{\OPT}{\cc{OPT}}

\newcommand{\Part}{\cc{Part}}

\newcommand{\instance}{\mathcal{I}}
\newcommand{\instanceClass}{\mathcal{C}}
\newcommand{\MPC}{\cc{MPC}}

\newcommand{\PoE}{\cc{PoE}}

\newcommand{\myopt}[2][]{\cc{OPT}\ifthenelse{\not\equal{}{#1}}{_{#1}}{}\!\left({\def\givenn{\middle|}#2}\right)}
\newcommand{\myoptPart}[2][]{\cc{OPT}^{\cc{part}}\ifthenelse{\not\equal{}{#1}}{_{#1}}{}\!\left({\def\givenn{\middle|}#2}\right)}
\newcommand{\myPoE}[2][]{\cc{PoE}\ifthenelse{\not\equal{}{#1}}{_{#1}}{}\!\left({\def\givenn{\middle|}#2}\right)}
\newcommand{\mpc}[2][]{\cc{MPC}\ifthenelse{\not\equal{}{#1}}{_{#1}}{} ({\def\givenn{\middle|}#2})}

\newcommand{\posteriorMean}{\mu}
\newcommand{\interimU}{u}
\newcommand{\designerExpU}{U}
\newcommand{\partionSet}{A}
\newcommand{\numSignals}{K}
\newcommand{\stateDim}{d}

\newcommand{\rr}{\mathbb R}
\newcommand{\zz}{\mathbb Z}

\newcommand{\abs}[1]{\left|{#1}\right|}

\newcommand{\genseq}[3]{{#1}_1 {#3} {#1}_2 {#3} \dots {#3} {#1}_{#2}}
\newcommand{\seq}[2]{\genseq{#1}{#2}{,}}

\newcommand{\ipncm}[3]{\begin{figure}[H]\begin{center}\includegraphics[scale = {#1}]{Media/#2.pdf}\caption{#3}\end{center}\end{figure}}

\newcommand{\prob}[2][]{\mathrm{Pr}\ifthenelse{\not\equal{}{#1}}{_{#1}}{}\!\left[{\def\givenn{\middle|}#2}\right]}
\newcommand{\expect}[2][]{\mathbb{E}\ifthenelse{\not\equal{}{#1}}{_{#1}}{}\!\left[{\def\givenn{\middle|}#2}\right]}

\newcommand{\sprob}[2][]{\text{Pr}\ifthenelse{\not\equal{}{#1}}{_{#1}}{}[#2]}
\newcommand{\sexpect}[2][]{\mathbb{E}\ifthenelse{\not\equal{}{#1}}{_{#1}}{}[#2]}

\newcommand{\suchthat}{\,:\,}

\newcommand{\dd}{{\mathrm d}}

\newcommand{\indicator}[1]{{\mathbbm{1}\left\{ #1 \right\}}}

\newcommand{\cc}[1]{\ensuremath{\mathsf{#1}}} % algorithm class

\newcommand{\R}{\mathbb{R}}
\newcommand{\N}{\mathbb{N}}

\newcommand{\Lip}{L}

\newcommand{\vol}{\operatorname{Vol}}
\newcommand{\smid}{\,|\,}

\title{% Classifier Design \\
Explainable Information Design\footnote{A one-page abstract of this work appeared on the ACM Conference on Economics and Computation, 2026, under the title of ``The Price and Complexity of Explainable Information Design''.}
}

\author{
Yiling Chen\thanks{Harvard University, \texttt{yiling@seas.harvard.edu}}
\qquad
Tao Lin\thanks{Microsoft Research \& The Chinese University of Hong Kong, Shenzhen, \texttt{lintao@cuhk.edu.cn}}
\qquad
Wei Tang\thanks{Chinese University of Hong Kong, \texttt{weitang@cuhk.edu.hk}}
\qquad 
Jamie Tucker-Foltz\thanks{Yale University, \texttt{j.tuckerfoltz@yale.edu}}
}

\begin{document}

\maketitle

\begin{abstract}

Optimal signaling schemes in information design (Bayesian persuasion) often involve randomization or disconnected partitions of state space, which might be too intricate to be audited or communicated. 
We propose \emph{explainable information design} in the context of linear information design with a continuous state space.
In the case of single-dimensional state, we restrict the information designer to use interval-partitional signaling schemes defined by deterministic and monotone partitions of the state space, where a unique signal is sent for all states in each part. 
We prove that the price of explainability (PoE) -- the ratio between the performances of the optimal explainable signaling scheme and unrestricted signaling scheme -- is exactly $\sfrac{1}{2}$ in the worst case, meaning that partitional signaling schemes are never worse than arbitrary signaling schemes by a factor of $2$.
For a uniform prior, this PoE can be improved to a tight $\sfrac{2}{3}$. 
We then extend the analysis to multi-dimensional state spaces by studying two notions of explainability: convex-partitional policies and axis-aligned rectangular policies. 
We prove a tight PoE of $\sfrac{1}{(m+1)}$ for convex-partitional policies, while for rectangular policies we establish a PoE guarantee under uniform prior that is independent of the number of signals but unavoidably
exponential in the dimension $m$.
We also study the computational complexity of explainable information design, 
proving that the exactly optimal explainable policy is NP-hard to compute, but an explainable policy with $\sfrac{1}{2}$ approximation guarantee can be computed in polynomial time for piecewise Lipschitz utility functions.

\end{abstract}
\thispagestyle{empty}

\newpage 

\newpage
\setcounter{page}{1}

\section{Introduction}
Information design, or Bayesian persuasion \citep{kamenica_bayesian_2011}, studies how an informed designer can strategically provide information to influence the behavior of a decision-maker %for desired decisions
(see, e.g., the surveys by \citealp{D-17,K-19,BM-19}).
The information designer commits to an information policy (also called a signaling scheme) that maps a state of the world to a possibly randomized signal; the decision-maker, assumed to fully understand the information policy, receives the signal and takes an optimal action according to her posterior belief. 
The information designer aims to find an information policy to maximize his expected utility. 
This framework has inspired an active line of research in economics, operations research, and theoretical computer science, with various real-world applications like online ad auctions (\citealp{EFGPT-14,BBX-18}), price discrimination (\citealp{BMSW-24}), 
contract design (\citealp{BTXZ-24,CC-25}), recommender systems (\citealp{zu_learning_2021, feng_online_2022, hossain_persuasive_2024}), and voting (\citealp{AC-16a, AC-16b}).

The structures of the \emph{optimal} information policies in various information design problems, however, are often opaque. 
In the setting with discrete states and discrete actions, the optimal information policy is the solution to a linear program that may involve randomization in the following form: ``\emph{if the quality of a product is \emph{good} (state 1), always recommend \emph{buy} (signal 1) to the consumer (decision-maker); if the quality of the product is \emph{bad} (state 2), recommend \emph{buy} (signal 1) with probability $43\%$ and \emph{not buy} (signal 2) with probability $57\%$}'' (adopted from \citealp{kamenica_bayesian_2011, dughmi_algorithmic_2016}).
It is unclear how an information designer can credibly commit to such a randomized policy in practice.\footnote{This commitment issue inspires a literature on credible persuasion \citep{lin_credible_2024}.}
Even in settings where the optimal information policy is deterministic, such a policy may still lack explainability. Consider the following example: 

\begin{example} \label{ex:intro-example}
{\rm 
A school is designing a grading scheme that maps a student's GPA, a number in $[0, 100]$, to a letter grade (e.g., A, B, C).  The student's GPA is the state of the world, assumed to be uniformly distributed over $[0, 100]$. There are two hiring companies, 1 and 2, where company 1 is considered better.  A student will be hired by company 1 if the company believes that the student's GPA is in $[80, 100]$; the student will be hired by company 2 if the company believes that the student's GPA is in $[40, 80)$.  The school gets a payoff of $1$ if the student is hired by company 1, payoff of $0.9$ if the student is hired by company 2, payoff of $0$ otherwise.
If no information about the student is revealed, then both companies believe that all students have average GPA $50$, so all students will be hired by company 2 and the school gets payoff $0.9$.

The optimal grading scheme that the school can design turns out to be the following: map GPA in $[0, 2.5] \cup [77.5, 100]$ to grade ``A''; map GPA in $(2.5, 77.5)$ to grade ``B''.  Under this grading scheme, the expected GPA of a student with grade ``A'' is exactly 80, so the student will be hired by company 1; the expected GPA of a student with grade ``B'' is exactly 40, so the student will be hired by company 2.  The first case happens with probability $0.25$ and the second case happens with probability $0.75$.  The school obtains expected payoff $0.925$, larger than the $0.9$ payoff above. 
}
\end{example}

The optimal information policy (grading scheme) illustrated above exhibits a non-monotone structure, mapping two disconnected subsets of states into a single signal.  The non-monotonicity of optimal information policy is known to be unavoidable in information design problems with a continuous state space \citep{arieli_optimal_2023}, but such non-monotonicity lacks explainability and is often undesirable in practice due to various considerations, \taoedit{such as transparency, legal regulations, and cognitive limitations of decision-makers to understand complicated signals \citep{hoong_calibrated_2025}.}
We do not expect a school to implement a non-monotone grading scheme. 

\subsection{Our Contributions}
\label{sec:our-contribution}
Motivated by the non-explainability of optimal information policies (illustrated in \Cref{ex:intro-example}), our work studies an \emph{explainable information design} problem.
We consider an information design environment with
% a single-dimensional continuous state space, normalized to $[0, 1]$.
an $m$-dimensional continuous state space $[0, 1]^m$.
The information designer's payoff, captured by function $\interimU: [0, 1]^m \to [0, 1]$, depends on the posterior mean of the state induced by the signal. % from the information policy.
This class of information problems is %often
referred to as \emph{linear information design} and % has been widely studied in %both theoretical computer science and economics
the $(m=1)$-dimensional case has been widely studied by previous works
% in the literature
(e.g., \citealp{gentzkow_rothschild-stiglitz_2016,DM-19,KMS-21,arieli_optimal_2023,BMSW-24,FHT-24,BMSW-25,KLZ-25, zhao_2025_public}), but the multi-dimensional case is less studied. 
In contrast to prior works that focus on characterizing the optimal -- but intricate -- information policies, we seek to design explainable policies and quantify the %associated
\emph{cost} to achieve such explainability.

% \xhdr{\bf Overview of results.}
\textbf{Overview of Results.}
We first and primarily study the single-dimensional case where $m=1$.
In the single-dimensional case, we posit that an information policy is explainable if it has a deterministic, monotone, partitional structure: it divides the state space into $K$ consecutive intervals $\{(a_{i-1}, a_i)\}_{i=1}^K$ and deterministically %sends a unique signal $s_i$ for all states in each interval $(a_{i-1}, a_i)$.
reveals which interval the state belongs to. 
We call such information policies \emph{$\numSignals$-(interval-)partitional information policies}. 
Focusing on explainable/partitional information policies, we study two questions:
\begin{itemize} 
\item[(1)] The first question is the \emph{price of explainability (PoE)}: what is the approximation ratio between the optimal explainable information policy and the optimal unrestricted (non-explainable) policy?  We show that the price of explainability is exactly $1/2$ in the worst case (\Cref{thm:K-partitional-1/2}).  This means that, by using an explainable information policy, the information designer can always guarantee $1/2$ of the expected payoff obtained from the optimal unrestricted policy, and this is the best universal guarantee one can hope for.
We then consider an interesting special case where the prior distribuiton over the state is uniform, in which case we prove that the PoE can be improved to $2/3$ and is tight (\Cref{sec:uniform}).

% describe a natural restricted class of instances in which we show (\Cref{lem13}) the worst case price of explainability is $2/3$. 

\item[(2)]  The second question is the \emph{computational complexity} of explainable information design: how difficult is it to compute an explainable information policy with a high payoff for the designer? 
% then study the computational complexity of finding the optimal explainable information policy.
We show that, unfortunately, the \emph{optimal} explainable information policy is NP-hard to compute in general (\Cref{thm:optimal-NP-hard}). 
% Despite this hardness, we provide a positive result. 
Nevertheless, because an optimal explainable information policy cannot be better than $1/2$ of the optimal unrestricted policy in the worst case (our Result (1)), an alternative goal is to find a (not necessarily optimal) explainable information policy that achieves this $1/2$ approximation guarantee.
We show that such a $1/2$-optimal explainable information policy can be computed efficiently by an algorithm with a polynomial running time, for any prior distribution and any piecewise Lipschitz utility functions (\Cref{thm:1/2-piecewise-Lipschitz}).

%   \begin{itemize}
%     % \item[(a)] When the information designer's utility function $\interimU$ is Lipschitz continuous, we construct a fully polynomial time approximation scheme (FPTAS) that can find an $\eps$-approximately optimal explainable policy with a running time that is polynomial in $\frac{1}{\eps}$, the Lipschitz constant of $\interimU$, and the number of signals $\numSignals$ (i.e., the number of intervals that the state space is divided into by the partitional information policy). This is formalized in \Cref{thm:nearly-optimal-Lipschitz}.
%     \item [(b)] Since an explainable information policy provides a $1/2$ approximation to the optimal unrestricted information policy in the worst case (our Result (1)), an alternative goal is to find a (not necessarily optimal) explainable information policy that achieves this $1/2$ approximation guarantee.
%     % We show that this goal is achievable for piecewise constant utility functions with a polynomial running time in the number of pieces (\Cref{thm:1/2-piecewise-constant}). We also show that this goal can be achieved approximately for piecewise Lipschitz utility functions with a polynomial running time in the number of pieces, the Lipschitz constant, and approximation error (\Cref{cor:1/2-piecewise-Lipschitz}). 
%     We achieve this goal for piecewise Lipschitz utility functions by designing an algorithm with polynomial running time in the number of pieces and the Lipschitz constant %, and desired approximation error
%     (\Cref{thm:1/2-piecewise-Lipschitz}).
%   \end{itemize}
%     % \item[(3)]
\end{itemize}

We also generalize our price of explainability results to the multi-dimensional case, where we consider two notions of ``explainable information policy''.  The first is convex-partitional policies, which divide the state space $[0, 1]^m$ into non-overlapping convex subsets and send one signal for each subset.  We prove that the PoE of those policies is exactly $1/(m+1)$, which directly generalizes our single-dimensional $1/2$ PoE result (see \Cref{thm:convex-1/(m+1)}).  
The second, inspired by the \emph{explainable clustering} literature, is information policies that divide the state space into axis-aligned rectangles, which we prove to have a PoE bound that is independent of the number of rectangles $K$ (see \Cref{thmUniformHigherDimension}) but unavoidably exponential in the dimension $m$ (see \Cref{thm:high-dimension-negative}).

\textbf{Overview of Techniques.}
Our Result (1) about the $1/2$ PoE of $\numSignals$-interval-partitional policies in the single-dimensional case, and the generalization to $1/(m+1)$ PoE of $\numSignals$-convex-partitional policies in the multi-dimensional case, share a common proof, which we overview here.  At a high level, we propose a procedure to convert any \emph{extremal} information policy $\pi^*$, which is an extreme point of the set of all information policies with at most $K$ signals, to an explainable information policy $\pi$ with at most $K$ signals that can achieve at least $1/(m+1)$ fraction of the expected utility of $\pi^*$.
We prove the existence of an optimal information policy that is extremal; this fact is not immediate because the designer's expected utility is not a linear function of the information policy, and the set of all information policies with $K$ signals is an infinite-dimensional non-convex set.  By applying the conversion procedure to the extremal optimal policy $\pi^*$, we obtain an explainable policy with $1/(m+1)$ approximation guarantee, which proves $\PoE \ge 1/(m+1)$. We prove the other direction $\PoE \le 1/(m+1)$ by constructing non-trivial worst-case examples. 

Our conversion procedure from extremal policies to explainable policies make use of some characterizations of extremal information policies from previous works \citep{arieli_optimal_2023, kleiner_extreme_2025}.
Those works first represent an information policy (with potentially infinitely many signals) as a mean-preserving contraction (MPC), which is the distribution over posterior means of the state induced by the policy. Under this representation, the designer’s expected utility becomes a linear function of the MPC, %over the set of MPCs,
implying the existence of an optimal solution among the extreme points of MPCs.
They then show that the extreme points of MPCs have a \emph{$(m+1)$-pooling} structure: the information policy divides the state space $[0, 1]^m$ into convex subsets, and sends at most $m+1$ signals for the states in each subset.

Previous characterizations, however, do not apply to our setting directly because we focus on information policies with a finite number of signals (i.e., $K$-partitional policies have $K$ signals), whose corresponding MPCs is not even a convex set.
Nevertheless, we show that the extreme points of this non-convex MPC set also have the $(m+1)$-pooling structure. 
Then, by inducing the ``most useful'' signal among the $m+1$ signals within each convex subset that the extremal policy divides the state space into, we obtain an explainable policy with $1/(m+1)$ approximation guarantee. 
This $1/(m+1)$-approximation conversion from extremal policy to explainable policy may be of independent interest beyond the proof of this result.  In fact, we also use this conversion to obtain our algorithmic Result (2) and the price of explainability of axis-aligned rectangular information policies. 

% \wtcomment{For a non-convex set, is the ``extreme points'' legitimate definition?}

% \vspace{0.6em}

For Result (2), to prove the NP-hardness of computing the optimal $\numSignals$-interval-partitional information policy in the single-dimensional case, we reduce from the NP-hard problem \textsc{Partition} \citep{KarpOriginalReductions}.
% : given a set of integers $\{c_1, \ldots, c_n\}$, decide whether there exist signs $(b_1, \ldots, b_n) \in \{-1, +1\}^n$ such that $\sum_{i=1}^n b_i c_i = 0$. We construct an information design instance with a utility function $\interimU$ that is non-zero (in particular, $\interimU(x) = 1$) in $2n$ points in $[0, 1]$, corresponding to $+c_i$ and $-c_i$ for each $i\in\{1, \ldots, n\}$. This instance is constructed in a way that the signals from any information policy can only induce one of the two posterior means corresponding to $\pm c_i$, which means that we can only choose one sign for each $c_i$. If the optimal information policy achieves expected utility $1$, then it must induce $n$ posterior means among the $2n$ points where $\interimU$ is non-zero.  This means that we can successfully assign a sign for every $c_i$ such that $\sum_{i=1}^n b_i c_i = 0$.  If the expected utility of the optimal information policy is less than $1$, then it induces less than $n$ posterior means among the $2n$ points, so the \textsc{Partition} problem instance is a NO instance. {\color{red} Maybe move most of this paragraph to Section 4.1.}
% 
% \vspace{0.6em}
% For Result (2a), we develop the FPTAS to find an $\eps$-approximately optimal $\numSignals$-partitional information policy for Lipschitz utility functions by designing a dynamic programming algorithm on a discretized state space. 
For the positive result, we obtain a $\numSignals$-partitional information policy with $1/2$ approximation guarantee by
% to the optimal unrestricted information policy by
first computing the optimal unrestricted policy $\signalProb^*$ (via a polynomial-time algorithm from \cite{arieli_optimal_2019}) and then converting $\signalProb^*$ to a $1/2$-approximation $\numSignals$-partitional policy using the conversion procedure (\Cref{alg:conversion}) we develop in the proof of Result (1).

To establish the PoE for axis-aligned rectangular information policies in the multi-dimensional case, we develop another conversion process to convert any convex-partitional policy into an axis-aligned rectangular policy with an approximation factor that is constant in $K$ but exponential in $m$. This, combined with the $1/(m+1)$ PoE result for convex-partitional policy directly, gives the PoE for axis-aligned rectangular policies. The exponential dependency on dimension $m$ is proven to be unavoidable.

\subsection{Additional Related Work}

\xhdr{Limited signals in information design.}
% Indeed, there have been quite a few works  that use limited number of signals to reduce the complexity of the information policies.
Recent papers have explored signal-limited information design as a way to reduce the complexity of information policies. 
For example, %by limiting the signal space,
\citet{gradwohl_algorithms_2022} characterize the optimal information structure with limited number of signals
under a symmetry or an independence assumption on the distribution of utility values for the actions.
\citet{aybas2024persuasioncoarsecommunication}
study a canonical Bayesian persuasion model of \cite{kamenica_bayesian_2011} % and explore the optimal design of information where the designer is exogenously constrained to using a limited set of messages.
% Their main result is a
and obtain tight bounds on the marginal value of an additional message. % for the design.
Closer to our work, \citet{lyu2024coarseinformationdesign} focus on linear information design. %(where the designer's payoff depends only on the expectation of the state)
% and characterize the structure of the optimal information policy with a restricted number of signals. 
% However, as observed in these works, the optimal information policies still exhibit a complex structure and lack of explainability.
However, these studies consistently observe that even with a limited number of signals, the resulting optimal policies often remain structurally complex and difficult to explain.
\citet{zhao_2025_public} study the ``price of finiteness'' in linear information design problems, proving that the approximation ratio between signaling schemes with finite and infinite numbers of signals can be arbitrarily close to $0$.
In contrast, our ``price of explainability'' notion compares explainable and non-explainable signaling schemes with the same number of signals.

\xhdr{Partitional information policies.}
Our work focuses on a simple yet expressive class of information policies: (monotone) partitional policies. %, which have been studied in previous works.
%For example,
\citet{DM-19} identify necessary and sufficient conditions under which partitional policies are optimal among all possible policies. 
Roughly speaking, their conditions require that the designer's interim payoff function is {\em affine-closed}. Moving forward, \citet{KLZ-25} study the structure of optimal partitional information policies when the payoff function is M-shaped, which does not satisfy the affine-closed condition.
\citet{M-21} and \citet{OR-23} explore partitional policies in information design environments that differ from ours.
%and they study the monotone information structures in their considered information design environment.
% a setting where the designer's interim payoff function is quadratic, and \cite{OR-23} consider a general information design problem. 
% 
% Partitional information policies have some desirable properties, for example, it is a  monotone information structure  where high states can be summarized with
% Our work view the partitional information policies as explainable policies as these policies can by a (potentially small) decision tree.
For non-linear information design problems, \citet{ali_design_2025} identify conditions for when any partitional information policy can be implemented by an equilibrium of a disclosure game where the sender discloses hard evidence about the state to the receiver without committing to the disclosure policy, and when the optimal information policy can be approximated by a \emph{non-monotone} partitional policy. % (where the state space $\Theta$ is partitioned into sets $\cup_{i=1}^K \Theta_i$ each of which can contain disconnected pieces). 
% 
% To our knowledge, we are the first to formally study the performance gap between (monotone) partitional information policies and arbitrary information policies, providing a tight characterization for the price of explainability in linear information design environments.
% 
\taoedit{On the more empirical side, \citet{hoong_calibrated_2025} argue that coarse information is beneficial to biased decision-makers in AI-assisted decision-making, and find that partitional/coarse information policies achieve better performance than more informative policies with continuous signal space in experiments with, e.g., loan specialists.  
Their empirical finding supports our focus on partitional information policies, but we take a different approach: we theoretically quantify the worst-case performance gap between partitional information policies and the optimal ones in linear information design environments.}

\xhdr{Explainable clustering.}
Conceptually, our work relates to the growing literature in explainable clustering in machine learning (e.g., \citealp{MDRF-20,MS-22a,CH-22,EMN-22,GPSY-23}), 
% where an explainable clustering is one where the clusters are specified by a tree of axis-aligned threshold cuts.
where a clustering is deemed explainable if it can be specified by a decision tree consisting of axis-aligned threshold cuts.
A central focus in this line of work is the price of explainability -- typically defined as the worst-case ratio, over all input instances, between the cost of the best explainable clustering %using at most $\numSignals$ clusters
and that of the optimal (unrestricted) clustering.
% A central focus in this line of literature is characterizing the price of explainability -- with using at most $\numSignals$ clusters, the competitive ratio between the best explainable clustering and the optimal (unrestricted) clustering.

We view partitional information policies as naturally explainable, since they can be described by a %(potentially small)
decision tree, aligning with the explainability notion in clustering.
% This perspective aligns with the literature in explainable clustering, where a clustering is deemed explainable if it can be represented by a decision tree with at most $\numSignals$ leaves. 
% The objective in the clustering has a particular structural form, i.e., it mostly concerns minimizing the sum of the distances -- the squared $\ell_2$ distance or $\ell_1$ distance -- between a set of data points in $\R^\stateDim$ and their nearest center, which is significantly different from ours.
However, the objective function in clustering, such as minimizing the sum of distances %or $\ell_1$ distances
between a set of data points in $\R^\stateDim$ and their assigned centers,
%typically
has a specific geometric structure and %is significantly different from ours.
differs from ours. 
% . This differs significantly from our setting.
% Moreover, the explainable clustering mainly focus on high-dimensional settings as in one-dimensional setting, the explainable clustering problem becomes straightforward as the the optimal (unrestricted) clustering is simply the partitional clustering and thus it is also an explainable clustering.
Moreover, %the
explainable clustering %literature
primarily focuses on high-dimensional spaces. The problem becomes trivial in one dimension: the optimal unrestricted clustering is itself partitional, and hence already explainable. 
% However in our setting, the optimal information policy in one-dimensional state space is already complicated, and it is non-trivial to obtain the explainable information policy even in one-dimensional settings.
% As such, their approaches cannot be directly applied to our information design setting.
% In contrast,
But in our setting, even with one-dimensional state space, the optimal information policy can be non-trivial, and obtaining an explainable (e.g., partitional) policy remains challenging. As a result, existing approaches from explainable clustering cannot be directly applied to our problem. % information design problem.
% these works largely consider high-dimensional settings
% and  

% The explainable clustering literature mostly focus on characterizing the price of explainability of an explainable clustering compared to the optimal (unrestricted) clustering that uses at most $\numSignals$ clusters. 
% Thus, their approaches can not be applied to our setting.
% As a result, their techniques do not directly apply to our information design framework.
% \cite{MDRF-20,MS-22a, GPSY-23}

\section{Preliminaries}

\newcommand{\partpolicies}{\mathcal{\partionSet}_\numSignals}

\xhdr{The information design environment.}
A payoff-relevant state $\state \in [0, 1]$ is drawn from an absolutely continuous prior probability distribution $F$ with full support on $[0, 1]$, with CDF also denoted by $\stateCDF$ and density $\statePDF$.
We will discuss the higher dimensional case in \Cref{sec:high-dimension}.
% \footnote{We defer the results and the discussions for the discrete prior distribution in \Cref{subsec:point mass}.}
% with full support on the unit interval $[0, 1]$.  
% We assume that the CDF $\stateCDF$ admits a continuous density function $\statePDF$ with full support.
% \taocomment{We assume $F$ to be continuous instead of $f$ to be continuous}
% We assume $F$ to have a density function $f$, use $F(a) = \prob[F]{\state\le a} = \int_0^a f(x)\; \dd x$ to denote the cumulative distribution function, and 
Let $\stateCDF([a, b]) = \prob[\state\sim \stateCDF]{a\le\state\le b} = \int_a^b \statePDF(x) \, \dd x = \stateCDF(b) - \stateCDF(a)$ denote the prior probability of interval $[a, b]$.
An informed information designer %(henceforth designer)
wants to persuade an uninformed receiver to take an action that the designer prefers. The receiver's optimal action depends on her information about the state, and initially she only knows that the state is distributed according to $\stateCDF$. The designer, who knows $\stateCDF$ and observes the realized state $\state$, may reveal information about the realized state to the receiver.

In particular, the designer can commit to an information policy $(\signalspace, \signalProb)$ where $\signalspace$ is a measurable space, and $\signalProb = (\signalProb(\cdot\smid \state))_{\state\in [0, 1]}: [0, 1] \to \Delta(\signalspace)$ prescribes the conditional probability distribution over the signal space $\signalspace$ given a realized state.
% which is a randomized mapping from states to signals, where $\signalspace$ is some signal space and $\Delta(\signalspace)$ is the set of all probability distributions over $\signalspace$.
Given the designer's information policy $(\signalspace, \signalProb)$ and if the realized state is $\state$, a signal $\signal \sim \signalProb(\cdot\smid \state)$ will be realized. 
The receiver, who knows the information policy $(\signalspace, \signalProb)$,
observes the signal realization $s$ and updates her posterior beliefs about the state according to Bayes' rule.
This posterior belief determines the receiver's optimal action % and implicitly \taocomment{indirectly?}
as well as the designer's payoff.
% We denote the designer's interim utility as $\interimU: \Delta([0, 1]) \rightarrow [0, 1]$ where $\interimU(\posterior)$ is the designer's expected utility when the receiver's posterior belief is given by $\posterior\in \Delta([0, 1])$.
Throughout this work, we will focus on an important class of information design environment, \emph{linear information design}, where the designer's payoff from inducing a posterior belief depends only on the mean of the posterior belief, and is given by an upper semi-continuous function $\interimU: [0, 1] \rightarrow [0, 1]$.\footnote{
Upper semi-continuity is a consequence of the standard assumption that the receiver % chooses the designer-preferred action when indifferent;
break ties in favor of the designer; 
it ensures that the designer's optimization problem \eqref{eq:OPT-K} has a solution. 
% could be derived as a consequence of the primitives under the assumption that the receiver chooses the sender-preferred action when indifferent.
% Upper semi-continuouty is a standard assumption in the literature that assumes the designer's optimization problem has a solution.
}
We refer to $\interimU$ as the designer's interim utility function.
% This class of information design problems is also referred to as \emph{linear information design}.
% \taocomment{No need to say ``primarily'' because we only focus on this case? }
% \taocomment{Do we need to mention the word ``upper semi-continuous''? }
% \begin{definition}[Linear information design environment]
% In the linear information design environment, the designer's payoff from inducing a posterior belief depends only on the mean of this posterior belief, and 
% % (slightly overloading the notation) 
% is given by an upper semi-continuous function $\interimU: [0, 1] \rightarrow [0, 1]$.
% We refer to such function as the designer's interim utility function.
% \end{definition}
Given a realized signal $\signal \in \signalspace$, %we denote by $\posteriorMean_{\signalProb, \signal}$
the receiver's posterior mean of the state is $\posteriorMean_{\signalProb, \signal} = \frac{\int_0^1 \state \cdot \signalProb(\signal\smid \state) \statePDF(\state) \, \dd \state}{\int_0^1 \signalProb(\signal \smid \state) \statePDF(\state) \, \dd \state} \in [0, 1]$.
% \yc{The $f(x)$ in the bottom should be $f(\theta)$.}
Under this information design environment, the designer's expected utility when committing to an information policy $(\signalspace, \signalProb)$ is thus given by\footnote{When the context is clear, we often omit $\signalspace$ and directly refer to $\signalProb$ as the designer's information policy.}
\begin{equation*}
      \designerExpU(\signalProb) \;=\; \int_0^1 \statePDF(\state)\cdot
      \expect[\signal \sim \signalProb(\cdot\mid \state)]{\interimU(\mu_{\signalProb, \signal})}
      % \sum\nolimits_{\signal \in \signalspace} \signalProb(\signal\mid \state)  
      \, \dd \state~. 
\end{equation*}
%\taocomment{Since $\signalspace$ is not necessarily finite, should we use summation or integration or expectation notation for $\signal$? }
The designer aims to design an information policy that maximizes his expected payoff $\designerExpU(\signalProb)$.
%This information design problem
Linear information design has been an important focus and widely studied in % both theoretical computer science and economics
the literature %(e.g., \citealp{gentzkow_rothschild-stiglitz_2016,DM-19,KMS-21,arieli_optimal_2023,KLZ-25, zhao_2025_public}).
(see the references in \Cref{sec:our-contribution}). 
It is known that the designer's optimal information policy may need to use infinitely many signals and %the optimal information policy may
admit a complicated structure that lacks explainability (which we will discuss shortly).

% An information design instance $\instance = (\stateCDF, \interimU)$ consists of
The prior probability distribution $\stateCDF$ and the designer's interim utility function $\interimU$ define an information design instance $\instance = (\stateCDF, \interimU)$.
% When we discuss computational complexity, % of various information design objectives,
% we assume that $\stateCDF$ and $\interimU$ are encoded in the obvious way, e.g., piecewise-linear utility functions are specified by line segments with rational coordinates, etc.
Given a positive integer $\numSignals \ge 1$, we are interested in the designer's problem when he is restricted to use at most $\numSignals$ signals, namely, $|\signalspace| \le \numSignals$.
We refer to an information policy that uses at most $\numSignals$ signals as a $\numSignals$-signaling scheme. %\footnote{In this work, we interchangeably use information policy and signaling scheme.}
Let 
\begin{equation} \label{eq:OPT-K}
    \myopt[\instance]{\numSignals} \, = \, \max\nolimits_{(\signalspace, \signalProb): |\signalspace| \le\numSignals} \designerExpU(\signalProb)
\end{equation}
be the designer's optimal payoff when using $K$-signaling schemes. 
% at most $\numSignals$ signals.  
% We refer to an information policy that uses at most $\numSignals$ signals as $\numSignals$-signaling scheme, and the optimal information policy that uses at most $\numSignals$ signals as the optimal $\numSignals$-signaling scheme.\footnote{In this work, we interchangeably use information policy and signaling scheme.}

\xhdr{Non-explainability of optimal information policy.}
We provide below a simple example to illustrate the non-explainability of optimal information policy.
\begin{example}[Adapted from \citealp{arieli_optimal_2023}]
\label{example:bi-pooling}
{\rm 
Consider the instance $\instance = (\stateCDF, \interimU)$ where $\stateCDF = \mathrm{Uniform}[0, 1]$ and $\interimU$ satisfies $\interimU(\frac{1}{3}) = \interimU(\frac{2}{3}) = 1$ and $\interimU(\posteriorMean) = 0$ for $\posteriorMean \notin\{\frac{1}{3}, \frac{2}{3}\}$.  The following two information policies with $\signalspace = \{\signal_1, \signal_2\}$ are both optimal and the designer obtains maximum expected utility of $1$ (see \Cref{fig:bipooling1}):
\begin{itemize}
    \item A randomized information policy $\signalProb_1$ defined by
    % An optimal signaling scheme should send $\numSignals = 2$ signals inducing posterior means at $\frac{1}{3}$ and $\frac{2}{3}$ respectively.  One way to achieve this is to send signal $\signal_1$ with probability 
    $\signalProb_1(\signal_1 \smid \state) = 1-\state, \signalProb_1(\signal_2 \smid \state) = \state$ for $\state \in [0, 1]$.
    It is easy to see that signal $\signal_1$ induces posterior mean $\posteriorMean_{\signalProb_1, \signal_1} 
    % = \frac{\int_0^1 \state \cdot (1-\state)\,  \dd \state}{\int_0^1 (1-\state)\, \dd \state} 
    = \frac{1}{3}$ and $\signal_2$ induces posterior mean $\posteriorMean_{\signalProb_1, \signal_2} 
    % \frac{\int_0^1 \state \cdot \state\, \dd \state}{\int_0^1 \state\, \dd \state} 
    = \frac{2}{3}$.  
    % The information designer obtains optimal expected utility $U(\signalProb_1) = 1$. 
    \item 
    A deterministic information policy $\signalProb_2$ defined by $\signalProb_2(\signal_1 \smid \state) = \indicator{\state\in (\frac{1}{12}, \frac{7}{12})}$, 
    $\signalProb_2(\signal_2 \smid \state) = \indicator{\state\in [0, \frac{1}{12}] \cup [\frac{7}{12}, 1]}$ for $\state\in [0, 1]$.
    Same as the previous policy, signals $\signal_1, \signal_2$ % under $\signalProb_2$
    induce 
    posterior means $\frac{1}{3}, \frac{2}{3}$, respectively.
    % Another optimal solution is a deterministic signaling scheme $\signalProb_2$ that does the following: 
    % it partitions the state space $[0, 1]$ into three intervals $[0, \frac{1}{12}], (\frac{1}{12}, \frac{7}{12}), [\frac{7}{12}, 1]$. 
    % It sends signal $\signal_1$ if the state is in $(\frac{1}{12}, \frac{7}{12})$, sends signal $\signal_2$ if the state is in $[0, \frac{1}{12}] \cup [\frac{7}{12}, 1]$. The two signals also induce posterior means $\frac{1}{3}$ and $\frac{2}{3}$. 
\end{itemize}}
\end{example}

\begin{figure}[h] %[htbp]
  \centering
  \begin{subfigure}[b]{0.49\textwidth}
    \centering
    \resizebox{\linewidth}{!}{% \begin{tikzpicture}[>=stealth, line cap=round, thick, scale=0.9]
\begin{tikzpicture}[
  >=stealth,
  line cap=round,
  thick,
  scale=1,
  every node/.append style={font=\scriptsize}
]

  % ======== Common part begins =======================
  %--- parameters --------------------------
  \def\L{1}      % left end of support
  \def\R{6}      % right end of support
  \def\Ytop{1.7}   % y-coordinate of top axis
  \def\Ybot{0}   % y-coordinate of bottom axis
  \pgfmathsetmacro{\xone}{\L+(\R-\L)/3}
  \pgfmathsetmacro{\xtwo}{\L+2*(\R-\L)/3}
  
    %--- axes & labels ----------------------------
  \draw[->] (\L-0.4,\Ytop-0.2) -- (\R+0.5,\Ytop-0.2);     % top axis
  \node[above] at (\L-0.12,\Ytop+0.15) {0};
  \node[above] at (\R+0.12,\Ytop+0.15) {1};
  \draw[->] (\L-0.4,\Ybot) -- (\R+0.5,\Ybot);     % bottom axis
  \node[below] at (\L-0.12,\Ybot-0.03) {0};
  \node[below] at (\R+0.12,\Ybot-0.03) {1};
  \node[below] at (\R+0.55,\Ytop-0.2) {$\state$};
  \node[below] at (\R+0.55,\Ybot-0.04) {$\posteriorMean$};
  %--- dashed vertical boundaries -------------------------
  \draw[dashed] (\L,\Ybot) -- (\L,\Ytop);
  \draw[dashed] (\R,\Ybot) -- (\R,\Ytop);
  % ======== Common part ends =======================

  %--- hatched rectangle -------------------------------
  \draw[pattern=north east lines,pattern color=blue]
    (\L,\Ytop-0.2) rectangle (\R,\Ytop+0.2);
  % \path[
  %   pattern={Lines[angle=45,distance=8pt,line width=0.5pt]},
  %   pattern color=blue
  % ] (\L,\Ytop-0.2) rectangle (\R,\Ytop+0.2);
  
  \path[pattern=north west lines,pattern color=red]
    (\L,\Ytop-0.2) rectangle (\R,\Ytop+0.2);

  \draw[decorate,
        decoration={brace,mirror,amplitude=6pt,raise=4pt, aspect=2/3},
        red,thick]
    (\L,\Ytop-0.1) -- (\R,\Ytop-0.1);

    \draw[decorate,
        decoration={brace,mirror,amplitude=6pt,raise=4pt, aspect=1/3},
        blue,thick]
    (\L,\Ytop-0.3) -- (\R,\Ytop-0.3);

  %--- blue arrow + bar ------------------------------
  \fill[blue!20] (\xone-0.3,\Ybot) rectangle (\xone+0.3,\Ybot+0.8);
  \draw          (\xone-0.3,\Ybot) rectangle (\xone+0.3,\Ybot+0.8);
  \draw[->, blue] (\xone, \Ytop-0.65) -- (\xone,\Ybot+0.82); 

  %--- red arrow + bar -------------------------------
  \fill[red!20] (\xtwo-0.3,\Ybot) rectangle (\xtwo+0.3,\Ybot+0.8);
  \draw         (\xtwo-0.3,\Ybot) rectangle (\xtwo+0.3,\Ybot+0.8);
  \draw[->, red] (\xtwo,\Ytop-0.45) -- (\xtwo,\Ybot+0.82); 

  % Add points at the posterior means 1/3 and 2/3
  \node[below] at (\xone, \Ybot) {1/3};
  \node[below] at (\xtwo, \Ybot) {2/3};
  \fill (\xone, \Ybot) circle[radius=2pt];
  \fill (\xtwo, \Ybot) circle[radius=2pt];
\end{tikzpicture}}
    \caption{The (randomized) information policy $\signalProb_1$}
    \label{fig:bipool-rand}
  \end{subfigure}%
  \hfill
  \begin{subfigure}[b]{0.49\textwidth}
    \centering
    \resizebox{\linewidth}{!}{% \begin{tikzpicture}[>=stealth, line cap=round, thick, scale=0.9]
\begin{tikzpicture}[
  >=stealth,
  line cap=round,
  thick,
  scale=1,
  every node/.append style={font=\scriptsize}
]

  % ======== Common part begins =======================
  %--- parameters --------------------------
  \def\L{1}      % left end of support
  \def\R{6}      % right end of support
  \def\Ytop{1.7}   % y-coordinate of top axis
  \def\Ybot{0}   % y-coordinate of bottom axis
  \pgfmathsetmacro{\xone}{\L+(\R-\L)/3}
  \pgfmathsetmacro{\xtwo}{\L+2*(\R-\L)/3}
  
    %--- axes & labels ----------------------------
  \draw[->] (\L-0.4,\Ytop-0.2) -- (\R+0.5,\Ytop-0.2);     % top axis
  \node[above] at (\L-0.12,\Ytop+0.15) {0};
  \node[above] at (\R+0.12,\Ytop+0.15) {1};
  \draw[->] (\L-0.4,\Ybot) -- (\R+0.5,\Ybot);     % bottom axis
  \node[below] at (\L-0.12,\Ybot-0.03) {0};
  \node[below] at (\R+0.12,\Ybot-0.03) {1};
  \node[below] at (\R+0.55,\Ytop-0.2) {$\state$};
  \node[below] at (\R+0.55,\Ybot-0.04) {$\posteriorMean$};
  %--- dashed vertical boundaries -------------------------
  \draw[dashed] (\L,\Ybot) -- (\L,\Ytop);
  \draw[dashed] (\R,\Ybot) -- (\R,\Ytop);
  % ======== Common part ends =======================
  
  \pgfmathsetmacro{\M}{\L+(\R-\L)/12} 
  \pgfmathsetmacro{\N}{\L+7*(\R-\L)/12}

  %--- hatched rectangle -----------------------------------------------
  \draw[pattern=north east lines,pattern color=blue, draw=black]
    (\M,\Ytop-0.2) rectangle (\N,\Ytop+0.2);
  
  \path[pattern=north west lines,pattern color=red, draw=black]
    (\L,\Ytop-0.2) rectangle (\M,\Ytop+0.2);
  
  \path[pattern=north west lines,pattern color=red, draw=black]
    (\N,\Ytop-0.2) rectangle (\R,\Ytop+0.2);

  \node[above] at ({\L+(\R-\L)/12},\Ytop + 0.1) {1/12};
  \node[above] at ({\L+ 7*(\R-\L)/12},\Ytop + 0.1) {7/12};

  %--- blue arrow + bar -----------------------------------------------
  \fill[blue!20] (\xone-0.3,\Ybot) rectangle (\xone+0.3,\Ybot+0.8);
  \draw          (\xone-0.3,\Ybot) rectangle (\xone+0.3,\Ybot+0.8);
  \draw[->, blue] ({\L+4*(\R-\L)/12},\Ytop- 0.2) -- (\xone,\Ybot + 0.82); 

  %--- red arrow + bar ------------------------------------------------
  \fill[red!20] (\xtwo-0.3,\Ybot) rectangle (\xtwo+0.3,\Ybot+0.8);
  \draw         (\xtwo-0.3,\Ybot) rectangle (\xtwo+0.3,\Ybot+0.8);
  \draw[->, red] ({\L+(\R-\L)/24},\Ytop-0.2) -- (\xtwo-0.33,\Ybot+0.81); 
  \draw[->, red] ({\L+9.5*(\R-\L)/12},\Ytop- 0.2) -- (\xtwo+0.3,\Ybot+0.83); 

  % Add points at the posterior means 1/3 and 2/3
  \node[below] at (\xone, \Ybot) {1/3};
  \node[below] at (\xtwo, \Ybot) {2/3};
  \fill (\xone, \Ybot) circle[radius=2pt];
  \fill (\xtwo, \Ybot) circle[radius=2pt];
\end{tikzpicture}}
    \caption{The (deterministic) information policy $\signalProb_2$}
    \label{fig:bipool-deter}
  \end{subfigure}
  \caption{Illustration of (non-explainable) optimal information policies in \Cref{example:bi-pooling}.
  The top rectangle represents the state space with uniform prior, and the two bottom rectangles represent the induced posterior means. }
  \label{fig:bipooling1}
\end{figure}
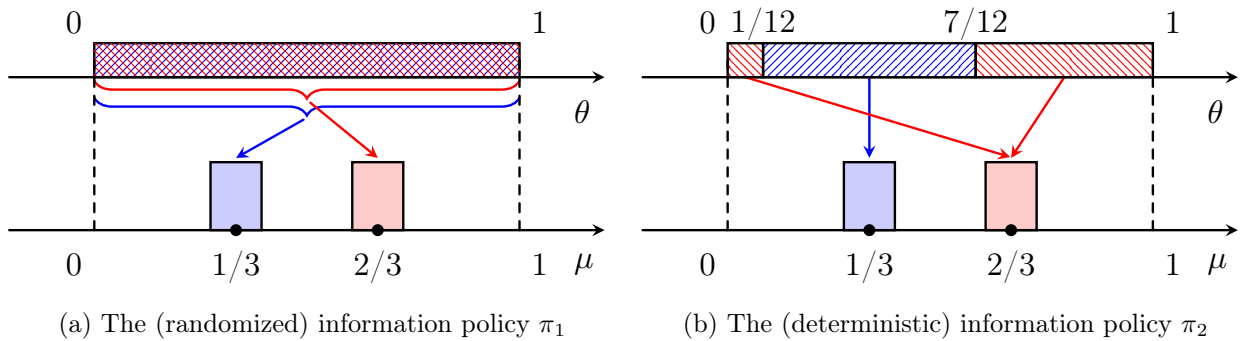

% \begin{figure}[htbp]
%     \centering
%     \input{Paper/plots/bipooling-random}
%     \caption{s}
%     \label{s}
% \end{figure}

% \begin{figure}[htbp]
%     \centering
%     \input{Paper/plots/bipooling-deter}
%     \caption{s}
%     \label{s}
% \end{figure}

As we can see from \Cref{example:bi-pooling}, 
the unrestricted optimal policy can be randomized 
% \jamie{I think this is only necessary when there are mass points. In any case, that's not shown by this example.} 
%across multiple signals
(policy $\signalProb_1$ randomizes between two signals conditioning on each state) or deterministic but non-monotone (policy $\signalProb_2$ sends $s_2$ when the state falls in two disconnected intervals). %\yc{Why non-monotone is undesirable may need more motivation.}
More generally, the optimal information policies in single-dimensional linear information design problems are known to have the following \emph{bi-pooling} property \citep{candogan2019optimality, KMS-21, arieli_optimal_2023}: 

\begin{definition}[Bi-pooling] \label{def:bi-pooling}
A $K$-signaling scheme $\signalProb$ is \emph{bi-pooling} if it divides the state space $[0, 1]$ into consecutive intervals $\{(b_{i-1}, b_i)\}_{i}$ and for all states in each interval, % $(b_{i-1}, b_i)$,
the policy sends either one or two signals. 
When it sends two signals for an interval, it can randomize or pool states in disconnected parts into one signal as in \Cref{example:bi-pooling}.
\end{definition}

% In fact, as we will show in \Cref{ex:continued}, the optimal $\numSignals$-signaling scheme %\yc{Still restricted to $K$ signals.}\taocomment{changed the term.} 
% in this instance must either randomize across multiple signals or send the same signal for disconnected intervals.

The randomization and non-monotonicity (pooling of disconnected parts) of optimal information policies are often undesirable in practice.
% Practical considerations such as transparency, legal regulations, often demand an information policy be more interpretable.
As mentioned in the Introduction, practical considerations often demand an information policy be easily interpretable.

% often preclude such non-explainable information policies. 

\xhdr{Explainable information design.}
% the optimal information policy may be necessarily randomizes over different signals, or may need to send signals in a non-monotone manner.
% Yet in practice, we may hope the information policy to be explainable due to various constraints like incentive, legal, or other practical considerations.
Motivated by the above observations, we define the following class of simple and
explainable information policies: %(interval-)partitional policies. 
% \yc{Our partitional information policy is just a $K$-monontone information policy? I just notice this year's EC paper ``On Monotone Persuasion.''}
% \taocomment{Yes. But I think ``partitional'' is a better word than ``monotone''.  It is more intuitive and applies to the multi-dimensional case.
% \cite{lyu2024coarseinformationdesign} used the word ``interval-partitional''.}
% where the information policy can be interpreted as a (possibly small) decision tree.
\begin{definition}[Interval-partitional information policy]
\wtedit{For a positive integer $\numSignals$, let
\begin{align*}
    \partpolicies
    :=
    \left\{
    (a_1,\ldots,a_\numSignals)\in[0,1]^K:
    0\le a_1\le \cdots \le a_{K-1}\le a_\numSignals=1
    \right\}\,.
\end{align*}
For any $\partionSet=(a_1,\ldots,a_\numSignals)\in\partpolicies$, with $a_0=0$, the corresponding $\numSignals$-(interval-)partitional information policy $\signalProb_\partionSet$ divides the state space $[0,1]$ into the consecutive intervals
$(a_{i-1},a_i)$, $i\in[\numSignals]$,\footnote{We use $[\numSignals]$ to denote $\{1, \ldots, \numSignals\}$.}  and %deterministically
sends a unique signal for all states in each interval $(a_{i-1},a_i)$.}\footnote{For states at the partition points $\{a_i\}_{i\in[\numSignals]}$, $\signalProb_\partionSet$ can send an arbitrary signal.  As the prior distribution is continuous, those points have probability $0$ and will not affect the designer's expected utility.
} 
% Let $\partionSet = (a_i)_{i\in[\numSignals]}$, where $0 \le a_1 \le \cdots \le a_{\numSignals-1} \le a_\numSignals = 1$, be a set of partition points.\footnote{We define $a_0 = 0$. And for a positive integer $\numSignals$, let $[\numSignals] = \{1, \ldots, \numSignals\}$.} 
% A {\em $\numSignals$-{(interval-)}partitional information policy} $\signalProb_\partionSet$ divides the state space $[0, 1]$ into $\numSignals$ consecutive intervals ${(a_{i-1}, a_{i})}_{i\in [\numSignals]}$ and for all states in each interval $(a_{i-1}, a_i)$ deterministically sends a unique signal $\signal_i$.
\end{definition}
% \jamie{I think the definition of $\mathcal A$ is a little funny in the above definition. I read this as the set of all possible partitions, but it's being used as one specific partition.}
% \taocomment{Changed the definition and notation; hopefully it is clear now. }
We highlight a few desirable properties of a partitional information policy: 
(1) \emph{deterministic}: it sends only one signal in each interval; 
% (2) it is a monotone partitional policy as it pools the states into convex sets (i.e., intervals and singletons) and the signal simply indicates which set contains the realized state. Given this desired monotone property, some literature also refer to partitional policy as monotone policy (see, e.g., \citealp{KLZ-25}). 
(2) \emph{monotone} (according to the definition by, e.g., \citealp{KLZ-25}): it divides the state space into convex sets (intervals) and simply reveals which convex set contains the realized state.
Note that a $\numSignals$-partitional information policy can be implemented by a decision tree with $\numSignals$ leaves.
This aligns with the literature on explainable clustering %(e.g., \citealp{MDRF-20,MS-22a, GPSY-23}),
(see Related Works) 
where an explainable clustering is one that can be described by a (potentially small) decision tree.

%We refer to a partitional information policy with at most $\numSignals$ intervals as $\numSignals$-partitional information policy. 
Given an information design instance $\instance = (\stateCDF, \interimU)$,
% the computation of the optimal $\numSignals$-partitional information policy can be formalized by the following program: 
an optimal $\numSignals$-partitional information policy is a solution to the following optimization problem: 
\begin{equation}\label{eq:OPT-part-K}
    \myoptPart[\instance]{\numSignals} \,=\,  \max\nolimits_{\partionSet \wtedit{\in \partpolicies}} \designerExpU(\signalProb_{\partionSet})~. 
\end{equation}
% \taocomment{Can we add the following expression for $U(\pi_\mathcal A)$ explicitly to help readers better understand $K$-partitional policies and other notations?  (This expression will also be used in Section 5.1, I think.) }
% \wtcomment{I prefer not, as our definitions are already self-contained in the prelim. below eq (5) is sort of analysis equation, and it should be in the technique sections}
% where 
% \begin{align}
%     \designerExpU(\signalProb_\partionSet) = \sum_{i=1}^k F([a_{i-1}, a_i]) \cdot u( \mu_{[a_{i-1}, a_i]}) = \sum_{i=1}^k \int_{a_{i-1}}^{a_i} f(x)\; \dd x \cdot u\Big( \frac{\int_{a_{i-1}}^{a_i} x f(x)\; \dd x}{ \int_{a_{i-1}}^{a_i} f(x)\; \dd x}\Big). 
% \end{align}
Because a $\numSignals$-partitional information policy uses at most $\numSignals$ signals, we have $\myoptPart[\instance]{\numSignals} \le \myopt[\instance]{\numSignals}$ by definition.
Below we revisit \Cref{example:bi-pooling} and illustrate the structure of $\numSignals$-partitional % information 
policies. 
% \begin{repexample}[continued]

\begin{figure}[h!] %[htbp]
  \centering
  \begin{subfigure}[b]{0.49\textwidth}
    \centering
    \resizebox{\linewidth}{!}{% \begin{tikzpicture}[>=stealth, line cap=round, thick, scale=0.9]
\begin{tikzpicture}[
  >=stealth,
  line cap=round,
  thick,
  scale=1,
  every node/.append style={font=\scriptsize}
]
  % ========= Common part begins ======================
  %--- parameters --------------------------
  \def\L{1}      % left end of support
  \def\R{6}      % right end of support
  \def\Ytop{1.45}   % y-coordinate of top axis
  \def\Ybot{0}   % y-coordinate of bottom axis
  \pgfmathsetmacro{\xone}{\L+(\R-\L)/3}
  \pgfmathsetmacro{\xtwo}{\L+2*(\R-\L)/3}
  
  %--- axes & labels ----------------------------
  \draw[->] (\L-0.4,\Ytop-0.2) -- (\R+0.5,\Ytop-0.2);     % top axis
  \node[above] at (\L-0.12,\Ytop+0.15) {0};
  \node[above] at (\R+0.12,\Ytop+0.15) {1};
  \draw[->] (\L-0.4,\Ybot) -- (\R+0.5,\Ybot);     % bottom axis
  \node[below] at (\L-0.12,\Ybot-0.03) {0};
  \node[below] at (\R+0.12,\Ybot-0.03) {1};
  \node[below] at (\R+0.55,\Ytop-0.2) {$\state$};
  \node[below] at (\R+0.55,\Ybot-0.04) {$\posteriorMean$};
  %--- dashed vertical boundaries -------------------------
  \draw[dashed] (\L,\Ybot) -- (\L,\Ytop);
  \draw[dashed] (\R,\Ybot) -- (\R,\Ytop);
  % ======== Common part ends =======================

  % ===== specific part begins ====== 
  \node[above] at (\xtwo,\Ytop + 0.1) {2/3};
  
  %--- hatched rectangle ----------------------------------
  \draw[pattern=north east lines,pattern color=blue, draw=black]
    (\L,\Ytop-0.2) rectangle (\xtwo,\Ytop+0.2);
  
  \path[pattern=north west lines,pattern color=gray, draw=black]
    (\xtwo,\Ytop-0.2) rectangle (\R,\Ytop+0.2);
 
 % \draw[->, red] ({\L+9.5*(\R-\L)/12},\Ytop- 0.2) -- (13/3,\Ytop- 1.15); 

  %--- blue arrow + bar --------------------------------
  \fill[blue!20] (8/3-0.3,\Ybot) rectangle (8/3+0.3,\Ybot+0.8);
  \draw          (8/3-0.3,\Ybot) rectangle (8/3+0.3,\Ybot+0.8);
  \draw[->, blue] ({\L+2*(\R-\L)/6},\Ytop- 0.2) -- (8/3,\Ybot + 0.82);

  %--- gray arrow + bar --------------------------------- 
  \fill[gray!20] ({\L+5*(\R-\L)/6 -0.3},\Ybot) rectangle ({\L+5*(\R-\L)/6 +0.3},\Ybot+0.4);
  \draw         ({\L+5*(\R-\L)/6 -0.3},\Ybot) rectangle ({\L+5*(\R-\L)/6 +0.3},\Ybot+0.4);
  \draw[->, gray] ({\L+5*(\R-\L)/6},\Ytop-0.2) -- ({\L+5*(\R-\L)/6}, \Ybot + 0.42);

  % Add points at the posterior means 1/3 and 2/3
  \node[below] at ({\L+2*(\R-\L)/6},\Ybot) {1/3};
  \node[below] at ({\L+5*(\R-\L)/6},\Ybot) {5/6};
  \fill ({\L+(\R-\L)/3}, \Ybot) circle[radius=2pt];
  \fill ({\L+(\R-\L)*2/3}, \Ybot) circle[radius=2pt];
\end{tikzpicture}}
    % \caption{s}
    % \label{fig:2-part}
  \end{subfigure}%
  \hfill
  \begin{subfigure}[b]{0.49\textwidth}
    \centering
    \resizebox{\linewidth}{!}{% \begin{tikzpicture}[>=stealth, line cap=round, thick, scale=0.9]
\begin{tikzpicture}[
  >=stealth,
  line cap=round,
  thick,
  scale=1,
  every node/.append style={font=\scriptsize}
]

  % ========= Common part begins ======================
  %--- parameters --------------------------
  \def\L{1}      % left end of support
  \def\R{6}      % right end of support
  \def\Ytop{1.45}   % y-coordinate of top axis
  \def\Ybot{0}   % y-coordinate of bottom axis
  \pgfmathsetmacro{\xone}{\L+(\R-\L)/3}
  \pgfmathsetmacro{\xtwo}{\L+2*(\R-\L)/3}
  
  %--- axes & labels ----------------------------
  \draw[->] (\L-0.4,\Ytop-0.2) -- (\R+0.5,\Ytop-0.2);     % top axis
  \node[above] at (\L-0.12,\Ytop+0.15) {0};
  \node[above] at (\R+0.12,\Ytop+0.15) {1};
  \draw[->] (\L-0.4,\Ybot) -- (\R+0.5,\Ybot);     % bottom axis
  \node[below] at (\L-0.12,\Ybot-0.03) {0};
  \node[below] at (\R+0.12,\Ybot-0.03) {1};
  \node[below] at (\R+0.55,\Ytop-0.2) {$\state$};
  \node[below] at (\R+0.55,\Ybot-0.04) {$\posteriorMean$};
  %--- dashed vertical boundaries -------------------------
  \draw[dashed] (\L,\Ybot) -- (\L,\Ytop);
  \draw[dashed] (\R,\Ybot) -- (\R,\Ytop);
  % ======== Common part ends =======================
  
  \node[above] at ({\L+(\R-\L)/6},\Ytop + 0.1) {1/6};
  \node[above] at ({\L+ 3*(\R-\L)/6},\Ytop + 0.1) {1/2};
  \node[above] at ({\L+ 5*(\R-\L)/6},\Ytop + 0.1) {5/6};

  % \node[right] at (\R+0.5,\Ytop - 0.6) {$\state$};
  % \node[right] at (\R+0.5,\Ybot - 0.4) {$\posteriorMean$};

  \pgfmathsetmacro{\Mone}{\L+(\R-\L)/6} 
  \pgfmathsetmacro{\Mtwo}{\L+3 * (\R-\L)/6} 
  \pgfmathsetmacro{\Mthree}{\L+5 * (\R-\L)/6} 

  %--- hatched rectangle --------------------------------------
  \draw[pattern=north west lines,pattern color=gray, draw=black]
    (\L,\Ytop-0.2) rectangle (\Mone,\Ytop+0.2);
  
  \path[pattern=north east lines,pattern color=blue, draw=black]
    (\Mone,\Ytop-0.2) rectangle (\Mtwo,\Ytop+0.2);
  
  \path[pattern=north west lines,pattern color=red, draw=black]
    (\Mtwo,\Ytop-0.2) rectangle (\Mthree,\Ytop+0.2);

  \path[pattern=north east lines,pattern color=gray, draw=black]
    (\Mthree,\Ytop-0.2) rectangle (\R,\Ytop+0.2);

  %--- blue arrow + bar ------------------------------------
  \fill[blue!20] (8/3-0.3,\Ybot) rectangle (8/3+0.3,\Ybot+0.4);
  \draw          (8/3-0.3,\Ybot) rectangle (8/3+0.3,\Ybot+0.4);
  \draw[->, blue] ({\L+1*(\R-\L)/3},\Ytop - 0.2) -- (8/3,\Ybot + 0.42); 
  
  %--- red arrow + bar -------------------------------------
  \fill[red!20] (13/3-0.3,\Ybot) rectangle (13/3+0.3,\Ybot+0.4);
  \draw         (13/3-0.3,\Ybot) rectangle (13/3+0.3,\Ybot+0.4);
  \draw[->, red] ({\L+2 * (\R-\L)/3},\Ytop-0.2) -- (13/3,\Ybot + 0.42); 

  % Add points at the posterior means 1/3 and 2/3
  \node[below] at (\xone,\Ybot) {1/3};
  \node[below] at (\xtwo,\Ybot) {2/3};
  \fill (\xone, \Ybot) circle[radius=2pt];
  \fill (\xtwo, \Ybot) circle[radius=2pt];
\end{tikzpicture}}
    % \caption{s}
    % \label{fig:4-part}
  \end{subfigure}
  \caption{Illustration of partitional information policies in \Cref{ex:continued}. Left: an optimal $(\numSignals = 2)$-partitional policy. Right: an optimal $(\numSignals = 4)$-partitional policy. Both of them obtain expected utility $2/3$. Blue and red signals induce desired posterior means; gray signals have interim utility $0$. }
  \label{fig:bipooling2}
\end{figure}
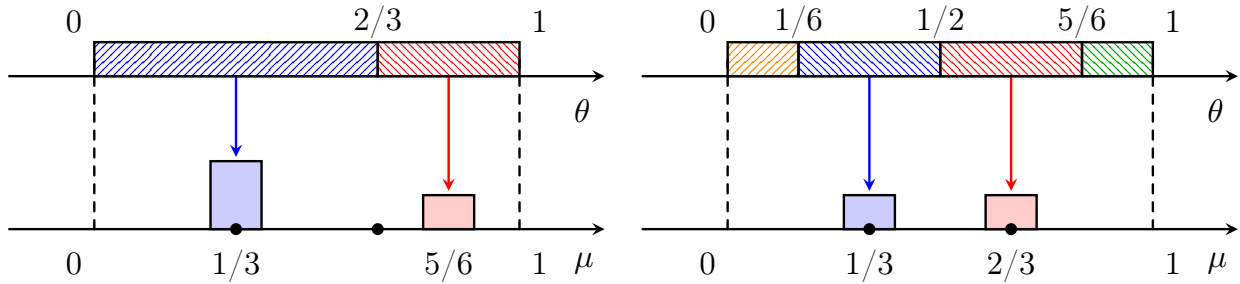
% \wtcomment{we may want to include a figure for this}
% \taocomment{Use $\Part$ or something like $\cc{Exp}$ to denote partitional/explainable signaling schemes?}

\begin{example}[\Cref{example:bi-pooling} continued]
\label{ex:continued}
{\rm In \Cref{example:bi-pooling}, an optimal $(\numSignals = 2)$-partitional information policy is to have two intervals $[0, \frac{2}{3}]$ and $(\frac{2}{3}, 1]$, with posterior means $\frac{1}{3}$ and $\frac{5}{6}$ respectively. The first posterior mean has interim utility $1$ and is induced with probability $F([0, \frac{2}{3}]) = \frac{2}{3}$, while the second posterior mean has interim utility $0$. The information designer's expected utility is thus $\frac{2}{3}$.
An optimal $(\numSignals = 4)$-partitional information policy have intervals $[0, \frac{1}{6}), [\frac{1}{6}, \frac{1}{2}), [\frac{1}{2}, \frac{5}{6}), [\frac{5}{6}, 1]$, where the second and third intervals induce posteriors $\frac{1}{3}$ and $\frac{2}{3}$ with total probability $\frac{2}{3}$. 
In this instance, no $K$-partitional information policy can induce the two desired posteriors $\frac{1}{3}$ and $\frac{2}{3}$ with total probability more than $\frac{2}{3}$, so $\myoptPart[\instance]{\numSignals \ge 2} = \frac{2}{3} < \myopt[\instance]{\numSignals \ge 2} = 1$.
}
\end{example}

\xhdr{Price of explainability.}
How much utility will be lost for the designer when restricting to use explainable information policies? 
We define the ratio between the objective values of the optimal $\numSignals$-partitional information policy and the optimal $\numSignals$-signaling scheme to be the \emph{price of explainability} (PoE). % of an information design problem.
% In our definition, PoE is in $[0, 1]$; PoE being close to $0$ means that the loss due to restriction to explainable/partitional signaling schemes is large; PoE = 1 means no loss due to such restriction. 
For a class $\instanceClass$ of information design instances, the PoE is defined to be the worst-case PoE over all instances in the class.

\begin{definition}[Price of explainability]
\label{defn:poe}
Given an information design problem instance $\instance = (\stateCDF, \interimU)$, the {\em price of explainability} of $\numSignals$-partitional information policies is 
\begin{align*}
    \myPoE{\instance, \numSignals} \, :=  \, \frac{\myoptPart[\instance]{\numSignals}}{\myopt[\instance]{\numSignals}}~.
\end{align*}
For a class of instances $\instanceClass$, we define $\myPoE{\instanceClass, K} := \inf_{\instance \in \instanceClass} \myPoE{\instance, \numSignals}$.
\end{definition}

\section{Price of Explainable Information Design}
\label{sec:price-of-explainable}
% \taocomment{Should we use the word ``simple'' or ``explainable''?  While a bi-pooling signaling scheme does seem unexplainable, a full revelation scheme feels very explainable to me (although it requires infinitely many signals). The comparison between $\numSignals$-partitional vs $\numSignals$-bi-pooling schemes can be called ``price of explainability''. But the comparison between $\numSignals$-partitional schemes and full-revelation schemes is more of ``price of finiteness/simplicity'' instead of ``price of explainability''.}

% \wtcomment{Extension to S-shaped utility?} \taocomment{TODO: add a short definition of S-shaped utility}

% \taocomment{Maybe add the term ``linear information design'' back to Section 2, so here we can say ``Previous work on linear information design''. }
We study the price of explainability for single-dimensional linear information design in this section. 
Previous work 
% on information design with single-dimensional state space 
on linear information design
\citep{lyu2024coarseinformationdesign, kleiner_extreme_2025} has shown that, when the information designer's interim utility function $\interimU$ is convex, concave, or S-shaped (convex in $[0, t]$ and concave in $[t, 1]$ for some $t$), $\numSignals$-partitional information policies are optimal among unrestricted $\numSignals$-signaling schemes, namely, $\myoptPart[\instance]{\numSignals} = \myopt[\instance]{\numSignals}$. This means that, in our terminology, the price of explainability is $1$. 

\begin{proposition}[PoE for special utility functions]
\label{prop:K-partitional-convex-concave}
% When the interim utility function $\interimU$ is convex, concave, or S-shaped, for any continuous prior distribution $\stateCDF$, for any $\numSignals \ge 1$,
For any single-dimensional linear information design instance $\instance = (\stateCDF, \interimU)$ with a convex, concave, or S-shaped interim utility function $\interimU$, for any $\numSignals \ge 1$, 
$\myPoE{\instance, \numSignals} = 1$.
% \yc{For any $\numSignals\ge 1$?} \taocomment{added}
\end{proposition}
% \wtcomment{Is that possible that we can deduce a sufficient condition of the instances under which it has $\PoE(\instance) = 1$? This would generalize all previous results, and we can frame it as sufficient conditions for no price of explainability}

\subsection{Tight Characterization of PoE for General Instances}
Our main result in this section is a tight characterization of the price of explainability for general instances %with general priors and interim utility functions
in the single-dimensional case.
We show that the price of explainability in any instance is at least $1/2$; in other words, the optimal $\numSignals$-partitional signaling scheme always provides a $1/2$ approximation to the optimal unrestricted $\numSignals$-signaling scheme: $\myoptPart[\instance]{\numSignals} \ge \frac{1}{2} \myopt[\instance]{\numSignals}$. 
We also show that this result is tight: there exist instances where the price of explainability is arbitrarily close to $1/2$. 
\begin{theorem}[PoE for general utility functions]
\label{thm:K-partitional-1/2}
% For any information design instance $\instance = (F, u)$, for any $\numSignals\ge 1$, the optimal $\numSignals$-partitional signaling scheme provides a $1/2$-approximation to the optimal unrestricted $\numSignals$-signaling scheme. And this approximation ratio is tight: 
% \begin{equation}
%     \inf_{\instance} \, \frac{\OPT^\Part_\instance(K)}{\OPT_\instance(K)} ~ = ~ \frac{1}{2}. 
% \end{equation}
% For any information design instance $\instance = (F, u)$, for any $\numSignals\ge 1$, the price of explainability $\PoE(\instance) \ge 1/2$. 
% Moreover, for any $\eps>0$, there exists an instance $\instance$ such that $\PoE(\instance) \le 1/2 + \eps$.  This means
% \begin{equation}
%     \PoE(\text{all instances}) ~ = ~ \inf_{\instance} \, \frac{\OPT^\Part_\instance(K)}{\OPT_\instance(K)} ~ = ~ \frac{1}{2}. 
% \end{equation}
% \begin{equation}
%     \inf_{\instance} \, \frac{\OPT^\Part_\instance(K)}{\OPT_\instance(K)} ~ = ~ \frac{1}{2}. 
% \end{equation}
In the $(m=1)$-dimensional case, we have the following characterization of $\PoE$:
\begin{itemize}
    \item $\PoE$ {\em \textbf{lower bound}}:
    For any information design instance $\instance = (\stateCDF, \interimU)$, for any $\numSignals\ge 2$, the price of explainability $\myPoE{\instance, \numSignals} \ge 1/2$. 
    \item $\PoE$ {\em \textbf{upper bound}}:
    For any $\eps>0$, there exists an instance $\instance = (\stateCDF, \interimU)$ such that, for any $\numSignals\ge 2$, the price of explainability $\myPoE{\instance, \numSignals} \le 1/2 + \eps$. 
\end{itemize}
As a result, for the class $\mathcal C_{\mathrm{all}}$ of all instances with continuous $F$ and arbitrary $u$, we have % This means
    \begin{equation*}
        \myPoE{\mathcal C_{\mathrm{all}}, \numSignals} ~ = \, \inf_{\instance \in \mathcal C_{\mathrm{all}}}  \frac{\myoptPart[\instance]{\numSignals}}{\myopt[\instance]{\numSignals}} ~ = ~ \frac{1}{2}. 
    \end{equation*}
\end{theorem}

This $1/2$ PoE result will be a corollary of our result for the high-dimensional case in \Cref{sec:high-dimension}, so we only provide a proof sketch here. As we mentioned in \Cref{def:bi-pooling}, an optimal signaling scheme in single-dimensional linear information design is bi-pooling: it divides the state space $[0, 1]$ into intervals $\{[b_{i-1}, b_i]\}_i$ and sends one or two signals for each interval. 
We further divide each bi-pooling interval $[b_{i-1}, b_i]$ into $[b_{i-1}, c_i] \cup [c_i, b_i]$ to construct a partitional policy that sends only one signal for either sub-interval; one of these two signals induces the same posterior mean as the better one of the two bi-pooling signals. Such a conversion does not change the number of signals used while preserves half of the expected utility of the optimal policy, thus achieving $1/2$ approximation.  % We also construct worst-case instances to show the tightness of $1/2$. 

% However, we will still prove this result here, because the proof for the single-dimensional case illustrates the main idea behind the proof for the high-dimensional case. \Cref{sec:proof-PoE-lower-bound} and \Cref{sec:proof-PoE-upper-bound} will prove the lower bound and upper bound parts of \Cref{thm:K-partitional-1/2}, respectively.

% After proving \Cref{thm:K-partitional-1/2}, we then discuss in \Cref{sec:uniform} an interesting special case where the prior distribution $F$ is uniform over $[0, 1]$, where the PoE can be improved to $2/3$. 

% \wtcomment{mention that we have proved stronger result such that there exists an instance $\instance$ such that for all $\numSignals$, we have $\frac{\myoptPart[\instance]{\numSignals}}{\OPT_\instance} = \sfrac{1}{2}$}

% Let $O =\{ x\in \MPC(K): x \text{ maximizes} \langle x, u\rangle \} \subseteq \MPC(K)$ and should be compact.  Take $x_0 \in O$.  Suppose $x_0 = \sum_{i=1}^n \lambda_i x_i$ is not an extreme point. All of $x_i$'s are in $O$. $O_1 =\{v \in O: \sup_{x\in[0, 1]} v(x) \text{ is minimal among all the points in $O$} \}\subseteq O$ 

\subsection{Better PoE under Uniform Prior}
\label{sec:uniform}
% We conclude this section with one more example class of instances with a nontrivial price of explainability that we can tightly characterize.
We then consider an interesting special case where the prior $F$ is the uniform distribution over $[0, 1]$. We show that, with uniform prior, the price of explainability of $K$-partitional information policies can be improved from $1/2$ to $2/3$, and the $2/3$ ratio is tight.  We first present a proposition for binary-valued utility functions (the proof is in Appendix \ref{proof:uniform-binary-2/3}): 

\begin{proposition}\label{prop:uniform-binary-2/3}
	Let $\mathcal{C}_{\mathrm{unif, \,bin}}$ be the class of information design instances $\mathcal{I} = (F, u)$ with $F = \mathrm{Uniform}[0, 1]$ and binary-valued $u$ (i.e., $u(x)\in\{0, 1\}$). Then, for any $K \geq 4$, $\PoE(\mathcal{C}_{\mathrm{unif, \,bin}}\,, K) = \frac23$. 
\end{proposition}
% \taocomment{Todo: polish the proof in Appendix \ref{proof:uniform-binary-2/3}}
For general utility functions, we can achieve $2/3$ approximation by increasing the number of signals of the policy.
In particular, for $K' \ge 3K/2$, the price of explainability of $K'$-partitional policies against unrestricted $K$-signaling scheme is exactly $2/3$.  

\begin{theorem} \label{thm:uniform-2/3}
    Let $\mathcal C_{\mathrm{unif}}$ be the class of instances $\instance = (F, u)$ with $F = \mathrm{Uniform}[0, 1]$ and any $u$. 
    \begin{itemize}
        \item For any $\instance \in \mathcal C_{\mathrm{unif}}$, $K \ge 2$, $K' \ge 3K/2$, we have $\myoptPart[\instance]{K'} \ge \frac{2}{3} \myopt[\instance]{K}$. 
        \item There exists $\instance \in \mathcal C_{\mathrm{unif}}$ such that, for $K \ge 2$, $K' \ge 3$, $\myoptPart[\instance]{K'} \le \frac{2}{3} \myopt[\instance]{K}$. 
    \end{itemize}
    Therefore, for $K\ge 2, K'\ge 3K/2$, $\myPoE{\mathcal C_{\mathrm{unif}}, K, K'} := \inf_{\instance \in \mathcal C_{\mathrm{unif}}}  \frac{\myoptPart[\instance]{K'}}{\myopt[\instance]{K}} = \frac{2}{3}$. \jamie{This is undefined.}
\end{theorem}

As we will show below, the increase in the number of signals is because the conversion process from optimal (bi-pooling) signaling scheme to partitional information policy increases the number of signals in each bi-pooling interval from $2$ to $3$.  Whether one can %design a conversion process to
obtain $2/3$ approximation without increasing the number of signals is an open question.

\begin{proof}[Proof of \Cref{thm:uniform-2/3}]
The second part of the theorem, about the existence of an instance $\instance \in \mathcal C_{\mathrm{uniform}}$ such that $\myoptPart[\instance]{K'} \le \frac{2}{3} \myopt[\instance]{K}$, follows from \Cref{ex:continued}. 

To prove the first part of the theorem, we take as given any optimal unrestricted $K$-signaling scheme $\pi^*$, which is bi-pooling (\Cref{def:bi-pooling}), and convert it into a $K'$-partitional policy $\pi_A$.
Specifically, we consider every bi-pooling interval $[b_{i-1}, b_i]$ of $\pi^*$ and further divide $[b_{i-1}, b_i]$ into three sub-intervals to obtain a 3-partitional policy with at least $2/3$ of the expected utility of $\pi^*$ on interval $[b_{i-1}, b_i]$; this is formalized in the following \Cref{lem:uniform-prior-2/3}. 
Performing this construction for every bi-pooling interval of $\pi^*$ gives a partitional policy $\pi_A$ on the entire state space $[0, 1]$ with utility $U(\pi_A) \ge \frac{2}{3} U(\pi^*) = \tfrac{2}{3} \myopt[\instance]{K}$, using at most $K' = 3K/2$ signals, which proves the theorem. 

\begin{lemma}\label{lem:uniform-prior-2/3}
Let $[b_{i-1}, b_i]$ be an interval where the optimal policy $\pi^*$ sends two signals.  Under uniform prior, there exists a $(K=3)$-partitional policy $\pi_{K=3}$ on interval $[b_{i-1}, b_i]$ such that the conditional expected utility satisfies $U_{[b_{i-1}, b_i]}(\pi_{K=3}) \ge \frac{2}{3} U_{[b_{i-1}, b_i]}(\pi^*)$.  
\end{lemma}

We prove \Cref{lem:uniform-prior-2/3} by constructing such a $(K=3)$-partitional policy. Without loss of generality, we normalize the interval $[b_{i-1}, b_i]$ to $[0, 1]$. Under uniform prior, the mean of the interval is $1/2$.  Suppose the posterior means associated with the two signals of $\pi^*$ are $\posteriorMean_1 \in [0, 1/2), \posteriorMean_2 \in (1/2, 1]$, with interim utility $u_1 = \interimU(\mu_1), u_2 = \interimU(\mu_2)$.  By symmetry, we assume without loss of generality that $u_1 \ge u_2 \ge 0$.  The expected utility of $\pi^*$ is 
\begin{equation} \label{eq:U(pi*)}
    U(\pi^*) = p_1 u_1 + p_2 u_2
\end{equation}
where $p_1, p_2$ are the total probabilities that the two signals are sent, %The two probabilities should satisfy
% The Bayes-plausibility condition $p_1 \mu_1 + p_2 \mu_2 = 1/2$ and $p_1 + p_2 = 1$ imply that % , which yields: % has a unique solution: 
% \begin{equation} \label{eq:p1-p2}
%     p_1 = \tfrac{\mu_2 - 1/2}{\mu_2 - \mu_1} \qquad p_2 = \tfrac{1/2-\mu_1}{\mu_2-\mu_1} \,. 
% \end{equation}
satisfying the Bayes-plausibility condition $p_1 \mu_1 + p_2 \mu_2 = 1/2$ and $p_1 + p_2 = 1$. 

Let's construct a partitional policy in a greedy way. Starting from the left endpoint 0, take the largest interval $I^{\max}_1$ whose conditional mean is $\mu_1$: by the uniform prior assumption, the interval is $I^{\max}_1 = [0, 2\mu_1] \subseteq [0, 1]$. If $2\mu_1 < \mu_2$, then we take another interval (starting from $2\mu_1$) whose conditional mean is $\mu_2$, which is $I_2 = [2\mu_1, 2\mu_2 - 2\mu_1]$.  If $2\mu_1 \ge \mu_2$, then we use $I^{\max}_1$ only.
Combining with the remaining interval (to the right endpoint 1), we obtain the following partitional information policy $\pi_1$ with at most 3 parts: 
\begin{align}
    \pi_1 = \begin{cases}
        I^{\max}_1 \cup I_2 \cup [2\mu_2-2\mu_1, 1] & \text{ if $2\mu_1 < \mu_2$} \\
        I^{\max}_1 \cup [2\mu_1, 1] & \text{ if $2\mu_1 \ge \mu_2$} 
    \end{cases}\,. 
\end{align}
Construct another partitional policy $\pi_2$ by taking the maximal interval starting from the right endpoint whose conditional mean is $\mu_2$, $I^{\max}_2 = [2\mu_2-1, 1]$, combined with the remaining interval to the left, $[0, 2\mu_2-1]$:  
\begin{align}
    \pi_2 = [0, 2\mu_2-1] \cup I^{\max}_2. 
\end{align}
We claim that the better policy among $\pi_1, \pi_2$ provides the desired 2/3-approximation: % that the expected utility is at least $2/3$ of $U(\pi^*)$:
\begin{lemma} \label{lemma:uniform-prior-better-among-2}
$\max \big\{\, U(\pi_1),\, U(\pi_2) \,\big\} \ge \tfrac{2}{3} U(\pi^*)$. 
\end{lemma}
To prove \Cref{lemma:uniform-prior-better-among-2}, we prove the following min-max inequality  
\[ \min_{0\le \mu_1 \le 1/2 \le \mu_2 \le 1, ~ 0\le u_2 \le u_1} \max \big\{\, U(\pi_1) - \tfrac{2}{3} U(\pi^*),\, U(\pi_2) - \tfrac{2}{3} U(\pi^*)\,\big\} ~ \ge ~ 0\, \]
by a case analysis.  
%In particular, we divide the feasible region of $(\mu_1, \mu_2)$ into five regions. Fixing $(\mu_1, \mu_2)$ within each region, we find the $u_2^*$ that solves $\min_{0\le u_2 \le u_1} \max \big\{  U(\pi_1) - \tfrac{2}{3} U(\pi^*), U(\pi_2) - \tfrac{2}{3} U(\pi^*) \big\}$, then show that the min-max objective value is non-negative.
See details in Appendix \ref{proof:lemma:uniform-prior-better-among-2}. 
\end{proof}

\section{Price of Explainability in Higher Dimensions}
\label{sec:high-dimension}
This section generalizes our single-dimensional PoE results to higher dimensions, where the state space is now $[0, 1]^m \subseteq \reals^m$. A key question in the high-dimensional case is how to define explainability. %``explainable information policy''.
\Cref{subsec:high-dimension-convex} considers \emph{convex-partitional} information policies, showing that their price of explainability is exactly $1/(m+1)$, which generalizes our single-dimensional $1/2$ PoE result. \Cref{subsec:high-dimension-rectangular} explores another definition, \emph{axis-aligned rectangular} policies. 

\subsection{Price of Explainability of Convex Partitional Policies} \label{subsec:high-dimension-convex}
This subsection considers \emph{convex-partitional} policies. %  as explainable policies.
A $K$-convex-partitional information policy divides the state space into $K$ convex subsets $[0, 1]^m = P_1 \cup \cdots \cup P_K$ with non-overlapping interiors and 
% sends one signal for all states in each set.
reveals which subset contains the state.\footnote{What signals to send at the boundary of the subsets do not matter under a continuous prior.}
It generalizes the $K$-interval-partitional policy we studied so far for the single-dimensional case. We show that the price of explainability of $K$-convex-partitional information policies is exactly $1/(m+1)$, which generalizes our $1/2$ PoE result in \Cref{thm:K-partitional-1/2}. 
\begin{theorem}
\label{thm:convex-1/(m+1)}
Assume continuous prior $F$ and arbitrary utility function $u$. 
For any $K\ge m+1$, the price of explainability of $K$-convex-partitional information policies is \[\PoE^{\cc{convex}}(\mathcal C_{\mathrm{all}}, K) := \inf_{\instance \in \mathcal C_{\mathrm{all}}} \frac{\OPT^{\cc{convex}}_{\instance}(K)}{\OPT_{\instance}(K)} = \frac{1}{m+1}.\] 
\end{theorem}

% The proof of this theorem uses an observation from \cite{kleiner_extreme_2025} that the extreme points of the set of mean-preserving contractions in the high-dimensional case divide the state space into convex subsets, with affinely independent supports in each subset. This observation allows us to convert an optimal $K$-signaling scheme into a $K$-convex-partitional policy with approximation factor $\frac{1}{m+1}$.
% We also prove that this $\frac{1}{m+1}$ factor is not improvable. See Appendix \ref{proof:convex-1/(m+1)} for details. 

We prove \Cref{thm:convex-1/(m+1)} by proving the lower and upper bound parts separately. 

\subsubsection{Proof of \Cref{thm:convex-1/(m+1)}: Lower Bound Part}
\label{sec:convex-1/(m+1)-lower-bound}
This subsection proves ``$\PoE^{\cc{convex}}(\mathcal C_{\mathrm{all}}, K) \ge \frac{1}{m+1}$''. 
Our proof strategy is to convert an optimal $\numSignals$-signaling scheme $\signalProb^*$ -- a solution to the problem $\myopt[\instance]{\numSignals}= \max\nolimits_{(\signalspace, \signalProb): |\signalspace| \le\numSignals} \designerExpU(\signalProb)$ -- into a $\numSignals$-convex-partitional information policy $\signalProb^\cc{convex}$ that satisfies $\designerExpU(\signalProb^\cc{convex}) \ge \frac{1}{m+1} \designerExpU(\signalProb^*)$.  
In particular, our proof proceeds as follows:
\begin{itemize}
    \item 
    We first characterize the optimal solution $\signalProb^*$ to % the optimization problem \eqref{eq:OPT-K},
    $\myopt[\instance]{\numSignals}$. % $= \max\nolimits_{(\signalspace, \signalProb): |\signalspace| \le\numSignals} \designerExpU(\signalProb)$.
    % This optimization problem is non-convex because it has a non-convex objective function and the space of all $\numSignals$-signaling schemes is a non-convex space.
    % has a non-convex objective function and a non-convex feasible region. 
    Because the objective function $\designerExpU(\signalProb)$ is non-convex, 
    % To overcome this obstacle, we reformulate the optimization problem
    we reformulate the problem by representing an information policy by the induced distribution over posterior means, known as a \emph{mean-preserving contraction (MPC)} of the prior $F$. % in the single-dimensional case %\citep{gentzkow_rothschild-stiglitz_2016}
    % or \emph{fusion} in the multi-dimensional case. % \citep{kleiner_extreme_2025}.
    This reformulation gives a linear objective function.
    % and transforms the feasible region to be the set of all mean-preserving contractions with support size at most $\numSignals$, denoted by $\mpc{\stateCDF, \numSignals}$.
    But the feasible region $\mpc{\stateCDF, \numSignals}$, which is the set of all MPCs with support size at most $\numSignals$, is not a convex set. 
    % The feasible region, however, is still non-convex.
    We prove that, although the feasible region is non-convex, there still exists an optimal solution at an extreme point of the feasible region.
    \item 
    We then characterize the extreme points of the non-convex feasible region $\mpc{\stateCDF, \numSignals}$.
    % , generalizing the result in \citep{KMS-21, arieli_optimal_2023}. These two works considered the set of all feasible mean-preserving contractions which may include distributions with infinite support, $\mpc{\stateCDF, +\infty}$, showing that the extreme points of $\mpc{\stateCDF, +\infty}$ have a ``bi-pooling'' structure: the information policy partitions the state space $[0, 1]$ into consecutive intervals $\{[b_{i-1}, b_i]\}$ and may send \emph{at most two} signals for each interval.
    Previous work \citep{kleiner_extreme_2025} shows that the extreme points of the set of MPCs with infinite support size, $\mpc{\stateCDF, \infty}$, have a ``$(m+1)$-pooling'' property.
    % We show that the same result holds even when the support size of the mean-preserving contractions is restricted to be at most $\numSignals < +\infty$.
    We generalize this result to MPCs with finite support size $K < \infty$. 
    % As a corollary, there must exist an optimal $\numSignals$-signaling scheme with a bi-pooling structure. 
    % for the space, denoted by $\mpc{\stateCDF}$, of all feasible mean-preserving contractions (which may include distributions with infinite support) to a restricted space, denoted by $\mpc{\stateCDF, \numSignals}$, of all mean-preserving contractions that have at most $\numSignals$ support. 
    % We show that the extreme points of this space must satisfy a ``bi-pooling'' structure: the policy partitions the state space $[0, 1]$ into consecutive intervals $\{[b_{i-1}, b_i]\}$ and it may send \emph{at most two} signals for each interval.
    \item 
    The above two steps together show that there must exist an optimal $\numSignals$-signaling scheme that is $(m+1)$-pooling. 
    % provide a characterization of the optimal $\numSignals$-signaling scheme which  
    % called ``bi-pooling'': the policy divides the state space $[0, 1]$ into consecutive intervals $\{[b_{i-1}, b_i]\}$ and may send \emph{at most two} signals for some interval.
    Exploiting this $(m+1)$-pooling property, we convert the optimal $\numSignals$-signaling scheme into a $\numSignals$-convex-partitional signaling scheme with a $1/(m+1)$ approximation guarantee.
\end{itemize}
% The remainder of this subsection details the above proof steps.
We then formalize the above three steps.

\xhdr{\bf Reformulation via mean-preserving contractions.}
To begin, we reformulate the information designer's optimization problem $\myopt[\instance]{\numSignals} = \max\nolimits_{(\signalspace, \signalProb): |\signalspace| \le\numSignals} \designerExpU(\signalProb)$ by representing any information policy as a \emph{mean-preserving contraction}.
Given an information policy $\signalProb$, %with signal space $\signalspace$,
let $G_\signalProb \in \Delta([0, 1]^m)$ be the distribution over posterior means induced by $\signalProb$: namely, $G_\signalProb$ is the distribution supported on $\{ \posteriorMean_{\signalProb, \signal} \mid \signal \in \signalspace\} \subseteq [0, 1]^m$ where the probability mass of $\posteriorMean_{\signalProb, \signal}$ is equal to the unconditional probability of signal $s$, $\int \signalProb(\signal | \state)\; \dd F(\theta)$. %\footnote{A mean-preserving contraction in the high-dimensional case is also called a \emph{fusion} \citep{kleiner_extreme_2025}. We will only use the terminology ``mean-preserving contraction''. } 
Let 
\begin{equation*}
	\mpc{\stateCDF, \numSignals} = \Big\{ G \in \Delta([0, 1]^m) \,\Big|~  \text{there exists a $\numSignals$-signaling scheme $\pi$ such that $G_\pi = G$} \Big\}
\end{equation*}
be the set of MPCs induced by $\numSignals$-signaling schemes from prior $\stateCDF$. Note that each $G \in \mpc{\stateCDF, \numSignals}$ is a distribution supported on at most $\numSignals$ points in $[0, 1]^m$, and $\mpc{\stateCDF, \numSignals}$ is not a convex set because the convex combination of two distributions with different supports of size $K$ may have a support of size larger than $K$. 

After representing information policies as MPCs, we can rewrite the information designer's optimization problem over $\numSignals$-signaling schemes as optimization over $\mpc{\stateCDF, \numSignals}$: 
\begin{equation} \label{eq:OPT-MPC-K}
     \myopt[\instance]{\numSignals} ~ = ~ \max_{G \in \mpc{\stateCDF, \numSignals}}\designerExpU(G) ~ = ~ \max_{G \in \mpc{\stateCDF, \numSignals}} \expect[\posteriorMean\sim G]{\interimU(\posteriorMean)} \, . 
\end{equation}
Importantly, the objective function $\designerExpU(G) = \expect[\posteriorMean\sim G]{\interimU(\posteriorMean)} = \int \interimU(\posteriorMean) \, \dd G(\posteriorMean)$ is now linear in $G \in \mpc{\stateCDF, \numSignals}$. 
Although the feasible region $\mpc{\stateCDF, \numSignals}$ is not convex, the linearity of $\designerExpU(G)$ still guarantees the existence of an optimal solution $G^*$ that is an \emph{extreme point} of $\mpc{\stateCDF, \numSignals}$.  A point $x$ in a non-convex set $M$ is said to be an extreme point of $M$ if $x$ cannot be written as the convex combination $\lambda y + (1-\lambda) z$ of two other points $y, z \in M \setminus \{x\}$, with $\lambda \in (0, 1)$.  

\begin{lemma} \label{lem:extreme-point-solution}
Problem \eqref{eq:OPT-MPC-K} has a solution $G^*$ that is an extreme point of $\mpc{\stateCDF, \numSignals}$. 
\end{lemma}

The above lemma would be obvious if $\mpc{\stateCDF, \numSignals}$ were a non-convex subset of a finite-dimensional vector space. However, $\mpc{\stateCDF, \numSignals}$ is infinite-dimensional, so the lemma is not immediate. We use the Krein-Milman Theorem for compact but non-convex sets to prove this lemma. See Appendix \ref{proof:extreme-point-solution} for details.

\xhdr{\bf Characterizing the extreme points of MPCs.}
We then characterize the extreme points of $\mpc{\stateCDF, \numSignals}$.
Based on a proof by \cite{kleiner_extreme_2025}, we show that the extreme points of $\mpc{\stateCDF, \numSignals}$ satisfy an \emph{$(m+1)$-pooling} property, which generalizes the bi-pooling property in the single-dimensional case (\Cref{def:bi-pooling}). 

\begin{definition}[$(m+1)$-pooling]
An information policy $\signalProb$ (with a finite signal space $\signalspace$) is an \emph{$(m+1)$-pooling policy} if it divides the state space $[0, 1]^m$ into finitely many convex subsets $[0, 1]^m = \bigcup_{i\in I} P_i$ and, for the states in each subset $P_i$, sends at most $m+1$ signals, inducing $(m+1)$ posterior means $\mu_{i, 1}, \ldots, \mu_{i, m+1}$ inside $P_i$. 

The MPC $G_\signalProb$ induced by an $(m+1)$-pooling information policy $\signalProb$ is called an \emph{$(m+1)$-pooling MPC}. 
\end{definition}

\citet{kleiner_extreme_2025} considered the set of MPCs induced by all signaling schemes without signal space size constraint, namely, $\mpc{\stateCDF, \numSignals = \infty}$. Unlike $\mpc{\stateCDF, K<\infty}$, $\mpc{\stateCDF, \infty}$ is a convex set.
\citet{kleiner_extreme_2025}'s Proposition 1 showed that any extreme point of $\mpc{\stateCDF, \infty}$ with a finite support must be an $(m+1)$-pooling MPC.  In particular, 
% shows that, without the support size constraint, any extreme point $G$ of the set $\mpc{\stateCDF, \numSignals = \infty}$ has the following structure: the state space is divided into multiple convex subsets $[0, 1]^m = \bigcup_{i\in I} P_i$, the conditional distribution $G|_{P_i}$ on each subset $P_i$ is a fusion of $F|_{P_i}$,
for any finite-support extreme point $G$ of $\mpc{\stateCDF, \infty}$, the support of $G$ within each convex subset $P_i$ must consist of \emph{affinely independent} vectors.\footnote{$N$ vectors $v_1, \ldots, v_N \in \reals^m$ are \emph{affinely independent} if $v_2 - v_1, \ldots, v_N - v_1$ are linearly independent; or equivalently, the only solution to $\sum_{i=1}^N \alpha_i v_N = 0$ and $\sum_{i=1}^N \alpha_i = 0$ is $\alpha_1 = \cdots = \alpha_N = 0$.}
Because the number of affinely independent vectors in $m$-dimensional space is at most $m+1$, the support size of $G$ within each convex subset $P_i$ cannot exceed $m+1$.
% In fact,
We note that \citet{kleiner_extreme_2025}'s proof for  $\mpc{\stateCDF, \infty}$ also works for $\mpc{\stateCDF, \numSignals < \infty}$, so we obtain the following characterization for the extreme points of $\mpc{\stateCDF, \numSignals}$: 
\begin{lemma}[Adaptation of \citealp{kleiner_extreme_2025}'s Proposition 1]\label{lem:(m+1)-pooling}
% Any extreme point $G$ of $\mpc{\stateCDF, \numSignals}$ is induced by some information policy that divides the state space into convex subsets $[0, 1]^m = \bigcup_{i\in I} P_i$, sends at most $m+1$ signals for the states in each $P_i$, and the induced posterior means (i.e., the support of $G|_{P_i}$) are affinely independent.
Any extreme point $G$ of $\mpc{\stateCDF, \numSignals}$ is an $(m+1)$-pooling MPC. 
\end{lemma}
\begin{proof}
We mention why \citet{kleiner_extreme_2025}'s proof for $K=\infty$ also works for $K<\infty$. They prove that any $G \in \MPC(F, \infty)$ that has a non-affinely-independent support in some convex subset $P_i$ is equal to the convex combination of two other $G_1, G_2 \in \MPC(F, \infty)$, so $G$ cannot be an extreme point of $\MPC(F, \infty)$.  %Importantly,
In their proof, $G_1$ and $G_2$ have the same support as $G$. So, if the support size of $G$ is at most $K$, i.e., $G\in \MPC(F, K)$, then the support sizes of $G_1$ and $G_2$ are also at most $K$, therefore $G_1, G_2 \in \MPC(F, K)$. So $G$ cannot be an extreme point of $\MPC(F, K)$. 
\end{proof}

As a corollary of Lemmas \ref{lem:extreme-point-solution} and \ref{lem:(m+1)-pooling}, % we obtain a characterization of an optimal $\numSignals$-signaling scheme for the information designer's optimization problem.  
there exists an optimal $\numSignals$-signaling scheme $\pi^*$ that is $(m+1)$-pooling.   
\begin{corollary} \label{cor:optimal-(m+1)-pooling}
The optimization problem $\myopt[\instance]{\numSignals} = \max\nolimits_{(\signalspace, \signalProb): |\signalspace| \le\numSignals} \designerExpU(\signalProb)$ has an $(m+1)$-pooling solution $\signalProb^*$.  
\end{corollary} 

\xhdr{\bf Converting optimal policy to convex-partitional policy.}
\Cref{cor:optimal-(m+1)-pooling} shows that the information designer's optimization problem $\myopt[\instance]{\numSignals}$ has an $(m+1)$-pooling $K$-signaling scheme solution $\signalProb^*$.  We then convert $\signalProb^*$ into a $\numSignals$-convex-partitional policy $\signalProb^\cc{convex}$ that satisfies $\designerExpU(\signalProb^\cc{convex}) \ge \frac{1}{m+1} \designerExpU(\signalProb^*)$.  

Let $[0, 1]^m = \bigcup_{i\in I} P_i$ be the convex partition structure of $\pi^*$, where $\pi^*$ sends at most $m+1$ signals for the states in each convex part $P_i$.  Consider each $P_i$.  If $\pi^*$ sends only one signal for all states in $P_i$, then we let $\cvxsignaling$ do the same. Otherwise, $\pi^*$ sends at least 2 and at most $m+1$ signals, denoted by $\{s_0, s_1, \ldots, s_m\}$, for the states in $P_i$.
Let $\mu_k = \frac{\int_{\theta \in P_i}\theta \cdot  \pi^*(s_k|\theta)\; \dd \stateCDF(\theta)}{\int_{\theta \in P_i} \pi^*(s_k | \theta)\; \dd \stateCDF(\theta)} \in P_i$
be the posterior mean induced by signal $s_k$. 
Without loss of generality, assume that $s_0$ is the signal with the highest expected utility, which means 
\begin{equation} \label{eq:s0-largest}
    \interimU(\mu_0) 
    \underbrace{\int_{\theta \in P_i} \pi^*(s_0|\theta)\; \dd \stateCDF(\theta)}_{\text{probability that $s_0$ is sent}} ~ 
    \ge ~ \interimU(\mu_k) \underbrace{\int_{\theta \in P_i}  \pi^*(s_k|\theta) \; 
    \dd \stateCDF(\theta)}_{\text{probability that $s_k$ is sent}}
    \qquad \forall k = 0, 1, ..., m. 
\end{equation}
Let 
\begin{equation*}
    \mathcal H_i ~ 
    = ~  \left \{ (w, c) \in \reals^{m+1} ~\Big|~ \tfrac{\int_{\theta \in P_i, w^\top \theta \le c} \theta\; \dd F(\theta)} {\int_{\theta \in P_i, w^\top \theta \le c}\; \dd F(\theta)} = \mu_0 \right\}
\end{equation*}
be the set of halfspaces, parameterized by $(w, c)$, whose intersection with $P_i$ has a center of mass (weighted by $F$) equal to $\mu_0$. 
Among all such halfspaces, let $H_i^*$ be the one whose intersection with $P_i$ is the largest (under measure $F$): 
\begin{equation*}
    H^*_i ~ := ~ \argmax_{(w, c) \in \mathcal H_i} \int_{\theta \in P_i, w^\top \theta \le c} \dd F(\theta) ~ = ~ \argmax_{\text{halfspace $H \in \mathcal H_i$}} F(P_i \cap H)\,. 
\end{equation*}
Let $\cvxsignaling$ send one signal for all states in $P_i \cap H^*_i$, another signal for all states in $P_i \setminus H^*_i$.  We do this for every convex part $P_i$ where $\pi^*$ sends at least 2 signals.
The resulting $\cvxsignaling$ is a convex-partitional information policy because $P_i \cap H^*_i$ and $P_i \setminus H^*_i$ are intersections of a convex set with halfspaces and hence convex. The total number of signals used by $\cvxsignaling$ is no more than the number of signals used by $\pi^*$, which is at most $K$. 

To prove that $\cvxsignaling$ achieves at least $\frac{1}{m+1}$ fraction of the expected utility of $\pi^*$, we prove it for each part $P_i$ separately. The expected utility of $\pi^*$ on $P_i$ is
\begin{equation*}
    U|_{P_i}(\pi^*) ~ = ~ \sum_{k=0}^m u(\mu_k) \int_{\theta \in P_i} \pi^*(s_k |\theta)\; \dd F(\theta). 
\end{equation*}
By definition, the conditional mean of the states in $P_i \cap H_i^*$ is $\mu_0$. And because the utility obtained from the signal for $P_i \setminus H_i^*$ is non-negative, the expected utility of $\cvxsignaling$ on $P_i$ is at least 
\begin{equation*}
    U|_{P_i}(\cvxsignaling) ~ \ge ~ u(\mu_0) \int_{\theta \in P_i \cap H_i^*} \dd F(\theta) ~ + ~ 0.  
\end{equation*}
We then show that the probability that posterior mean $\mu_0$ is induced under $\pi^\cc{convex}$, $\int_{\theta \in P_i \cap H_i^*} \dd F(\theta)$, is weakly larger than the probability under $\pi^*$, $\int_{\theta \in P_i} \pi^*(s_0 | \theta)\; \dd F(\theta)$: 
\begin{lemma}
\label{lemma:compare-two-probability}
$\int_{\theta \in P_i \cap H_i^*} \dd F(\theta) \ge \int_{\theta \in P_i} \pi^*(s_0 | \theta)\; \dd F(\theta)$. 
\end{lemma}
\begin{proof}[Proof sketch]
Consider the problem of maximizing the total probability of sending signal $s_0$, $\int_{\theta \in P_i} \pi(s_0 | \theta)\; \dd F(\theta)$, subject to the posterior mean constraint $\tfrac{\int_{\bm \theta \in P_i}\theta \cdot \pi(s_0|\theta)\; \dd \stateCDF(\theta)}{\int_{\theta \in P_i} \pi(s_0 | \theta)\; \dd \stateCDF(\theta)} = \mu_0$.
This is equivalent to finding a maximal-weight subset $R$ in region $P_i$ while ensuring that the $F$-weighted center of mass of $R$ equals $\mu_0$. This is an infinite-dimensional linear program. Via Lagrange duality analysis, the optimal subset $R^* \subseteq P_i$ should be the intersection of a halfspace (which is $H^*_i$) with $P_i$. 
See Appendix \ref{proof:compare-two-probability} for details.  
\end{proof}

Using \Cref{lemma:compare-two-probability}, we obtain
\begin{align*}
    U|_{P_i}(\cvxsignaling) & ~ \ge ~ u(\mu_0) \int_{\theta \in P_i} \pi^*(s_0 | \theta)\; \dd F(\theta) \\
    \text{by \eqref{eq:s0-largest}} & ~ \ge ~ \frac{1}{m+1} \sum_{k=0}^m u(\mu_k) \int_{\theta \in P_i} \pi^*(s_k | \theta)\; \dd F(\theta) ~ = ~ \frac{1}{m+1} U|_{P_i}(\pi^*). 
\end{align*}
Summing over $P_i$ gives $U(\cvxsignaling) \ge \frac{1}{m+1} U(\pi^*)$, therefore $\PoE^{\cc{convex}}(\mathcal C_{\mathrm{all}}, K) \ge \frac{1}{m+1}$. 

\subsubsection{Proof of \Cref{thm:convex-1/(m+1)}: Upper Bound Part} \label{sec:general-proof-upper-bound}
We then prove ``$\PoE^{\cc{convex}}(\mathcal C_{\mathrm{all}}, K) \le \frac{1}{m+1}$'' by constructing, for any $\eps > 0$, an instance $\instance = (\stateCDF, \interimU)$ %in the $m$-dimensional space
such that $\PoE^{\cc{convex}}(\instance, K) \le \frac{1}{m+1} + O(\eps)$ for any $K \ge m+1$.  

Let the prior $\stateCDF$ be a discrete distribution over $m+2$ points $x_0, x_1, \ldots, x_{m+1}$, where $x_1, \ldots, x_{m+1}$ are any $m+1$ affinely independent points in $[0, 1]^m$ (e.g., the $m$ standard basis vectors plus the origin), and $x_0$ is their unweighted average $\tfrac{1}{m+1} \sum\nolimits_{i=1}^{m+1} x_i$.
The prior probability mass is $f(x_i) = \tfrac{1}{m+1} - \eps$ for each $x_i$, $1 \le i\le m+1$, while the prior mass for $x_0$ is small: $f(x_0) = (m+1)\eps = O(\eps)$.
We consider such a discrete prior for simplicity of presentation; one can also construct a similar (but more complicated) instance with a continuous prior. 
The interim utility $\interimU$ is 1 only at the prior-weighted mean of $x_i$ and $x_0$, denoted by $\mu_i = \frac{f(x_i) x_i + f(x_0) x_0}{f(x_i) + f(x_0)} = (1 - O(\eps)) x_i + O(\eps) x_0 \approx x_i$, for $1\le i \le m+1$.  
The interim utility is $0$ everywhere else. 

To obtain positive utility, a signal has to induce a posterior mean at one of the $\mu_i$. In fact, the optimal $(K=m+1)$-signaling scheme can induce $m+1$ posterior means all at $\mu_i$, $1\le i \le m+1$, and hence achieve an expected utility of 1. 

\input{plots/fig3}

\begin{lemma} \label{lem:m-optimal-1}
In the above instance, the optimal $(K=m+1)$-signaling scheme $\pi^*$ achieves expected utility $U(\pi^*) = 1$. 
\end{lemma}
\begin{proof}[Proof sketch]
Consider the signaling scheme $\pi^*$ that, at each state $x_i$ for $1 \le i \le m+1$, sends signal $s_i$ with a large probability $q \approx 1$ and signal $s_j$ with a small probability $\frac{1-q}{m} \approx 0$ for every $j \in [m+1]\setminus \{i\}$; and at state $x_0$, sends a signal among $s_1, \ldots, s_{m+1}$ uniformly at random. By choosing $q$ appropriately, the posterior mean induced by each signal $s_i$ will be equal to $\mu_i$, hence obtaining utility $1$.  See \Cref{fig:m-optimal} for an illustration of the optimal signaling scheme, and see Appendix \ref{app:lem:m-optimal-1} for the detailed proof. 
\end{proof}

We then prove that the expected utility of any convex-partitional information policy $\cvxsignaling$ is at most $\frac{1}{m+1} + m\eps$:
\begin{lemma} \label{lem:m-convex-1/m+1}
For any convex-partitional %information policy
$\cvxsignaling$, $U(\cvxsignaling) \le \frac{1}{m+1} + m\eps$. 
\end{lemma}
\begin{proof}[Proof sketch]
See \Cref{fig:m-convex} for an illustration of a convex-partitional information policy. In order to induce a posterior mean at $\mu_i$, the policy has to send signal $s_i$ with positive probability at states $x_i$ and $x_0$, but not at any state $x_j$ with $j\notin\{0, i\}$.
As a result, the total probability that the policy sends signals inducing the desired posterior means is at most $\sum_{i=1}^{m+1} \frac{f(x_i)}{m+1} + f(x_0) \approx \frac{1}{m+1} + O(\eps)$. See Appendix \ref{app:lem:m-convex-1/m+1} for a detailed proof.  
\end{proof}

Lemmas \ref{lem:m-optimal-1} and \ref{lem:m-convex-1/m+1} imply that the ratio between the expected utilities of any $\pi^\cc{convex}$ and $\pi^*$ is $\frac{U(\cvxsignaling)}{U(\pi^*)} \le \frac{1}{m+1} + m\eps$. Letting $\eps \to 0$ proves the theorem.

\subsection{Price of Explainability of Rectangular Information Policies}
\label{subsec:high-dimension-rectangular}
Drawing on a connection with explainable clustering \citep{MDRF-20}, 
% a natural generalization in high-dimensional setting
another natural definition of explainable information policies in the high-dimensional case is those that partition the state space into axis-aligned rectangles and %assign a unique deterministic signal for all states in each rectangle.
reveal which rectangle contains the state. (In $m$ dimensions, by ``rectangle'' we mean a $m$-dimensional rectangular prism.)
% a natural generalization of explainable information policy when the state is high-dimensional is that the information policy partitions the state space using axis-aligned rectangles, and for all states in each rectangle, it deterministically sends a unique signal.
% Similar to \Cref{defn:poe}, we define the price of explainability $\myPoE{\instance, \numSignals}$ in the high-dimensional setting as the ratio between the designer's best payoff under an explainable policy with at most $\numSignals$ signals and the best payoff under an unrestricted policy with the same number of signals.
% in cost that is inevitable if we force our final clustering to have an interpretable form
% Let $\stateDim \ge 2$ be a positive integer. 
We then study the price of explainability $\PoE^{\cc{rect}}(\instance, \numSignals)$ of such rectangular information policies against the unrestricted information policies, with the same number of signals $K$.  
% optimal policy with the same number of signals $K$. 

% \begin{figure}[t!] % [htbp]
%     \centering
%     \input{Paper/plots/2-D}
%     \caption{Illustration of \Cref{ex:high-D}. Here the gray circles are the point masses of the prior probability distribution $\stateCDF$, and blue circles are the points where the interim utility function $\interimU$ is strictly positive.}
%     \label{fig:high-D}
% \end{figure}

% However, unlike in the one-dimensional setting -- where we can always bound the price of explainability by a constant factor of $\sfrac{1}{2}$ --
Unfortunately, unlike convex-partitional policies whose PoE is at least $1/(m+1)$, this guarantee fails
% in higher dimensions
for rectangular policies, 
% without any assumptions on the designer's prior $F$ or interim utility function $\interimU$.
% In particular, the price of explainability can become arbitrarily bad even in just two dimensions.
% we show that the price of explainability can be arbitrarily bad if we do now assume any Lipchitz assumption on designer's interim utility function $\interimU$.
whose $\PoE^{\cc{rect}}(\instance, \numSignals)$ can be arbitrarily close to \jamie{The proposition says exactly} $0$. % in some instances. 

\begin{proposition}
\label{prop:high-d no poe}
In the $(m=2)$-dimensional case, 
there exists an information design instance $\instance = (\stateCDF, \interimU)$ such that $\PoE^{\cc{rect}}(\instance, \numSignals = 2) = 0$.
% the prior state distribution $\stateCDF$ 
\end{proposition}

\begin{example}
\label{ex:high-D}
Let $\instance = (\stateCDF, \interimU)$ be an instance with a discrete prior supported on $4$ points:
%is discrete and has probability density function
$\statePDF((0, 0)) = \statePDF((0, 0.5)) = \statePDF((1, 0.5)) = \statePDF((1, 1)) = \frac{1}{4}$.
% and then it uniformly distributes the remaining probability mass $\varepsilon$ over whole $[0, 1]^2$ region. 
The interim utility function is % $\interimU$ satisfies 
$\interimU(\state) = \indicator{\state = (0.5, 0.25)} + \indicator{\state = (0.5, 0.75)}$.\footnote{%Note that in this instance, we consider prior distribution as a distribution with mass points.
It is straightforward to generalize this instance to allow continuous prior. } %support over the whole domain $[0, 1]^2$.}
\end{example}

% \begin{wrapfigure}{r}{0.46\textwidth}
\begin{figure}
\centering
\begin{tikzpicture}[>=stealth, line cap=round, thick, scale=0.75]
  %── 1) gray‐shaded square ─────────────────────────────────────────────
  \fill[gray!20] (0,0) rectangle (4,4);
  \draw          (0,0) rectangle (4,4);

  %── 2) four gray circles on the boundary ─────────────────────────────
  \foreach \X/\Y in {0/0, 0/2, 4/2, 4/4} {
    \draw[fill=gray!70] (\X,\Y) circle[radius=0.2];
  }

  %── 3) two blue circles inside ────────────────────────────────────────
  \draw[fill=blue!70] (2,3) circle[radius=0.15];
  \draw[fill=blue!70] (2,1) circle[radius=0.15];

  %── 4) dashed connections ─────────────────────────────────────────────
  % horizontal middle
  \draw[dashed] (0,2) -- (4,2);

  % diagonal from top‐blue to top‐right
  \draw[dashed] (0,2) -- (4,4);

  % diagonal from bottom‐blue to bottom‐left
  \draw[dashed] (0,0) -- (4,2);

\end{tikzpicture}
\caption{Illustration of \Cref{ex:high-D}. Here the gray circles are the point masses of the prior probability distribution $\stateCDF$, and blue circles are the points where the interim utility function $\interimU$ is strictly positive.}
\label{fig:high-D}
\end{figure}
% \end{wrapfigure} 

\begin{proof} [Proof of \Cref{prop:high-d no poe}]
% \wtcomment{I find this example is not so good, as it considers a non-continuous prior CDF, which violates our model primitives.}
Consider \Cref{ex:high-D}. % ; see \Cref{fig:high-D} for illustration.
Let $\numSignals = 2$.
%Under the instance in \Cref{ex:high-D},
We note that the optimal $\numSignals$-signaling scheme is to fully pool the state $(0, 0.5)$ and $(1, 1)$ to induce a posterior mean at $(0.5, 0.75)$, and fully pool the state $(0, 0)$ and $(1, 0.5)$ to induce a posterior mean at $(0.5, 0.25)$, which yields designer utility $1$, namely
$\myopt[\instance]{\numSignals} = 1$.
However, any axis-aligned rectangular information policy cannot induce the posterior means at $(0.5, 0.25)$ and $(0.5, 0.75)$.
Thus, $\OPT^{\cc{rect}}_{\instance}(\numSignals) = 0$ and $\PoE^{\cc{rect}}(\instance, \numSignals = 2) = 0$.
% We thus finish the proof.
\end{proof}

% The prior distribution in \Cref{ex:high-D} is unnatural, and its non-uniformity is essential to the proof.
% In contrast, when the prior is uniform, we are able to obtain a constant price of explainability in any fixed number of dimensions, with no assumptions on the utility function.
The prior distribution in \Cref{ex:high-D} is not uniform. 
In contrast, when the prior is uniform, we are able to obtain a bounded price of explainability, with no assumptions on the utility function. We first show that any $K$-convex-partitional information policy can be converted into a $K$-rectangular information policy achieving a constant approximation guarantee with respect to $K$ (see Appendix \ref{proof:rectangle-convex} for the proof).  

\begin{lemma}\label{lem:rectangle-convex}
	For any instance $\instance = (\stateCDF, \interimU)$ with uniform prior $F$ over $[0, 1]^m$, for any $K$-convex partitional information policy $\cvxsignaling$, there exists a $K$-rectangular information policy $\pi^{\cc{rect}}$ satisfying 
	%$$\myPoE{\instance, \numSignals} \geq \max_{x \geq 0} \frac{(2x)^m}{d 2^m (2x + 1)^m (m + 1)^m (x + 1)^m} \geq \frac{1}{d 2^m(3^m - 1)(m + 1)^m}.$$
	$\frac{U(\pi^{\cc{rect}})}{U(\cvxsignaling)} %\geq \frac{\lfloor K / (1 + 4^m) \rfloor}{K 2^m} \cdot \max_{x \geq 0} \frac{(2x)^m}{((2x + 1)^m - 1) (m + 1)^m (x + 1)^m} 
    \ge \frac{\lfloor K / (1 + 4^m) \rfloor}{K 2^m (3^m - 1)(m + 1)^{m}}$. % {\color{red} remeber to do the optimization for $x$ in the proof}
\end{lemma}

% \begin{proof}[Proof sketch of \Cref{lem:rectangle-convex}]
% The full proof is in Appendix \ref{proof:rectangle-convex}; we provide a sketch here. 
% To construct a $K$-rectangular policy to approximate the $K$-convex-partitional policy, we: 
% (i) keep only the most ``valuable'' convex parts (a top fraction by utility $\times$ volume) among the $K$ parts, 
% (ii) group them into a single orthant so that axis-aligned growth behaves monotonically, 
% (iii) take a {\em staircase closure} of what has been covered so far, 
% and (iv) run a greedy packing that, for each target centroid, either skips it (if it is too close to the existing staircase set) or adds a new disjoint rectangle with centroid whose corner touches the current boundary. 
% Two technical lemmas make this process work: a centroid-extent bound for convex sets (see \Cref{lemCentroidConvexSetBound}) to relate where mass can lie to the centroid coordinates, and a rectangular decomposition lemma (see \Cref{lemRectangleSubdivision}) to ensure the final rectangular partition uses at most $K$ pieces after carving out overlaps. 
% % {\color{red}
% % rectangle-approximation
% % factor plus the combinatorial losses from ``discarding'' parts and ``orthant selection,'' to obtain a PoE bound that is uniform across dimensions.
% % }
% \end{proof}

Then, we combine \Cref{lem:rectangle-convex} with the $1/(m+1)$ PoE result for $K$-convex-partitional policies in \Cref{thm:convex-1/(m+1)} to obtain the PoE result for $K$-rectangular policies: 
\begin{theorem}\label{thmUniformHigherDimension}
	For any instance $\instance = (\stateCDF, \interimU)$ with uniform prior $F$ over $[0, 1]^m$, for any positive integer $K$,
	%$$\myPoE{\instance, \numSignals} \geq \max_{x \geq 0} \frac{(2x)^m}{d 2^m (2x + 1)^m (m + 1)^m (x + 1)^m} \geq \frac{1}{d 2^m(3^m - 1)(m + 1)^m}.$$
	$\PoE^{\cc{rect}}(\instance, \numSignals) \geq \frac{\lfloor K / (1 + 4^m) \rfloor}{K 2^m (3^m - 1)(m + 1)^{m+1}}$. 
% \taocomment{I think we can omit the first inequality and put the maximization over $x$ to the proof. }
\end{theorem}

% Note that the lower bound we get for $m = 1$ is $\frac{1}{20}$, which is already much lower than the worst known price of explainability for the case of uniform priors, which is $\frac23$ from \Cref{lem13}. As $m$ increases, the lower bound vanishes super-exponentially; even for $m = 2$ the bound we get is $\frac{1}{4131} \approx 0.0002$ (assuming $K$ is divisible by $(1 + 4^m)$). 

While the above PoE lower bound is at a constant level with respect to $K$, it approaches $0$ exponentially as the dimension $m$ increases.
% , meaning that the explainable/rectangular information policy that we constructed in the proof of \Cref{thmUniformHigherDimension} might be significantly worse than the unrestricted optimal policy.  Can we construct a better explainable information policy?
Is this exponential bound improvable? 
Unfortunately, no. There exists an instance with uniform prior where no rectangular policy can approximate the unrestricted optimal policy by a factor better than $1/C^m$ for some $C > 1$, % , so the exponential PoE is unavoidable:
as shown by the following theorem (with the proof given in Appendix \ref{proof:high-dimension-negative}):
\begin{theorem} \label{thm:high-dimension-negative}
There exists an instance $\instance$ with uniform prior $F$ over $[0, 1]^m$ and binary utility function $u$ such that, for any $K\ge 2$, $\PoE^{\cc{rect}}(\instance, \numSignals) \le O\big( \sqrt{m} (\frac{2}{e} )^m \big) = O( 1/C^m)$.  
\end{theorem}

%\Cref{prop:high-d no poe} highlights that, in order to obtain a meaningful (i.e., non-trivial) bound on the price of explainability, one must either (i) relax the notion of explainable information policies -- e.g., by allowing non-axis-aligned partitions or incorporating randomization; or (ii) restrict attention to information design instances where the designer’s interim utility function satisfies certain regularity conditions, such as Lipschitz continuity. 
%We leave the exploration of those questions as interesting directions for future research.

\section{Computational Complexity of Explainable Information Design}
\label{sec:complexity}

We now turn to the computational complexity of explainable information design. We first establish that the problem of finding an exactly optimal explainable information policy is computationally intractable (NP-hard), even under mild assumptions on the interim utility function. We then show how to find an approximately optimal explainable information policy efficiently. We focus on the $(m=1)$-dimensional case in this section. 

\subsection{NP-Hardness of Computing Optimal Explainable Policies}
\label{subsec:hardness}

% This subsection presents our main hardness result for explainable information design: % as follows. 
\taoedit{Explainable (interval-partitional) information policies, despite practical attractiveness, are difficult to optimize over.   
Specifically, we prove a computational hardness result:}   
\begin{theorem}
    \label{thm:optimal-NP-hard}
    Computing the optimal $K$-partitional information policy is NP-hard, even for binary-valued or $L$-Lipschitz piecewise-linear utility functions for any $L$.
\end{theorem}

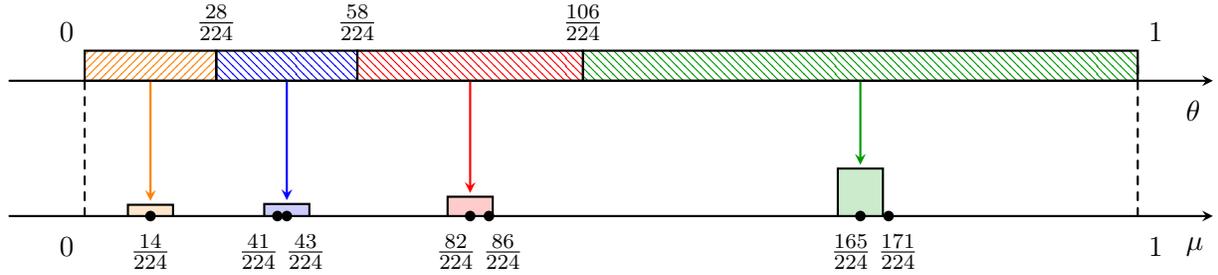
\begin{figure} 
    \centering
    % \begin{tikzpicture}[>=stealth, line cap=round, thick]
\begin{tikzpicture}[
  >=stealth,
  line cap=round,
  thick,
  scale=1,
  every node/.append style={font=\scriptsize}
]
  %--- parameters ------------------------------------------------------
  \def\stretchFactor{11}
  \def\L{1}      % left end of support
  \def\R{1+\stretchFactor}      % right end of support
  \def\Ytop{1.4}   % y-coordinate of top axis
  \def\Ybot{0}   % y-coordinate of bottom axis
  \pgfmathsetmacro{\Mone}{\L+28 * (\R-\L)/224} 
  \pgfmathsetmacro{\Mtwo}{\L+58 * (\R-\L)/224} 
  \pgfmathsetmacro{\Mthree}{\L+106 * (\R-\L)/224} 

  %--- hatched rectangle -----------------------------------------------
  \draw[pattern=north east lines,pattern color=orange, draw=black]
    (\L,\Ytop-0.2) rectangle (\Mone,\Ytop+0.2);
  
  \path[pattern=north west lines,pattern color=blue, draw=black]
    (\Mone,\Ytop-0.2) rectangle (\Mtwo,\Ytop+0.2);
  
  \path[pattern=north east lines,pattern color=red, draw=black]
    (\Mtwo,\Ytop-0.2) rectangle (\Mthree,\Ytop+0.2);

  \path[pattern=north west lines,pattern color=teal, draw=black]
    (\Mthree,\Ytop-0.2) rectangle (\R,\Ytop+0.2);

    \draw[->, orange] ({\L+14*(\R-\L)/224},\Ytop - 0.2) -- (1+\stretchFactor*14/224, \Ybot + 0.167);

    \draw[->, blue] ({\L+43*(\R-\L)/224},\Ytop - 0.2) -- (1+\stretchFactor*43/224, \Ybot + 0.176); 

    \draw[->, red] ({\L+82*(\R-\L)/224},\Ytop - 0.2) -- (1+\stretchFactor*82/224, \Ybot + 0.273); 
    
    \draw[->, teal] ({\L+165*(\R-\L)/224},\Ytop - 0.2) -- (1+\stretchFactor*165/224,\Ybot+0.653);
 
  \draw[->] (0.15,\Ytop-0.2) -- (\R+1,\Ytop- 0.2);     % top axis
  \draw[->] (0.15,\Ybot) -- (\R+1,\Ybot);     % bottom axis

  \draw[dashed] (\L,\Ybot) -- (\L,\Ytop);
  \draw[dashed] (\R,\Ybot) -- (\R,\Ytop);
  \node[left ] at (\L,\Ytop + 0.45) {0};
  \node[right] at (\R,\Ytop + 0.45) {1};
  \node[left ] at (\L,\Ybot - 0.3) {0};
  \node[right] at (\R,\Ybot - 0.3) {1};

  \node[below] at (1+\stretchFactor*14/224,\Ybot) {$\frac{14}{224}$};
  \node[below] at (1+\stretchFactor*47/224,\Ybot) {$\frac{43}{224}$};
  \node[below] at (1+\stretchFactor*79/224,\Ybot) {$\frac{82}{224}$};
  \node[below] at (1+\stretchFactor*163/224,\Ybot) {$\frac{165}{224}$};
  
  \node[below] at (1+\stretchFactor*37/224,\Ybot) {$\frac{41}{224}$};
  \node[below] at (1+\stretchFactor*89/224,\Ybot) {$\frac{86}{224}$};
  \node[below] at (1+\stretchFactor*173/224,\Ybot) {$\frac{171}{224}$};

  \node[above] at ({\L+28*(\R-\L)/224},\Ytop + 0.17) {$\frac{28}{224}$};
  \node[above] at ({\L+ 58*(\R-\L)/224},\Ytop + 0.17) {$\frac{58}{224}$};
  \node[above] at ({\L+ 106*(\R-\L)/224},\Ytop + 0.17) {$\frac{106}{224}$};

  \node[right] at (\R+0.5,\Ytop - 0.45) {$\state$};
  \node[right] at (\R+0.5,\Ybot - 0.32) {$\posteriorMean$};

  % Bottom bars
  \fill[orange!20] (1+\stretchFactor*14/224-0.3,\Ybot) rectangle (1+\stretchFactor*14/224 +0.3,\Ybot+0.15);
  \draw          (1+\stretchFactor*14/224-0.3,\Ybot) rectangle (1+\stretchFactor*14/224+0.3,\Ybot+0.15);
  
  \fill[blue!20] (1+\stretchFactor*43/224-0.3,\Ybot) rectangle (1+\stretchFactor*43/224 +0.3,\Ybot+0.1607);
  \draw          (1+\stretchFactor*43/224-0.3,\Ybot) rectangle (1+\stretchFactor*43/224+0.3,\Ybot+0.1607);
  
  \fill[red!20] (1+\stretchFactor*82/224-0.3,\Ybot) rectangle (1+\stretchFactor*82/224 +0.3,\Ybot+0.2571);
  \draw          (1+\stretchFactor*82/224-0.3,\Ybot) rectangle (1+\stretchFactor*82/224+0.3,\Ybot+0.2571);
  
  \fill[teal!20] (1+\stretchFactor*165/224-0.3,\Ybot) rectangle (1+\stretchFactor*165/224 +0.3,\Ybot+0.6321);
  \draw          (1+\stretchFactor*165/224-0.3,\Ybot) rectangle (1+\stretchFactor*165/224+0.3,\Ybot+0.6321);

  % Add points at the posterior means 1/3 and 2/3
  \fill ({\L+14*(\R-\L)/224}, \Ybot) circle[radius=2pt];
  \fill ({\L+43*(\R-\L)/224}, \Ybot) circle[radius=2pt];
  \fill ({\L+82*(\R-\L)/224}, \Ybot) circle[radius=2pt];
  \fill ({\L+165*(\R-\L)/224}, \Ybot) circle[radius=2pt];
  
  \fill ({\L+41*(\R-\L)/224}, \Ybot) circle[radius=2pt];
  \fill ({\L+86*(\R-\L)/224}, \Ybot) circle[radius=2pt];
  \fill ({\L+171*(\R-\L)/224}, \Ybot) circle[radius=2pt];
\end{tikzpicture}
    \caption{\label{figConnectedPartitionReduction} Illustration of the reduction from \textsc{Partition} with $c = (1, 2, 3)$. The black points at the bottom are the elements of the set $X$. The leftmost of these points must be chosen as the centerpoint of an interval, in addition to one point from each pair of the subsequent pairs, corresponding to choosing each $b_j \in \{-1, +1\}$. Once the set of centerpoints has been chosen, the $(n + 1)$-partitional information policy is uniquely determined, covering the entire interval $[0, 1]$ exactly if and only if $b$ is a solution to the \textsc{Partition} instance $c$.}
\end{figure}
    
% We first sketch the high-level proof idea. % \taocomment{Seems that this high-level proof idea is not finished yet.} \Cref{thm:optimal-NP-hard} is proved by a reduction from the \textsc{Partition} problem \citep{KarpOriginalReductions}:
\taoedit{\Cref{thm:optimal-NP-hard} is proved by a reduction from the \textsc{Partition} problem (defined below). A reduction is a standard way in theoretical computer science to prove computational hardness: we show how to transform any instance of the \textsc{Partition} problem into an instance of our optimal partitional-policy problem such that solving the latter would also solve the former. Since \textsc{Partition} is NP-hard \citep{KarpOriginalReductions}, this implies that there is no polynomial-time algorithm for computing an optimal partitional policy unless P = NP.}
        
\begin{definition}[The \textsc{Partition} problem]
    Given a vector of positive integers $c \in \zz_{> 0}^n$, % each encoded in binary, %. The goal is to decide
    decide whether there exists a vector $b \in \{-1, +1\}^n$ such that
	$\sum_{i = 1}^{n} b_i c_i = 0$. 
\end{definition}
    
%At a high level, we will
To perform the reduction, we transform a \textsc{Partition} problem instance to an information design instance with an interim utility function $\interimU$ that is non-zero in $2n+1$ posterior means in $[0, 1]$, corresponding to $+c_i$ and $-c_i$ for each $i\in\{1, \ldots, n\}$. The information design instance is constructed in a way that the signals from any information policy can only induce one of the two posterior means corresponding to $\pm c_i$, which means that we can only choose one sign for each $c_i$.
If the optimal information policy achieves expected utility $1$, then it must induce $n$ posterior means among the $2n$ points where $\interimU$ is non-zero.  This means that we can successfully assign a sign for every $c_i$ such that $\sum_{i=1}^n b_i c_i = 0$.  If the expected utility of the optimal information policy is less than $1$, then it induces less than $n$ posterior means among the $2n$ points, so the \textsc{Partition} problem instance is a NO instance. 

Specifically, %our reduction is as follows. Given
given any \textsc{Partition} problem instance $c = (\seq{c}{n})$, we let $d := \frac{1}{2^{n + 1}}$ and $T := 1 + \sum_{i = 1}^{n} c_i$, and then define
$$X := \left\{d\right\} \cup \left\{d\left(3 \cdot 2^{j - 1} - \frac{c_j}{2T}\right) \suchthat j \in [n]\right\} \cup \left\{d\left(3 \cdot 2^{j - 1} + \frac{c_j}{2T}\right) \suchthat j \in [n]\right\}.$$
%Note that $X$ can be easily computed and encoded with a polynomial number of bits as a function of the number of bits of the input vector $c$.
    Let the prior be $\stateCDF = \mathrm{Uniform}[0, 1]$.  Let the interim utility be the binary-valued function
	$u(x) = \indicator{x\in X}$.\footnote{Alternatively, we could construct an $L$-Lipschitz piecewise-linear utility function with $2\abs{X} = 4n + 2$ pieces sloping downward on either side from each point in $X$ with slope $L$.}
    % \twocases{\txt{if } x \in X}{1}{\txt{if } x \notin X}{0}$. 
    %In either case,
    In the full proof in Appendix \ref{proof:optimal-NP-hard}, we show that the designer can achieve expected utility 1 (meaning all posteriors are in $X$) with a partitional information policy if and only if the \textsc{Partition} problem instance $c$ has a solution; and furthermore, when it is possible, it can be done with $K = n + 1$ intervals.
    \Cref{figConnectedPartitionReduction} illustrates a simple instance produced by this reduction, explains the intuition of the reduction, and presents a 4-partitional information policy that achieves utility 1.

\subsection{Achieving $\frac{1}{2}\myopt[\instance]{\numSignals}$ for Piecewise Lipschitz Utility Functions}

The previous NP-hardness result means that one cannot hope to compute an optimal $\numSignals$-partitional information policy in all instances efficiently. 
However, since our PoE characterization in \Cref{thm:K-partitional-1/2} shows that the optimal $\numSignals$-partitional policy cannot be better than $1/2$ of the optimal unrestricted $\numSignals$-signaling scheme in the worst case, another reasonable and easier goal is to find a $\numSignals$-partitional policy $\pi_A$ that achieves this $1/2$ approximation guarantee, i.e., $U(\pi_A) \ge \frac{1}{2} \myopt[\instance]{\numSignals}$. We show that this goal is achievable (up to $\eps$ error) by a polynomial-time algorithm, for a large class of interim utility functions. 

We consider the class of possibly discontinuous piecewise Lipschitz interim utility functions.  A function $u$ is \emph{$(N, \Lip)$-piecewise Lipschitz} if there exist $0 = z_0 < z_1 < \cdots < z_N = 1$ such that $u$ is $\Lip$-Lipschitz continuous in every interval $(z_{i-1}, z_i)$.  In particular, piecewise Lipschitz functions include Lipschitz functions and piecewise constant functions. 

Our algorithm is based on the conversion idea in the proof of the PoE lower bound in \Cref{sec:convex-1/(m+1)-lower-bound}: first find an optimal $\numSignals$-signaling scheme $\signalProb^*$, then use the procedure in \Cref{sec:convex-1/(m+1)-lower-bound} to convert $\signalProb^*$ into a $\numSignals$-partitional information policy $\pi_A$ that achieves $1/2$ of the expected utility of $\signalProb^*$. 

The first step of finding an optimal $\pi^*$ admits a polynomial-time algorithm:  

\begin{lemma}[Proposition 3 of \citealp{arieli_optimal_2019}] 
For $(N, \Lip)$-piecewise Lipschitz utility function $\interimU$, an $\eps$-additively  approximately optimal $\numSignals$-signaling scheme $\signalProb^*$ can be computed in polynomial (in $N, L, 1/\eps$) time.
Moreover, $\pi^*$ is a bi-pooling policy with number of signals at most $NL/\eps+1$. \jamie{I feel like the $\varepsilon$ should be multiplied by something here... are you sure this is right?}
\end{lemma}

% According to \Cref{cor:optimal-bi-pooling}, the optimal $\numSignals$-signaling scheme $\signalProb^*$ is a bi-pooling policy.

The second step of converting $\pi^*$ to $\pi_A$ is described in \Cref{alg:conversion}; it is the conversion procedure in \Cref{sec:convex-1/(m+1)-lower-bound} specialized to the single-dimensional case.
As $\pi^*$ is bi-pooling, it divides the state space into multiple bi-pooling intervals $\{(b_{i-1}, b_i)\}_{i=1}^n$. 
We enumerate those intervals.
If $\pi^*$ sends only one signal for the interval, then let $\pi_A$ do the same.  If $(b_{i-1}, b_i)$ is a bi-pooling interval, then we further divide it into two parts $(b_{i-1}, c_i), (c_i, b_i)$, %, depending on which of the two bi-pooling signals has a higher utility.
using a carefully chosen cutoff $c_i$. 
One of these two parts achieves at least half of the expected utility of $\pi^*$ in interval $(b_{i-1}, b_i)$. Enumerating over all intervals, we thus obtain a $\pi_A$ that achieves an expected utility of at least $\frac{1}{2} U(\pi^*)$.

\begin{algorithm} [H]
\caption{Converting bi-pooling policy $\pi^*$ to interval-partitional policy $\pi_A$}
\label{alg:conversion}

% \SetKwInOut{Input}{Input}
% \SetKwInOut{Output}{Output}
% \SetKwInOut{Parameter}{Parameter}

% \Input {a bi-pooling policy $\signalProb^*$ whose total number of signals is $\numSignals$.} 
% \Parameter {}
% \Output {a $\numSignals$-interval-partitional information policy $\signalProb_\partionSet$.}

\DontPrintSemicolon
\LinesNumbered

Let $\{(b_{i-1}, b_i)\}_{i=1}^n$ be the bi-pooling interval structure of $\signalProb^*$. \; 

\For{$i=1$ to $n$}{
    % \If {$\signalProb^*$ sends one signal for interval $(b_{i-1}, b_i)$}
    \textbf{if} {\em $\signalProb^*$ sends one signal for $(b_{i-1}, b_i)$} \textbf{then}
    {let $\signalProb_\partionSet$ do the same for $(b_{i-1}, b_i)$.} \;
    \If {$\signalProb^*$ sends two signals for $(b_{i-1}, b_i)$}{
        Let $\posteriorMean_{i, 1} < \posteriorMean_{i, 2}$ be the posterior means induced by the two signals. \;
        Let $p_{i, 1}$,
        %$ = \frac{\int_{b_{i-1}}^{b_i} \signalProb^*(\signal_{i, 1}| \state)\statePDF(\state)}{\int_{b_{i-1}}^{b_i} \statePDF(\state)},
        $p_{i, 2}$ be the probabilities that these two signals are sent by $\pi^*$ conditioning on the state being in $(b_{i-1}, b_i)$. \;
        \If { $p_{i, 1} \interimU(\posteriorMean_{i, 1}) \, \ge \, p_{i, 2} \interimU(\posteriorMean_{i, 2})$}{
            Find $c_i \in (b_{i-1}, b_i)$ such that % with conditional mean % of the state in $(b_{i-1}, c_i)$ equals $\posteriorMean_{i, 1}$: 
            $\expect[\state \sim \stateCDF]{\state \smid (b_{i-1}, c_i)} := \frac{\int_{b_{i-1}}^{c_i} \state \statePDF(\state)\; \dd \state}{\int_{b_{i-1}}^{c_i} \statePDF(\state)\; \dd \state} = \posteriorMean_{i, 1}$. \;
        }\Else{
            Find $c_i \in (b_{i-1}, b_i)$ such that % with conditional mean of the state in $(c_i, b_i)$ equals $\posteriorMean_{i, 2}$: 
            $\expect[\state \sim \stateCDF]{\state \smid (c_i, b_i)} := \frac{\int_{c_i}^{b_i} \state \statePDF(\state)\; \dd \state}{\int_{c_i}^{b_i} \statePDF(\state)\; \dd \state} = \posteriorMean_{i, 2}$. \;
        }
        Divide $(b_{i-1}, b_i)$ into $(b_{i-1}, c_i)$ and $(c_i, b_i)$. Let $\signalProb_\partionSet$ send one signal for each of the two intervals. \;
    }
}
\Return{$\signalProb_\partionSet$}
\end{algorithm}

\begin{lemma}[corollary of the analysis in \Cref{sec:convex-1/(m+1)-lower-bound}] \label{lem:1/2-conversion}
For any bi-pooling policy $\signalProb^*$ with $\numSignals$ signals, \Cref{alg:conversion} returns a $\numSignals$-partitional policy $\signalProb_\partionSet$ satisfying $\designerExpU(\signalProb_\partionSet) \ge \frac{1}{2} \designerExpU(\signalProb^*)$. 
\end{lemma}

% We then show that the second step (converting $\pi^*$ to $\pi_A$) can also be done in polynomial time, under some oracle assumption. 

% Since $\pi^*$ is a bi-pooling policy, we can convert $\signalProb^*$ into a $\numSignals$-partitional information policy $\signalProb_\partionSet$ using \Cref{alg:conversion}.

A key step %(lines 9 and 11)
in \Cref{alg:conversion} is to solve the following problem: given $b \in (0, 1)$ and $\posteriorMean \in (b, 1)$, find a $c \in (\posteriorMean, 1)$ such that the conditional mean of the state $\expect[\stateCDF]{\state \mid \state \in (b, c)}$ equals $\posteriorMean$.  Due to the strict monotonicity of the conditional mean with respect to $c$, this problem can be solved efficiently by, e.g., a binary search.  % So, we assume an oracle that solves this problem exactly. Under this assumption, \Cref{alg:conversion} runs in polynomial time.
\Cref{alg:conversion} thus runs in polynomial time. 

% \begin{assumption} \label{ass:oracle}
% Assume an oracle that, given $b$ and $\mu$, can find a solution $c$ to $\expect[\stateCDF]{\state \mid \state \in (b, c)} = \mu$. 
% \end{assumption}

% \begin{lemma}
% Under \Cref{ass:oracle}, the running time of \Cref{alg:conversion}, which converts bi-pooling policy $\signalProb^*$ to $\numSignals$-partitional policy $\signalProb_\partionSet$, is polynomial in $\numSignals$. 
% \end{lemma}

As $\pi^*$ can be computed and converted to $\pi_A$ in polynomial time, and \Cref{lem:1/2-conversion} guarantees %that the converted $\numSignals$-partitional policy $\signalProb_\partionSet$ satisfies 
$\designerExpU(\signalProb_\partionSet) \ge \frac{1}{2} \designerExpU(\signalProb^*)$, we conclude that:  
\begin{theorem} \label{thm:1/2-piecewise-Lipschitz}
% Assume \ref{ass:oracle}. 
For $(N, \Lip)$-piecewise Lipschitz utility function $u$, for $K\ge \Theta(N\Lip/\eps)$, there exists a polynomial-time (in $N, \Lip, 1/\eps$) algorithm to compute a $\numSignals$-partitional policy $\signalProb_\partionSet$ that satisfies $\designerExpU(\signalProb_\partionSet) \ge \frac{1}{2} \myopt[\instance]{\numSignals} - \eps$. 
\end{theorem}

\section{Conclusions}

In this work, we introduce explainable information design and quantify the price of explainability in linear information design environments.
Inspired by the %growing
literature on explainable clustering, we view deterministic, monotone, partitional signaling policies as naturally explainable in the single-dimensional case, where we derive matching upper and lower bounds on the price of explainability for $K$-partitional policies against unrestricted $K$-signaling schemes:
% guarantees a tight $\sfrac{1}{2}$ fraction of the optimal unrestricted $K$-signaling performance, and under a uniform prior this guarantee can be improved to a tight $\sfrac{2}{3}$. 
tight bound of $\sfrac{1}{2}$ in general,
and tight bound of $\sfrac{2}{3}$ under uniform prior. 
We then extend the analysis to multi-dimensional state spaces by studying two notions of explainability: convex-partitional policies and axis-aligned rectangular policies. 
For convex-partitional policies, we prove a tight PoE of $\sfrac{1}{(m+1)}$, while for rectangular policies we establish a PoE guarantee under uniform prior that is independent of $K$ but %unavoidably
exponential in dimension $m$, with a lower bound showing that this exponential dependence in $m$ is inherent. 
On the computational complexity side,
% we prove that optimizing the best explainable policy is NP-hard in general,
% but provide efficient approximation methods, including an FPTAS for Lipschitz utilities and a polynomial-time algorithm that achieves the worst-case $\sfrac{1}{2}$ benchmark for broad discontinuous classes such as piecewise Lipschitz utilities. 
we prove that optimizing for the best explainable policy is NP-hard in general, but an explainable policy with  the tight $\sfrac{1}{2}$ approximation guarantee can be computed efficiently for the large class of piecewise Lipschitz utility functions. 

% Many interesting questions remain open for the future work, including identifying a reasonable yet tractable problem structure such that it admits constant (over the state dimensions and the number of signals both) price of explainability in high-dimensional setting; studying what is a good and tractable notion of the price of  explainability when the state space is discrete and when the designer utilities is beyond linear. 
Many interesting questions remain open for future work, including identifying broad yet tractable structural conditions under which the high-dimensional problem admits a constant price of explainability that is independent of both the state dimension and the number of signals; developing a useful and computationally tractable notion of the price of explainability when the state space is discrete; and extending our framework beyond linear designer utilities.

\appendix

\crefalias{section}{appendix}

\section{Omitted Proofs in \Cref{sec:price-of-explainable}}

\begin{claim}
\label{claim:uniform-capacity}
Under uniform prior on $[0,1]$, for any signal $s$ sent with total probability $p\in(0,1]$, its posterior mean $\mu_s$ satisfies $\frac{p}{2} \le \mu_s \le 1-\frac{p}{2}$. 
\end{claim}

\begin{proof}
% The posterior mean is $\mu_s = \frac{\int_0^1 \theta \pi(s\smid\state) \; \dd \state}{\int_0^1 \pi(s\smid \state) \; \dd \state}$. As the total probability is $\int_0^1 \pi(s\smid \state) \; \dd \state = p$, we have 
The posterior mean is 
\begin{equation*}
    \mu_s ~ = ~ \tfrac{\int_0^1 \theta \pi(s\smid\state) \; \dd \state}{\int_0^1 \pi(s\smid \state) \; \dd \state} ~=~ \frac{1}{p}\int_0^1 \theta \pi(s\smid\state) \; \dd \state ~\ge~ \frac{1}{p} \int_0^p \theta \cdot 1 \; \dd \state ~=~ \frac{p}{2}. 
\end{equation*}
Symmetrically, one can prove $\mu_s \le 1 - \frac{p}{2}$. 
\end{proof}

\subsection{Proof of \Cref{prop:uniform-binary-2/3}}
\label{proof:uniform-binary-2/3}
	For the lower bound, see \Cref{ex:continued}. For the upper bound, fix an instance $\mathcal{I} = (F, u)$ and let $X = \{ x \in [0, 1]: u(x) = 1\}$. % Note that $X$ is a closed set by our assumption that $u$ is upper semi-continuous.
    There are three cases to consider.
	
	\textbf{Case 1:}
    There is some point $x \in X \cap [\frac13, \frac23]$. Then consider the 2-partitional policy with one interval centered at $x$, expanded as large as possible, and the other interval its complement. The first interval has length at least $\frac23$, so the designer's expected utility is at least $\frac23$; whereas the optimal signaling scheme yields utility at most 1. Thus, $\PoE(\mathcal{I}, K) \leq \frac23$.
	
	\textbf{Case 2:} There are points $x \in X \cap [\frac16, \frac13]$ and $y \in X \cap [\frac23, \frac56]$. Then we can apply the same construction from Case 1 separately on the left half and right half of $[0, 1]$ and merge them together to get a 4-partitional policy attaining designer utility $\frac23$.
	
	\textbf{Case 3:} Neither Case 1 nor Case 2 occurs. This means that either $X \cap [\frac16, \frac23] = \emptyset$ or $X \cap [\frac13, \frac56] = \emptyset$. %These subcases are symmetric, so without loss of generality
    By symmetry, assume $X \cap [\frac16, \frac23] = \emptyset$. Let $u_1$ be the amount of designer utility in the optimal $K$-signaling scheme coming from signals with posterior means within $[0, \frac16]$ and let $u_2$ be the amount of designer utility coming from signals with posterior means within $[\frac23, 1]$. As there are no valuable posterior means in $(\frac{1}{6}, \frac{2}{3})$, the total optimal utility is precisely $u_1 + u_2$. Let $x$ be the maximum of $X \cap [0, \frac16]$ and $y$ be the minimum of $X \cap[\frac23, 1]$. Note that $2x \leq \frac13 \leq 2y - 1$. We claim that the partition $P = [0, 2x] \cup (2x, 2y - 1) \cup [2y - 1, 1]$
	attains utility at least $u_1 + u_2$, so that $\PoE(\mathcal{I}, K) = 1$.
	
	First consider the utility from the first interval of $P$, $[0, 2x]$. The mean of this interval is $x \in X$, so the designer's utility is the length of the interval, $2x$. We claim that $u_1 \leq 2x$. Let $S$ be the set of all signals in the optimal signaling scheme that induce posterior means in $X \cap [0, \frac16]$, which has total probability mass $u_1$. By the definition of $x$, each of these posterior means lies in $[0, x]$. If the designer were to pool all signals in $S$ into one new signal, its posterior mean would also lie in $[0, x]$.
    % If we were to shift the mass of the states at which signals in $S$ are sent toward the left as much as possible, then the mean could only decrease, and we would be left with an interval with density one from 0 to $u_1$, which has mean $u_1/2$.
    Since the total probability of the pooled signal is $u_1$, its posterior mean is at least $u_1/2$ by \Cref{claim:uniform-capacity}, thus $u_1/2 \leq x$, as desired. An analogous argument shows that $u_2 \leq 2(1 - y)$, which is the length of the third interval of $P$. %, centered at $y \in X$.
    Thus, these intervals collectively give a utility of at least $u_1 + u_2$.

\subsection{Proof of \Cref{lemma:uniform-prior-better-among-2}}
\label{proof:lemma:uniform-prior-better-among-2}
Recall that the bi-pooling interval is normalized to $[0,1]$. The total probability of the two signals are $p_1, p_2$. Their posterior means are $\mu_1, \mu_2$. We normalize the interim utility of signal 1 to be $u_1 = u(\mu_1) = 1$ and denote $r := u_2 / u_1 \in [0,1]$. 
The payoff of the original bi-pooling policy is $U(\pi^*) = p_1 + p_2 r$. 

Let $d:=\mu_2-\mu_1$. By Bayes plausibility $p_1 \mu_1 + p_2 \mu_2=\frac12$ and $p_1 + p_2 = 1$, we get 
\[
    \mu_1=\frac12 - p_2 d, \qquad \mu_2=\frac12 + p_1 d .
\]
By Claim~\ref{claim:uniform-capacity}, the signal inducing $\mu_1$ with
probability $p_1$ satisfies $p_1 \le 2\mu_1$, and the signal inducing $\mu_2$ with
probability $p_2$ satisfies $p_2 \le 2(1-\mu_2)$. Adding these two inequalities, 
\[
    1 ~=~ p_1 + p_2 ~\le~ 2\mu_1+2(1-\mu_2) ~=~ 2(1-d) \qquad \implies \qquad 
    d\le \frac12 .
\]

We now lower bound the payoffs of the two partitional policies $\pi_1$ and
$\pi_2$. Policy $\pi_1$ first takes the maximal interval
from the left with mean $\mu_1$, namely $[0,2\mu_1]$. If possible, it then takes
an interval with mean $\mu_2$. Therefore $U(\pi_1) \ge A := 2\mu_1 + \big(2\mu_2-4\mu_1\big)_+ r$,
where $(x)_+ := \max\{x, 0\}$. 
Policy $\pi_2$ takes the maximal interval from the right with mean $\mu_2$,
namely $[2\mu_2-1,1]$, and hence $U(\pi_2) \ge B := 2(1-\mu_2)r$. 
We will prove $\max\{A,B\} \ge \frac23(p_1 + p_2 r)$ by analyzing two cases: 

\textbf{Case 1: $d\le 1/3$.}
Using $\mu_1 = \frac12 - p_2 d$ and $\mu_2 = \frac12 + p_1 d$, we have
\[
    2\mu_1 ~=~ 1-2p_2d ~\ge~ 1-\frac{2p_2}{3} ~=~ \frac{1+2p_1}{3},
\]
and
\[
    2(1-\mu_2) ~=~ 1-2p_1 d ~\ge~ 1-\frac{2p_1}{3} ~=~ \frac{3-2p_1}{3}.
\]
If $A \ge \frac23(p_1+p_2 r)$ then the claim is proven. Otherwise, since $A\ge 2\mu_1$, we must have
\[
    \frac{1+2p_1}{3} ~ \le~ 2\mu_1 ~\le~ A ~<~ \frac23(p_1 + p_2 r) \quad \implies \quad 
    r>\frac{1}{2p_2}.
\]
Since $p_2 = 1-p_1$ and $4p_1(1-p_1)\le 1$, we have $\frac{1}{2p_2} = \frac{1}{2(1-p_1)} \ge 2p_1$. 
Thus $r> 2p_1$, and 
\[
    B ~=~ 2(1-\mu_2)r ~\ge~ \frac{3-2p_1}{3}r ~ > ~ \frac23(p_1 + p_2 r),
\]
where the last inequality follows %from $p_2 = 1-p_1$ and
$r > 2p_1$. Hence $\max\{A,B\}\ge \frac23(p_1 + p_2 r)$.

\textbf{Case 2: $d\ge 1/3$.}
We show that $A$ alone gives the desired bound. Since $d\le 1/2$, at $r=0$ we have $A = 2\mu_1 = 1-2p_2d \ge 1-p_2 = p_1 \ge \frac23 p_1$. 
At $r=1$, we have $A = 2\mu_1+(2\mu_2-4\mu_1)_+ \ge 2(\mu_2-\mu_1) = 2d \ge \frac23$. 
%because: if $2\mu_2-4\mu_1>0$ then $A=2\mu_2-2\mu_1=2d$, whereas if $2\mu_2-4\mu_1\le 0$ then $\mu_2\le 2\mu_1$, so $d=\mu_2-\mu_1\le \mu_1$ and therefore $A=2\mu_1\ge 2d$.
We note that both $A$ and $\frac23(p_1 + p_2 r)$ are affine functions of $r$. Since $A\ge \frac23(p_1 + p_2 r)$
holds at $r=0$ and at $r=1$, it holds for every $r\in[0,1]$. Thus $A\ge \frac23(p_1 + p_2 r)$. 

Combining the two cases, we obtain $\max\{U(\pi_1),U(\pi_2)\} \ge \max\{A,B\} \ge \frac23 U(\pi^*)$.

\section{Omitted Proofs in \Cref{sec:complexity}}

\subsection{Proof of \Cref{thm:optimal-NP-hard}}
\label{proof:optimal-NP-hard}
    To prove that the reduction is correct, we verify two directions.

    First, suppose $b$ is a solution to the \textsc{Partition} problem $c$; we need to construct a partitional policy $\pi_A$ achieving expected utility $1$. Let $d := \frac{1}{2^{n + 1}}$ and $T := 1 + \sum_{i = 1}^{n} c_i$.  Consider the $(n + 1)$-partitional information policy $\pi_{\partionSet}$ with partition points
	\begin{equation*}
		\partionSet := \Big\{d \Big(2^{j + 1} + \tfrac{(-1)^{j + 1}}{T} \sum \nolimits_{i = 1}^{j} b_i c_i \Big) \suchthat j \in \{0, 1, 2, \dots, n\}\Big\}.
	\end{equation*}
	Observe that these endpoints determine the following partition of $[0, 1]$:
    \begin{equation*} %\label{equDefinePartition}
		\mathcal{P} := \Big\{ [0, 2d] \Big\} \cup \left\{\left[d\Big(2^{j} + \tfrac{(-1)^{j}}{T} \sum \nolimits_{i = 1}^{j - 1} b_i c_i\Big), d\Big(2^{j + 1} + \tfrac{(-1)^{j + 1}}{T} \sum \nolimits_{i = 1}^{j} b_i c_i\Big)\right] \suchthat j \in [n]\right\}.
	\end{equation*}
    To see this, note that the first interval obviously has positive measure, and each other interval indexed by $j$ has length
	% \begin{align*}
	% 	& d\left(2^{j + 1} + \tfrac{(-1)^{j + 1}}{T} \Big(b_j c_j + \sum \nolimits_{i = 1}^{j - 1} b_i c_i\Big)\right) - d\Big(2^{j} + \tfrac{(-1)^{j}}{T} \sum \nolimits_{i = 1}^{j - 1} b_i c_i\Big)\\
	% 	=\hspace{.1cm} & d\left(2^j + \tfrac{(-1)^{j + 1}}{T} \Big(b_j c_j + 2\sum\nolimits_{i = 1}^{j - 1} b_i c_i\Big)\right) \\
	% 	\geq\hspace{.1cm} & d\left(2^j - \abs{\tfrac{(-1)^{j + 1}}{T} b_j c_j} - 2\sum\nolimits_{i = 1}^{j - 1} \abs{\tfrac{(-1)^{j + 1}}{T} b_i c_i}\right) \triineq\\
	% 	\geq\hspace{.1cm} & d\left(2^j - 2\abs{\frac{(-1)^{j + 1}}{T} b_j c_j} - 2\sum_{i = 1}^{j - 1} \abs{\frac{(-1)^{j + 1}}{T} b_i c_i}\right)\\
	% 	=\hspace{.1cm} & d \left(2^j - 2 \sum_{i = 1}^{j} \frac{c_i}{T}\right)\\
	% 	\geq\hspace{.1cm} & d \left(2 - 2 \sum_{i = 1}^{n} \frac{c_i}{T}\right)\\
	% 	=\hspace{.1cm} & 2d \left(1 - \frac{1}{T}\sum_{i = 1}^{n} c_i\right)\\
	% 	=\hspace{.1cm} & 2d \left(1 - \frac{T - 1}{T}\right)\\
	% 	>\hspace{.1cm} & 0.
	% \end{align*}
    \begin{align*}
		& d\left(2^{j + 1} + \tfrac{(-1)^{j + 1}}{T} \Big(b_j c_j + \sum \nolimits_{i = 1}^{j - 1} b_i c_i\Big)\right) - d\Big(2^{j} + \tfrac{(-1)^{j}}{T} \sum \nolimits_{i = 1}^{j - 1} b_i c_i\Big)\\
		% & = d\left(2^j + \tfrac{(-1)^{j + 1}}{T} \Big(b_j c_j + 2\sum\nolimits_{i = 1}^{j - 1} b_i c_i\Big)\right) \\
		& \geq d\left(2^j - \tfrac{1}{T} c_j - \tfrac{2}{T} \sum\nolimits_{i = 1}^{j - 1} c_i\right) \qquad \text{(because $|(-1)^{j+1} b_i | \le 1$)} \\
		% & \ge d \left(2^j - \tfrac{2}{T} \sum \nolimits_{i = 1}^{j} c_i \right)
        & \ge d \left(2 - \tfrac{2}{T} \sum \nolimits_{i = 1}^{n} c_i \right) ~ = ~ 2d \left(1 - \tfrac{T-1}{T} \right) ~ >  ~ 0.
	\end{align*}    
	Thus, each part is a well-defined interval of positive measure. To see that the intervals form a partition, first observe that the left endpoint of the interval for $j = 1$ is $2d$, which is the right endpoint of the first interval $[0, 2d]$. Similarly, for each $2 \leq j \leq n$, the left endpoint of the interval indexed by $j$ is
	$$d\left(2^{j} + \tfrac{(-1)^{j}}{T} \sum \nolimits_{i = 1}^{j - 1} b_i c_i\right) = d\left(2^{(j - 1) + 1} + \tfrac{(-1)^{(j - 1) + 1}}{T} \sum \nolimits_{i = 1}^{j - 1} b_i c_i\right),$$
	which is the right endpoint of interval $(j - 1)$. Finally, since $b$ is a solution to the \textsc{Partition} instance $c$, the rightmost endpoint of the very last interval is
	$$d\left(2^{n + 1} + \tfrac{(-1)^{n + 1}}{T} \sum\nolimits_{i = 1}^{n} b_i c_i\right) = d\left(2^{n + 1} + \tfrac{(-1)^{n + 1}}{T} \cdot 0\right) = d 2^{n + 1} = 1.$$
	Thus, $\mathcal{P}$ is a connected interval partition.
	
	The posterior for the first interval is $d \in X$. The posterior for any other interval at index $j$ is the average of the two endpoints,
	\begin{align}
		& \frac12\left(d\Big(2^{j} + \tfrac{(-1)^{j}}{T} \sum\nolimits_{i = 1}^{j - 1} b_i c_i\Big) + d\left(2^{j + 1} + \tfrac{(-1)^{j + 1}}{T} \Big(b_j c_j + \sum\nolimits_{i = 1}^{j - 1} b_i c_i\Big)\right)\right)\nonumber\\
		& = \frac{d}{2}\left(2^j + 2^{j + 1} + b_j(-1)^{j + 1} \tfrac{c_j}{T} + \tfrac{(-1)^{j}}{T} \sum \nolimits_{i = 1}^{j - 1} b_i c_i - \tfrac{(-1)^{j}}{T} \sum \nolimits_{i = 1}^{j - 1} b_i c_i\right) \nonumber \\ % \label{equEndpointDerivation}\\
		% & = \frac{d}{2}\left(3 \cdot 2^{j} + b_j(-1)^{j + 1} \tfrac{c_j}{T}\right)\nonumber\\
		& = d\left(3 \cdot 2^{j - 1} \pm \tfrac{c_j}{2T}\right),\nonumber
	\end{align}
	which is in $X$ as well. Thus, the expected utility of $\pi_A$ is 1.
	
	Conversely, suppose $\pi_A$ is a partitional policy with expected utility $1$; $\pi_A$ gives an interval partition $\mathcal{P}$ of $[0, 1]$ with all posteriors in $X$.
    We need to construct a solution $b$ to the \textsc{Partition} problem $c$.
    
    We claim that there exist $b_1, b_2, \dots, b_{n} \in \{-1, 1\}$ such that, for each $j \in [n]$, 
	\begin{equation*}
	    I(b_j, b_{j + 1}, \dots, b_n) := \left[d\left(2^{j} + \tfrac{(-1)^{n - j}}{T} \sum\nolimits_{i = j}^{n} b_i c_i\right), d\left(2^{j + 1} + \tfrac{(-1)^{n - j + 1}}{T} \sum\nolimits_{i = j + 1}^{n} b_i c_i\right)\right] 
	\end{equation*}
    belongs to $\mathcal{P}$.
	We prove this by induction on $j$ in the reverse direction, starting with the base case of $j = n$. In this case, observe that the point 1 can only be covered by an interval centered at $d\left(3 \cdot 2^{n - 1} \pm \frac{c_n}{2T}\right)$, since an interval centered at any other point in $X$ would have to extend into negative values, as all other points in $X$ are strictly less than $\frac12$.\footnote{For the first interval, we have $d = \frac{1}{2^{n + 1}} \leq \frac{1}{2^2} < \frac12$, and for any other interval indexed by $j' \leq n - 1$, we have $d\left(3 \cdot 2^{j' - 1} \pm \frac{c_j'}{2T}\right) < d\left(3 \cdot 2^{(n - 1) - 1} \pm \frac12\right) \leq d\left(3 \cdot 2^{(n - 1) - 1} + 2^{(n - 1) - 1}\right) = d 2^n = \frac12$. 
    }
	Letting $b_n \in \{-1, 1\}$ be such that the interval is centered at $d\left(3 \cdot 2^{n - 1} + b_n \frac{c_n}{2T}\right)$, the left endpoint is uniquely determined since the right endpoint is at 1. Specifically, % since the left endpoint of an interval centered at $a$ with right endpoint $1$ is $2a - 1$,
    this interval is
	\begin{align*}
		\left[2 d\left(3 \cdot 2^{n - 1} + b_n \tfrac{c_n}{2T}\right) - 1, 1\right] &= \left[d\left(3 \cdot 2^{n} + b_n \tfrac{c_n}{T} - \tfrac{1}{d}\right), 1\right]\\
		& = \left[d\left(3 \cdot 2^{n} + b_n \tfrac{c_n}{T} - 2^{n+1} \right), 1\right] = \left[d\left(2^{n} + b_n \tfrac{c_n}{T}\right), 1\right] 
		% &= \left[d\left(2^{n} + \frac{(-1)^{n - n}}{T} \sum_{i = n}^n b_i c_i\right), d \left(2^{n + 1} + \frac{(-1)^{n - n + 1}}{T} \sum_{i = n + 1}^{n} b_i c_i\right)\right]\\
		= I(b_n).
	\end{align*}
	Thus, we have shown that $I(b_n) \in \mathcal{P}$, proving the base case.
	
	The inductive case is quite analogous; we just have to keep track of the left endpoint of the previous interval coming from the inductive hypothesis. Fix any $j \in [n - 1]$ and suppose the claim holds for $j + 1$, i.e., there exist $b_{j + 1}, b_{j + 2}, \dots b_n \in \{0, 1\}$ such that $I(b_{j + 1}, b_{j + 2}, \dots, b_n) \in \mathcal{P}$. The leftmost endpoint of this interval is
	\[ x := d\left(2^{j + 1} + \tfrac{(-1)^{n - j + 1}}{T} \sum \nolimits_{i = j + 1}^{n} b_i c_i\right), \]
	which must be covered by an interval to the left. Note that any centerpoint of index greater than $j$ cannot work, as it lies to the right of $x$. As in the base case, the only valid choice is an interval centered at $d\left(3 \cdot 2^{j - 1} \pm \frac{c_n}{2T}\right)$, since an interval centered at any other point in $X$ would have to extend into negative values, as all other points in $X$ are strictly less than $\frac{x}{2}$. 
 %    : For the first interval, we have
	% \begin{align*}
	% 	d \leq d\frac{2^j}{2} \leq d\frac{2^{j + 1} - 1}{2} &< \frac{x}{2},
	% \end{align*}
	% and for any other interval indexed by $j' \leq j - 1$, we have
	% \begin{align*}
	% 	d\left(3 \cdot 2^{j' - 1} \pm \frac{c_j'}{2T}\right) &< d\left(3 \cdot 2^{(j - 1) - 1} \pm \frac{1 - \sum_{i = j + 1}^n b_i c_i}{2}\right)\\
	% 	&\leq d\left(3 \cdot 2^{(j - 1) - 1} + \frac12 - \frac{1}{2} \sum_{i = j + 1}^n b_i c_i\right)\\
	% 	&\leq d\left(3 \cdot 2^{(j - 1) - 1} + 2^{(j - 1) - 1} - \frac{1}{2} \sum_{i = j + 1}^n b_i c_i\right)\\
	% 	&= d\left(2^j - \frac{1}{2} \sum_{i = j + 1}^n b_i c_i\right)\\
	% 	&\leq d\left(2^{j} + \frac{(-1)^{n - j + 1}}{2T} \sum_{i = j + 1}^{n} b_i c_i\right)\\
	% 	&= \frac{x}{2}.
	% \end{align*}
    
	Letting $b_j \in \{-1, 1\}$ be such that the interval is centered at $d\left(3 \cdot 2^{j - 1} + (-1)^{n - j}b_j \frac{c_j}{2T}\right)$, the left endpoint is uniquely determined since the right endpoint is $x$. Specifically, this interval is
	\begin{align*}
		& \left[2 d\left(3 \cdot 2^{j - 1} + (-1)^{n - j}b_j \tfrac{c_j}{2T}\right) - x, ~ x\right]\\
		% & = \left[d\left(3 \cdot 2^{j} + (-1)^{n - j}b_j \tfrac{c_j}{T} - \tfrac{x}{d}\right), ~ x\right]\\
		& = \left[d\left(3 \cdot 2^{j} + (-1)^{n - j}b_j \tfrac{c_j}{T} - 2^{j + 1} - \tfrac{(-1)^{n - j + 1}}{T} \sum \nolimits_{i = j + 1}^{n} b_i c_i\right), ~ x\right]\\
		% & = \left[d\left(3 \cdot 2^{j} + (-1)^{n - j}b_j \tfrac{c_j}{T} - 2 \cdot 2^{j} - \tfrac{(-1)^{n - j + 1}}{T} \sum\nolimits_{i = j + 1}^{n} b_i c_i\right), ~ x\right]\\
		% & = \left[d\left(2^{j} + (-1)^{n - j}b_j \tfrac{c_j}{T} + \tfrac{(-1)^{n - j}}{T} \sum \nolimits_{i = j + 1}^{n} b_i c_i\right), ~ x\right]\\
		& = \left[d\left(2^{j} + \tfrac{(-1)^{n - j}}{T} \sum \nolimits_{i = j}^n b_i c_i\right), ~ d \left(2^{j + 1} + \tfrac{(-1)^{n - j + 1}}{T} \sum \nolimits_{i = j + 1}^{n} b_i c_i\right)\right]\\
		& = I(b_j, b_{j + 1}, \dots, b_n).
	\end{align*}
	Thus, we have shown that $I(b_j, b_{j + 1}, \dots, b_n) \in \mathcal{P}$, proving the inductive case.
	
	By induction, it follows that $I(b_j, b_{j + 1}, \dots, b_n) \in \mathcal{P}$ for all $j \in [n]$. Specializing $j = 1$, we have that $\mathcal{P}$ contains interval $I(\seq{b}{n})$, which has left endpoint \[d\left(2 + \tfrac{(-1)^{n - 1}}{T} \sum \nolimits_{i = 1}^{n} b_i c_i\right).
    \]
	Note that this is greater than zero, so the point must be covered by some interval to the left of $I(\seq{b}{n})$. In fact, the only way to cover this point is by the interval centered at $d$. But the left endpoint of this interval is zero, so the right endpoint must be $2d$, so the above expression must equal $2d$, otherwise the intervals do not exactly touch. This is equivalent to $\sum_{i = 1}^{n} b_i c_i = 0$, so $b$ is a solution to the \textsc{Partition} problem $c$.

% \subsection{Proof of \Cref{lem:V(A)-Lipschitz}}
% \label{proof:V(A)-Lipschitz}
% The partial derivative of $V(\partionSet)$ with respect to $a_i$ is 
% \begin{equation*}
%     \frac{\partial V(\partionSet)}{\partial a_i} = \statePDF(a_i) \Big[ \interimU(\posteriorMean_{[a_{i-1}, a_i]}) + \interimU'(\posteriorMean_{[a_{i-1}, a_i]}) (a_i - \posteriorMean_{[a_{i-1}, a_i]}) - \interimU(\posteriorMean_{[a_i, a_{i+1}]}) - \interimU'(\posteriorMean_{[a_i, a_{i+1}]}) (a_i - \posteriorMean_{[a_i, a_{i+1}]}) \Big]. 
% \end{equation*}
% Given \Cref{ass:Lipschitz-utility}, we have
% \[ \Big| \frac{\partial V(A)}{\partial a_i} \Big| \le 2 \fbound (\ubound + \Lip \cdot 1)\,,\] 
% so $V(A)$ is $2 \fbound (\ubound + \Lip)$-Lipschitz with respect to each input component $a_i$.  This implies that $V(\partionSet)$ is $2 \fbound (\ubound + \Lip)$-Lipschitz with respect to $\partionSet$ under the $\ell_1$-distance $\sum_{i=1}^{K-1} |a_i - a_i'|$ between $\partionSet$ and $\partionSet'$. 

\section{Omitted Proofs in \Cref{sec:high-dimension}}

\subsection{Proof of \Cref{lem:extreme-point-solution}}
\label{proof:extreme-point-solution}
% Because the objective function $\designerExpU(G) = \int_0^1 \interimU(\posteriorMean) \, \dd G(\posteriorMean)$ is continuous, in fact, linear, in $G$, and the feasible region $\mpc{\stateCDF, \numSignals}$ is compact, a solution to $\myopt[\instance]{\numSignals} = \max_{G \in \mpc{\stateCDF, \numSignals}} \designerExpU(G)$ must exist. 

Because $u$ is upper semicontinuous, the objective function $\designerExpU(G) = \int_{[0, 1]^m} \interimU(\posteriorMean) \, \dd G(\posteriorMean)$ is upper semicontinuous in the probability measure $G$ under the weak topology (by Portmanteau theorem). 
Since the feasible region $\MPC(F, K)$ is a compact set, a solution to $\myopt[\instance]{\numSignals} = \max_{G \in \mpc{\stateCDF, \numSignals}} \designerExpU(G)$ exists. 
Let
\begin{align*}
    M^* 
    & = \big\{ G^* \in \mpc{\stateCDF, \numSignals} \mid \designerExpU(G^*) 
    = \myopt[\instance]{\numSignals} \big\} \\
    & = \big\{ G^* \in \mpc{\stateCDF, \numSignals} \mid \designerExpU(G^*) \ge \myopt[\instance]{\numSignals} \big\} 
\end{align*}
be the set of solutions. 
Because $\designerExpU(G)$ is upper semicontinuous, $M^*$ as the preimage of a closed superlevel set is a closed subset of $\mpc{\stateCDF, \numSignals}$, and hence compact.  The Krein-Milman theorem states that the set of extreme points of a compact (but not necessarily convex) subset $M^*$ of a Hausdorff locally convex topological vector space (the space of finite signed measures on $[0, 1]$ in our case) has the same closed convex hull as $M^*$.
So, $M^*$ must contain at least one extreme point, denoted by $G^*$.

We then prove that $G^*$ is an extreme point of $\mpc{\stateCDF, \numSignals}$. $G^*$ cannot be written as the convex combination of two points in $M^*$ because $G^*$ is an extreme point of $M^*$. If $G^*$ can be written as the convex combination
of two points $G_1$ and $G_2 $ in $\mpc{\stateCDF, \numSignals}$ where one of them is not in $M^*$, then by the linearity of $U$, we have 
% \begin{equation*}
%     \designerExpU(G^*) = \lambda \designerExpU(G_1) + (1-\lambda) \designerExpU(G_2), \quad \forall \lambda \in (0, 1).
% \end{equation*}
$\designerExpU(G^*) = \lambda \designerExpU(G_1) + (1-\lambda) \designerExpU(G_2)$ for any $\lambda \in (0, 1)$. 
Say $G_1 \in \mpc{\stateCDF, \numSignals}$ and $G_2 \in \mpc{\stateCDF, \numSignals} \setminus M^*$, so $\designerExpU(G_1) \le \designerExpU(G^*)$ and $\designerExpU(G_2) < \designerExpU(G^*)$. Then, $\lambda \in (0, 1)$ implies $\lambda \designerExpU(G_1) + (1-\lambda) \designerExpU(G_2)  < \designerExpU(G^*)$, %which leads to
a contradiction. Therefore, $G^*$ as an optimal solution to $\myopt[\instance]{\numSignals}$ is an extreme point of $\mpc{\stateCDF, \numSignals}$.

\subsection{Proof of \Cref{lemma:compare-two-probability}}
\label{proof:compare-two-probability}
Consider the infinite-dimensional linear program with variable $\pi(s_0|\theta), \forall \theta \in P_i$:  
\begin{align*}
    \max_{0 \le \pi(s_0 | \cdot) \le 1} \quad & \int_{\theta \in P_i} \pi(s_0 | \theta)\; \dd F(\theta) \\
    \text{s.t.} \qquad & \tfrac{\int_{\theta \in P_i}\theta \cdot  \pi(s_0 | \theta)\; \dd \stateCDF(\theta)}{\int_{\theta \in P_i} \pi(s_0 | \theta)\; \dd \stateCDF(\theta)} = \mu_0 \quad \iff \quad \int_{\theta \in P_i} (\theta - \mu_0) \cdot  \pi(s_0 | \theta)\; \dd \stateCDF(\theta) = \bm 0. % \\ & 0 \le \pi(s_0 | \theta) \le 1, \quad \forall \theta \in P_i. 
\end{align*}
 By definition, $\pi^*$ is a feasible solution to this program. We then prove that the indicator function of the halfspace $H_i^*$, $\pi(s_0|\theta) = I_{\{\theta \in P_i\cap H_i^*\}}$, is an optimal solution to this linear program, which will prove the claim $\int_{\theta \in P_i \cap H_i^*} \dd F(\theta) \ge \int_{\theta \in P_i} \pi^*(s_0 | \theta)\; \dd F(\theta)$. 

%To prove that the indicator function of $H_i^*$ is an optimal solution, we note that
The Lagrangian with only the first constraint is $\int_{P_i} \big( 1 + \lambda^\top (\theta - \mu_0) \big) \pi(s_0 | \theta)\; \dd F(\theta)$,
with Lagrange multipliers $\lambda \in \reals^m$. Fixing $\lambda$, the maximizer over $\pi(s_0|\theta) \in [0, 1]$ is:
\begin{equation*}
    \pi_{\lambda}(s_0| \theta) =
        ~ 1 ~ \text{ if } \lambda^\top ( \theta - \mu_0) > -1 \quad \text{and} \quad 
        0 ~ \text{ if } \lambda^\top (\theta - \mu_0) < -1. 
\end{equation*}
The case of $\lambda^\top (\theta - \mu_0) = -1$ does not matter because such $\theta$ has $F$-measure $0$. Therefore, the optimal solution must be the indicator function of a halfspace, $H_{\lambda} = \big\{ \theta : \lambda^\top (\theta - \mu_0) \ge -1 \big\}$. 
Recall that $H_i^*$ is a halfspace that satisfies the constraint of the linear program and its objective value $\int_{P_i} 1_{\{\bm \theta \in H_i^*\}} \dd F(\theta)$ is maximal among all such halfspaces. So, the indicator function of $H_i^*$ is an optimal solution to the linear program.

\subsection{Proof of \Cref{lem:m-optimal-1}}
\label{app:lem:m-optimal-1}

Recall that we have an instance where the prior distribution $F$ is supported on $m+1$ affinely independent points $x_1, \ldots, x_{m+1}$ in $[0, 1]^m$ and $x_0 ~ = ~ \tfrac{1}{m+1} \sum\nolimits_{i=1}^{m+1} x_i$, 
% We focus on such a discrete prior for simplicity of presentation; one can also construct a similar (but more complicated) instance with a continuous prior.
with mass $f(x_i) = p = \tfrac{1}{m+1} - \eps$ for $i \in [m+1]$ and $f(x_0) = 1 - (m+1) p = (m+1)\eps$. For each $i \in [m+1]$, $\mu_i$ is the mean of $x_i$ and $x_0$ weighted by the prior mass:  
\begin{equation*}
    \mu_i ~ = ~ \frac{f(x_i) x_i + f(x_0) x_0}{f(x_i) + f(x_0)}  ~ = ~ \frac{p x_i + (1-(m+1)p) x_0}{1 - mp}. 
\end{equation*}
The interim utility $\interimU$ is $1$ on $\{\mu_1, \ldots, \mu_{m+1}\}$ and $0$ everywhere else.  
We will prove that the optimal signaling scheme $\pi^*$ has expected utility $U(\pi^*) = 1$. 

Consider the signaling scheme $\pi^*$ that sends signals $s_1, \ldots, s_{m+1}$ with the following conditional probability: $\forall i \in [m+1]$, 
\begin{equation*}
    \pi^*(s_i \smid x_i) = q \in (0, 1), \quad
    \pi^*(s_i \smid x_0) = \tfrac{1}{m+1}, \quad 
    \pi^*(s_i \smid x_j) = \tfrac{1-q}{m}, ~~ \forall j \in [m+1]\setminus\{i\}.
\end{equation*}
Note that $\sum_{i=1}^{m+1} \pi^*(s_i \smid x_k) = 1$ for any $k \in \{0, 1, \ldots, m+1\}$, so the signaling scheme is valid.  We will find a $q \in (0, 1)$ such that the posterior mean induced by each signal $s_i$ is equal to $\mu_i$. The posterior mean induced by $s_i$ is
\begin{align*}
    \mu_{s_i, \pi^*} ~ & = ~ \frac{f(x_i) q \cdot x_i ~ + ~ f(x_0) \tfrac{1}{m+1} \cdot x_0 ~ + ~ \sum\nolimits_{j \in [m+1]\setminus\{i\}} f(x_j) \tfrac{1-q}{m} \cdot x_j}{f(x_i) q ~ + ~ f(x_0) \tfrac{1}{m+1} ~ + ~ \sum\nolimits_{j \in [m+1]\setminus\{i\}} f(x_j) \tfrac{1-q}{m}} \\
    & = ~\frac{p q \cdot x_i ~ + ~ \tfrac{1-(m+1)p}{m+1} \cdot x_0 ~ + ~ \sum\nolimits_{j \in [m+1]\setminus\{i\}} \tfrac{p(1-q)}{m} \cdot x_j}{p q ~ + ~ \tfrac{1-(m+1)p}{m+1} ~ + ~ \tfrac{p(1-q)}{m} \cdot m} \\
    % & = ~ \frac{p q \cdot x_i ~ + ~ \frac{1-(m+1)p}{m+1} \cdot x_0 ~ + ~ \frac{p(1-q)}{m} \sum_{j \in [m+1]\setminus\{i\}} x_j}{p q ~ + ~ \frac{1}{m+1} - p ~ + ~ p - pq} \\
    & = ~ (m+1) p q x_i ~ + ~ (1-(m+1)p) x_0 ~ + ~ \tfrac{(m+1) p(1-q)}{m} \sum \nolimits_{j \in [m+1]\setminus\{i\}} x_j\,.
\end{align*}
From $x_0 = \tfrac{1}{m+1} \sum \nolimits_{i=1}^{m+1} x_i$ we have $\sum_{j \in [m+1]\setminus \{i\}} x_j = (m+1) x_0 - x_i$. So, 
\begin{align*}
    \mu_{s_i, \pi^*} ~ & = ~ (m+1) p q x_i ~ + ~ (1-(m+1)p) x_0 ~ + ~ \tfrac{(m+1) p(1-q)}{m} \Big( (m+1) x_0 - x_i \Big) \\
    & = ~ (m+1) p \Big( q - \tfrac{1-q}{m}\Big) x_i ~ + ~ \Big( \tfrac{(m+1)^2 p (1-q)}{m} - (m+1) p + 1 \Big) x_0 \,.
\end{align*}
Equating the above with $\mu_i = \frac{p x_i + (1-(m+1)p) x_0}{1 - mp}$
and comparing the coefficients of $x_i$ and $x_0$, we obtain the following two equations:
\begin{equation*}
    (m+1) p \Big( q - \tfrac{1-q}{m}\Big)  = \tfrac{p}{1-mp}, \qquad \tfrac{(m+1)^2 p (1-q)}{m} - (m+1) p + 1 = \tfrac{1-(m+1)p}{1-mp}, 
\end{equation*}
which have solution $q = \frac{2m+1 - m(m+1) p}{(m+1)^2 (1-mp)}$.
One can verify that $q \in (0, 1)$ for $p = \frac{1}{m+1} - \eps < \frac{1}{m+1}$. So, $q$ is valid solution that ensures each posterior mean $\mu_{s_i, \pi^*}$ to be equal to $\mu_i$, where the interim utility is $1$, so the expected utility of $\pi^*$ is $1$.

\subsection{Proof of \Cref{lem:m-convex-1/m+1}}
\label{app:lem:m-convex-1/m+1}

Let $p:=f(x_i)=\frac{1}{m+1}-\eps$ for $1\le i \le m+1$ and $r:=f(x_0) = (m+1)\eps$. Then 
\[
    x_0=\frac{1}{m+1}\sum \nolimits_{j=1}^{m+1}x_j,
    \qquad
    \mu_i=\frac{p x_i+r x_0}{p+r}.
\]
The interim utility is equal to $1$ at $\mu_1,\dots,\mu_{m+1}$ and equal to $0$ elsewhere. Hence only signals inducing one of the posterior means $\mu_i$ can
contribute to the designer's payoff.

Consider any signal $s_\ell$ of a convex-partitional policy, and suppose that it
induces posterior mean $\mu_i$. Let
\[
    \alpha_k:=\cvxsignaling(s_\ell\mid x_k)
    \quad \text{for } k=0,1,\dots,m+1,
\]
and let $q_k$ be the posterior probability of state $x_k$ conditioning on signal
$s_\ell$, namely 
\[
    q_0=\frac{\alpha_0 r}{\Pr(s_\ell)},
    \qquad
    q_j=\frac{\alpha_j p}{\Pr(s_\ell)}
    \quad \text{for } j=1,\dots,m+1.
\]

We first prove the following structural claim:
\begin{claim}\label{claim:structural}
    If $s_\ell$ induces $\mu_i$, then $s_\ell$ is sent only at $x_i$ and $x_0$, and moreover $\alpha_i=\alpha_0$. 
\end{claim}
\begin{proof}
Let $\lambda:=\frac{p}{p+r}$ and $\beta:=\frac{r}{p+r}$. 
%Using the barycentric coordinates with respect to the affinely independent vertices $x_1,\dots,x_{m+1}$, we have
Write $\mu_i$ as a convex combination of $x_1, \ldots, x_{m+1}$:  
\[
    \mu_i ~ =~ \lambda x_i+\beta x_0 ~= ~ \big( \lambda + \tfrac{\beta}{m+1} \big) x_i + \sum\nolimits_{1\le j\le m+1, j\ne i}\tfrac{\beta}{m+1} x_j.
\]
% and therefore the barycentric coordinates of $\mu_i$ are
% \[
%     \lambda+\frac{\beta}{m+1} \quad \text{on } x_i, \qquad \frac{\beta}{m+1} \quad \text{on each } x_j,\ j\neq i.
% \]
On the other hand, the posterior distribution $(q_0,q_1,\dots,q_{m+1})$ induces posterior mean
\[
    \mu_i ~=~ q_0x_0+\sum \nolimits_{j=1}^{m+1}q_jx_j ~ = ~ \big(q_i + \tfrac{q_0}{m+1}\big) x_i + \sum\nolimits_{1\le j\le m+1, j\ne i} \big( q_j + \tfrac{q_0}{m+1} \big) x_j. \]
% By uniqueness of barycentric coordinates,
Because $x_1, \ldots, x_{m+1}$ are affinely independent, the convex decomposition of $\mu_i$ is unique, so we have: % we have % we have: 
\[
    q_i+\tfrac{q_0}{m+1} = \lambda+\tfrac{\beta}{m+1}, \qquad q_j+\tfrac{q_0}{m+1} = \tfrac{\beta}{m+1} \quad \text{ for every } j\neq i. 
\]
% Equivalently,
% \[
%     \qquad
%     q_i=\lambda+\tfrac{\beta-q_0}{m+1}, \qquad q_j=\tfrac{\beta-q_0}{m+1}
%     \quad \text{for every } j\neq i. 
% \]

Thus either $q_j=0$ for all $j\neq i$, or $q_j>0$ for all $j\neq i$. Suppose,
toward a contradiction, that $q_j>0$ for some, equivalently all, $j\neq i$.
Then the convex cell corresponding to $s_\ell$ contains all of
$x_1,\dots,x_{m+1}$. Since $x_0$ lies in the interior of their convex hull, $x_0$
is an interior point of this cell. Therefore the convex-partitional policy must send
$s_\ell$ at $x_0$ with probability one, i.e. $\alpha_0=1$.
But then
\[
    \frac{q_i}{q_0} ~=~ \frac{\alpha_i p}{\alpha_0 r} ~\le~ \frac{p}{r}.
\]
On the other hand, we note that $q_j > 0$ and $q_j = \frac{\beta-q_0}{m+1}$ implies: 
\begin{equation*}
    q_0 < \beta \qquad \text{ and } \qquad q_i ~=~ \lambda + \tfrac{\beta - q_0}{m+1} ~>~ \lambda. 
\end{equation*}
Thus, we have $\frac{q_i}{q_0} > \frac{\lambda}{\beta} = \frac{p}{r}$, which leads to a contradiction. 
Therefore $q_j=0$ for all $j\neq i$.

It follows that $q_0=\beta$ and $q_i=\lambda$. Hence $\frac{\alpha_i p}{\alpha_0 r} = \frac{q_i}{q_0} = \frac{\lambda}{\beta} = \frac{p}{r}$, 
so $\alpha_i=\alpha_0$.
\end{proof}

We now bound the expected utility of $\cvxsignaling$. Let $\mathcal L_i$ be the set of signals that
induce posterior mean $\mu_i$. By \Cref{claim:structural}, for any
$\ell\in\mathcal L_i$, $\Pr(s_\ell) = \alpha_i p+\alpha_0 r = \alpha_0(p+r)$. 
Thus
\[
\begin{aligned}
    U(\cvxsignaling) & = \sum\nolimits_{i=1}^{m+1}\sum\nolimits_{\ell\in\mathcal L_i}\Pr(s_\ell) \\
    & = (p+r)\sum\nolimits_{i=1}^{m+1}\sum\nolimits_{\ell\in\mathcal L_i}\alpha_0 ~ \le ~ (p+r)\sum\nolimits_{\ell}\cvxsignaling(s_\ell\mid x_0) ~ = ~ p+r.
\end{aligned}
\]
By definition, $p+r = \left(\frac{1}{m+1}-\eps\right)+(m+1)\eps = \frac{1}{m+1}+m\eps$, so $U(\cvxsignaling) \le \frac{1}{m+1}+m\eps$.

\subsection{Proof of \Cref{lem:rectangle-convex}}
\label{proof:rectangle-convex}
% We will need the following additional lemmas.

\begin{lemma}\label{lemCentroidConvexSetBound}
	Let $S \subseteq [0, 1]^m$ be a convex set with centroid $\mu$. For any coordinate $k \in [m]$, $\max_{p \in S} p_{k} \leq (m + 1) \mu_{k}$. 
\end{lemma}

\begin{proof} 
We use the following classical Minkowski-Radon inequality \citep{stephen_grunbaums_2017}: for any convex body
$K\subseteq\reals^m$ with centroid $g(K)$, % $K-g(K) \subseteq m(g(K)-K)$. Equivalently,
any direction $v\in\reals^m$,
\[
    h_K(v)-\langle g(K),v\rangle ~\le~ m\bigl(\langle g(K),v\rangle-h_K(-v)\bigr),
\]
where $h_K(v):=\max_{x\in K}\langle x,v\rangle$. % is the support function of $K$.
Apply this to $K=S$ and $v=e_k$. Since $g(S)=\mu$, we get
\[
    \max_{p\in S}p_{k}-\mu_{k} ~ \le ~ m\big(\mu_{k} - \min_{p\in S}p_{k}\big).
\]
Because $S\subseteq[0,1]^m$, we have $\min_{p\in S}p_{k}\ge 0$. Therefore %$\max_{p\in S}p_k-\mu_k \le m\mu_k$, and hence
$\max_{p\in S}p_{k}\le (m+1)\mu_{k}$. 
\end{proof}

\begin{lemma}\label{lemRectangleSubdivision}
	% For any $R_1, \dots, R_n \subseteq S \subseteq \rr^m$, where $S$ and $R_i$ are axis-aligned rectangules, %in $m$ dimensions (with $n>0$),
    % there exists a decomposition of $S$ into at most $(1 + 4^m)n$ rectangules, including every $R_i$. % each of $R_1, \dots, R_n$.
    For any non-overlapping axis-aligned rectangles $R_1, \dots, R_n \subseteq S \subseteq \rr^m$, where $S$ is also an axis-aligned rectangle, there exists a decomposition of $S$ into at most $(1 + 4^m)n$ axis-aligned rectangles, including $R_1, \dots, R_n$.
\end{lemma}

\begin{proof}
	Starting with a single part equal to $S$, we iteratively build the decomposition by placing one $R_i$ at a time, subdividing $S$ as we go. With each placement, we preserve the property that each part is a rectangle.
    % For a given step $1 \leq i \leq n$, we first form a new part for $R_i$, taking it out of all previous parts that it overlapped with. This increases the number of parts by 1 and leaves potentially non-rectangular regions behind. Next, for each of the $2^m$ corners $c$ of $R_i$, we slice the region $R_c$ previously containing $c$ by all $m$ axis-aligned orthogonal hyperplanes through $c$, resulting in a partition of $R_c$ into at most $2^m$ rectangular regions. In total, we have added $1 + 2^m \cdot 2^m$ new regions for each $i$. (The very first time, we add fewer regions than this, which compensates for the fact that there is initially one region; thus the bound we get is still $(1 + 4^m)n$ rather than $(1 + 4^m)n + 1$.)
    For a given step $1 \leq i \leq n$, we first form a new part for $R_i$.
    For each previous part that overlaps with $R_i$, remove $R_i$ for the previous part.
    This increases the number of parts by 1 and leaves potentially non-rectangular regions behind.
    Next, for each of the $2^m$ corners $c$ of $R_i$, we slice the region $P_c$ previously containing $c$ by $m$ axis-aligned orthogonal hyperplanes through $c$, resulting in a partition of $P_c$ into at most $2^m$ rectangular regions. In total, we have added $1 + 2^m \cdot 2^m$ new regions for each $i$. We added fewer regions than this at the very first time because there was only one region; thus the bound we finally get is $(1 + 4^m)n$ rather than $(1 + 4^m)n + 1$. 
\end{proof}

\begin{proof}[Proof of \Cref{lem:rectangle-convex}]
	Throughout this proof, $d(\cdot, \cdot)$ denotes Euclidean distance and $\vol(\cdot)$ denotes $m$-dimensional volume in $\rr^m$. For a $K$-convex partitional information policy $\cvxsignaling$ with %valuable
    parts $P_1, P_2, \dots, P_K$ %(meaning the signals sent in these parts are the only ones for which the designer receives positive utility)
    and corresponding centroids $\mu^1, \mu^2, \dots, \mu^K$ where $u(\mu^1) \geq u(\mu^2) \geq \cdots \ge u(\mu^K)$, we can write the designer's utility as the telescoping sum % \taocomment{Abel's summation formula?} \jamie{I think that's much more advanced, I wouldn't mention it.}
	\begin{equation}\label{equUtilGrid}
		U(\cvxsignaling) ~=~ \sum\nolimits_{i = 1}^{K} u(\mu^i) \vol(P_i) ~=~ \sum \nolimits_{i = 1}^K \big(u(\mu^i) - u(\mu^{i + 1})\big) \vol\Big(\bigcup\nolimits_{j = 1}^i P_j \Big)
	\end{equation}
	(where we define $u(\mu^{K + 1}) := 0$).
	
	% Let $\pi^*$ be the optimal $K$-convex partitional signaling scheme with parts $P_1^*, P_2^*, \dots, P_K^*$.
    We define an alternative convex-partitional signaling schemes $\cvxsignaling_1$. Let $A$ be the set of the $\lfloor K / (1 + 4^m) \rfloor$ indices $i \in [K]$ maximizing $u(\mu^i) \vol(P_i)$. Next, we partition $A$ into at most $2^m$ subsets based on which orthant $\mu^i$ is contained in. Let $B \subseteq A$ be the subset maximizing $\sum_{i \in B} u(\mu^i) \vol(P_i)$. Without loss of generality, assume that the orthant corresponding to $B$ is $[0, 1/2]^m$. Write $n := \abs{B}$. Let $\cvxsignaling_1$ be the signaling scheme where the only valuable parts are $P_i$ for $i \in B$; i.e., we ignore all other parts and assume that they give the designer utility zero. Let us call these parts $P^1_i$. Clearly,
	$$U(\cvxsignaling_1) \geq \frac{\lfloor K / (1 + 4^m) \rfloor}{K 2^m} \cdot U(\cvxsignaling).$$
    
	%Thus, letting $x \ge 0$ be the $\argmax$ from the lemma statement, it suffices to show that \
    Fix any $x \ge 0$. We will construct a $K$-axis-aligned rectangular partitional signaling scheme $\pi^{\cc{rect}}$ %using only rectangular parts
    satisfying the following; choosing $x=1$ will prove the lemma. 
	$$U(\pi^{\cc{rect}}) \geq \frac{(2x)^m}{(m + 1)^m (x + 1)^m ((2x + 1)^m - 1)} \cdot U(\cvxsignaling_1).$$
	
	Say that a region $Z \subseteq \rr_{\geq 0}^m$ is \emph{staircase-shaped} if % $Z$ is nonempty and,
    for any $p = (p_{1}, \ldots, p_{m}) \in Z$, $j \in [m]$, $\alpha \in [0, 1]$, we have $(p_{1}, \dots, p_{j - 1}, \alpha p_{j}, p_{j + 1}, \dots, p_{m}) \in Z$.
	We define the \emph{staircase closure} of a set $S \subseteq \rr_{\geq 0}^m$ to be the smallest staircase shape containing $S$. We construct $\pi^{\cc{rect}}$ via the following greedy algorithm. For each $i$ from 1 to $n$, we iteratively add a rectangular part $P^2_i$ with centroid $\mu^i$ as defined below. We will then denote by $Z_i$ the staircase closure of all parts added so far, with $Z_0 := \{\bm 0\}$. On each iteration $i$, we first let $p^i$ be the point on the ray from the origin to $\mu^i$ of farthest distance from the origin such that $p^i \in Z_{i - 1}$. If $d(p^i, \mu^i) \leq x \cdot d(\bm 0, p^i)$, we skip this index, i.e., just define $P^2_i$ to be the empty set. Otherwise, we let $P^2_i$ be the rectangle centered at $\mu^i$ with one corner at $p^i$, which is disjoint from $Z_{i - 1}$ except at the boundary. \Cref{fig2dGreedyAlgorithm} shows %each of
    these two possible cases in $m = 2$ dimensions.
	
	\ipncm{.25}{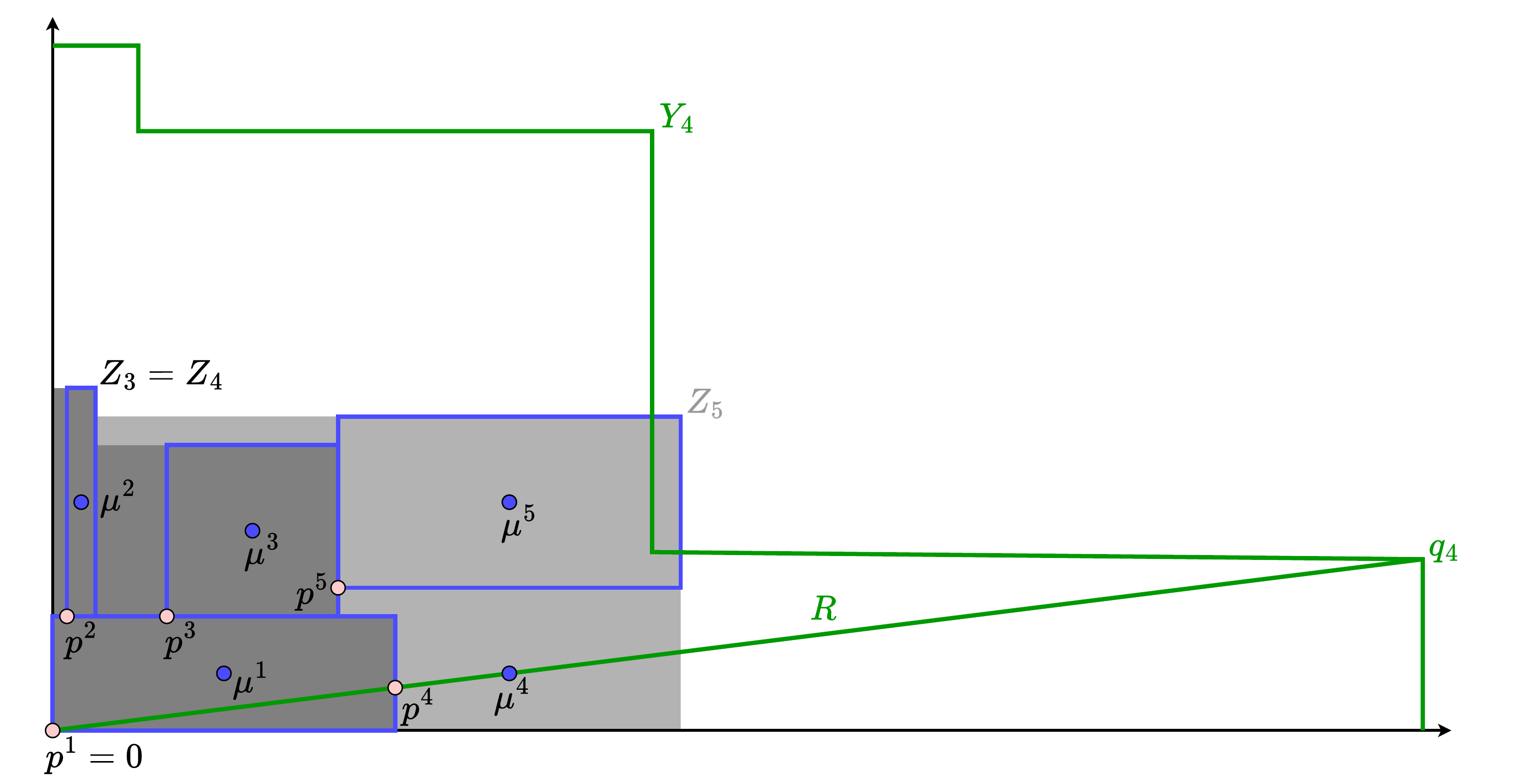}{\label{fig2dGreedyAlgorithm}The first 5 iterations of the greedy algorithm %from the proof of \Cref{thmUniformHigherDimension}
    to construct $\pi^{\cc{rect}}$. % on a two-dimensional instance.
    On iteration $i = 4$, the centroid $\mu_4$ is not far enough away from $Z_3$ (the darker shaded region), so a rectangle is not added, and $Z_4$ remains the same as $Z_3$. On iteration $i = 5$, a new rectangle is added, and $Z_5$ expands to include the lighter shaded region.}
	
	In light of Eqn.~\eqref{equUtilGrid}, it suffices to show that
	$$\frac{\vol\big(\bigcup\nolimits_{j = 1}^i P^2_j\big)}{\vol\big(\bigcup\nolimits_{j = 1}^i P^1_j\big)} \geq \frac{(2x)^m}{(m + 1)^m (x + 1)^m ((2x + 1)^m - 1)}.$$
	This follows from the following two claims:
	\begin{enumerate}[label=(\arabic*)]
		\item\label{itmZBig} $\vol(Z_i) \geq \frac{1}{(m + 1)^m (x + 1)^m} \vol(\bigcup_{j = 1}^i P^1_j)$.
		\item\label{itmZValuable} At least $\frac{(2x)^m}{(2x + 1)^m - 1}$
		of the volume of $Z_i$ is composed of the rectangles $\bigcup_{j=1}^i P_j^2$.
	\end{enumerate}

    To prove \ref{itmZBig}, we note that after processing $j$, the staircase shape $Z_j$ contains the box $\big[\bm 0, \mu^j / (x+1) \big]$
    for the following reason: If $j$ is skipped, then $d(\mu^j, p^j) \le x \cdot d(\bm 0, p^j)$, so $p^j \ge \mu^j / (x+1)$ coordinatewise. Since $Z_{j-1}$ contains $p^j$, it contains $[\bm 0, p_j] \supseteq [\bm 0, \mu^j / (x+1)]$. If $j$ is not skipped, then the newly added rectangle has centroid $\mu^j$, so the staircase closure contains $[\bm 0, \mu^j] \supseteq [\bm 0, \mu^j / (x+1)]$. Therefore, for every $i$, $Z_i \supseteq \bigcup \nolimits_{j \le i} [\bm 0, \mu^j/(x+1)]$. 
    By \Cref{lemRectangleSubdivision}, each convex part $P_j^1 \subseteq [\bm 0, (m+1) \mu^j]$. Hence $\bigcup_{j\le i} P_j^1 \subseteq (m+1)(x+1)Z_i$. The volume then satisfies $\vol(Z_i) \ge \frac{1}{(m+1)^m (x+1)^m} \vol(\bigcup_{j\le i} P_j^1)$.  
	
	%We use induction on the iterations to
    We then prove \ref{itmZValuable} by %Specifically, we show that,
    showing that, on each iteration $i$ where a new rectangle is added, the volume of the new rectangle is at least $\frac{(2x)^m}{(2x + 1)^m - 1}$ fraction of the volume of the added part of the staircase closure $Z_i \setminus Z_{i - 1}$. Let $y > x$ be such that $d(p^i, \mu^i) = y \cdot d(\bm 0, p^i)$.
 %    Since $p^i \in Z_{i - 1}$, which is staircase-shaped, we know that the rectangle from $p^i$ to the origin is disjoint from $Z_i \setminus Z_{i - 1}$.
 %    Thus, the ratio of the new rectangle's volume to the total added volume of $Z_i \setminus Z_{i - 1}$ is at least
	% \begin{align*}
	% 	\frac{\prod_{k = 1}^m 2(\mu^i_k - p^i_k)}{\prod_{k = 1}^m (2\mu^i_k - p^i_k) - \prod_{k = 1}^m p^i_k} &= \frac{\prod_{k = 1}^m 2(\mu^i_k - p^i_k)}{\prod_{k = 1}^m (2(\mu^i_k - p^i_k) + (p^i_k - 0)) - \prod_{k = 1}^m (p^i_k - 0)}\\
	% 	&= \frac{\prod_{k = 1}^m 2y(p^i_k - 0)}{\prod_{k = 1}^m (2y(p^i_k - 0) + (p^i_k - 0)) - \prod_{k = 1}^m (p^i_k - 0)}\\
	% 	&= \frac{(2y)^m}{(2y + 1)^m - 1}\\
	% 	&> \frac{(2x)^m}{(2x + 1)^m - 1}.
	% \end{align*}
	% as desired.
    The new rectangle have corners at $p^i$ and $2\mu^i - p^i = (1+2y)p_i$, so its volume is $\prod_{k=1}^m 2yp_k^i$. 
    Since $p^i \in Z_{i - 1}$, the rectangle $[\bm 0, p^i]$ is disjoint from $Z_i \setminus Z_{i - 1}$, so the volume of $Z_i \setminus Z_{i - 1}$ is at most $\vol([\bm 0, 2\mu^i-p_i]) - \vol([\bm 0, p^i]) = \prod_{k=1}^m (1+2y)p_k^i - \prod_{k=1}^m p_k^i$. 
    Thus, the ratio of the new rectangle's volume to the volume of $Z_i \setminus Z_{i - 1}$ is at least
	\begin{equation*}
		\frac{\prod\nolimits_{k=1}^m 2yp_k^i}{\prod\nolimits_{k=1}^m (1+2y)p_k^i - \prod\nolimits_{k=1}^m p_k^i} ~= ~\frac{(2y)^m}{(2y + 1)^m - 1} ~ > ~ \frac{(2x)^m}{(2x + 1)^m - 1}.
	\end{equation*}
	% as desired.
	
	Thus, we have shown that we can construct $n \leq \lfloor K/(1 + 4^m) \rfloor$ disjoint rectangles that collectively achieve the target fraction of utility. By \Cref{lemRectangleSubdivision}, we can extend this into a $K$-rectangular signaling scheme $\pi^{\cc{rect}}$ covering the entire state space $[0, 1]^m$. %involving only rectangular parts (some with arbitrary centroids).
\end{proof}

\subsection{Proof of \Cref{thm:high-dimension-negative}}
\label{proof:high-dimension-negative}
Consider the example where the interim utility $u$ is $1$ at only one point $p = (\frac{1}{m+1}, \ldots, \frac{1}{m+1}) \in [0, 1]^m$, and $0$ elsewhere.  Let $\Delta_c^m = \{x \in [0, 1]^m : x_1 + \cdots + x_m \le 1 \}$. % be a ``corner'' of the unit cube. %, which is the part below the simplex $\{\bm x \in[0, 1]^m: x_1 + \cdots + x_m = 1\}$.
Let $\pi^*$ send one signal for all states in $\Delta_c^m$, and another signal for $[0, 1]^m \setminus \Delta_c^m$. Under uniform prior, the posterior mean induced by the first signal is equal to the centroid of $\Delta_c^m$, which equals $p$: 
\begin{equation*}
    \mu(\Delta_c^m) ~ = ~ \frac{1}{\vol(\Delta_c^m)} \int_{\Delta_c^m} x \;\dd x ~ = ~ \Big(\frac{1}{m+1}, \ldots, \frac{1}{m+1} \Big) ~ = ~ p\,. 
\end{equation*}
The total probability of the first signal is equal to the volume $\vol(\Delta_c^m) ~ = ~ \frac{1}{m!} |\mathrm{det}(I)| ~ = ~ \frac{1}{m!}$. 
So, the expected utility of $\pi^*$ is $U(\pi^*) = \frac{1}{m!} \cdot u(p) = \frac{1}{m!}$.

For any axis-aligned rectangular signaling scheme $\pi^{\mathrm{rect}}$ to obtain positive utility, at least one rectangle should have its centroid at $p$. The largest rectangle inside $[0, 1]^m$ with centroid at $p$ is $\big[0, \frac{2}{m+1}\big]^m$, which has volume $\vol\big(\big[0, \tfrac{2}{m+1}\big]^m\big) = \big( \tfrac{2}{m+1} \big)^m$. 
Thus, the ratio between the expected utilities of $\pi^{\mathrm{rect}}$ and $\pi^*$ is at most
\begin{align*}
    \frac{U(\pi^{\mathrm{rect}})}{U(\pi^*)} ~ \le ~ \frac{\big( \tfrac{2}{m+1} \big)^m}{\tfrac{1}{m!}} \stackrel{\text{Stirling's approximation}}{\le} \big( \tfrac{2}{m+1} \big)^m \sqrt{2\pi m} \big(\tfrac{m}{e}\big)^m e^{\frac{1}{12m}} ~ \le ~ \big(\tfrac{2}{e} \big)^m O(\sqrt m), 
\end{align*}
which implies $\PoE^{\cc{rect}}(\instance, \numSignals) \le \big(\frac{2}{e} \big)^m O(\sqrt m)$ and proves the theorem. 

% \newpage

\bibliography{bibfile}

% \newpage

% \newpage

% \input{note}

\end{document}